\documentclass[preprint,12pt]{elsarticle}

\usepackage{amssymb}
\usepackage{amsmath}
\usepackage{amsthm}
\usepackage{algorithm}
\usepackage{algpseudocode}
\usepackage{float}
\numberwithin{equation}{section}
\newtheorem{prop}{Proposition}[section]
\numberwithin{algorithm}{section}
\newtheorem{lemma}{Lemma}[section]
\usepackage{graphicx} 
\usepackage{caption} 
\usepackage{subcaption} 
\usepackage{threeparttable} 
\usepackage{hhline}
\usepackage{diagbox}
\usepackage{booktabs}
\usepackage{makecell}

\journal{Journal of Differential Equations}

\begin{document}

\begin{frontmatter}



\title{Hyperbolic mean curvature flow computed by physics-informed neural networks}


\author[label1,label3]{Shuangshuang Duan}\ead{ssduan@zjnu.edu.cn}\author[label2]{Chunlei He}
\ead{chunlei@zjnu.edu.cn} 
\author[label2]{Shoujun Huang\corref{cor1}}\ead{sjhuang@zjnu.edu.cn}
\author[label2]{Dexing Kong}
\ead{dkong@zjnu.edu.cn} 

\cortext[cor1]{Corresponding author} 

\affiliation[label1]{organization={School of Mathematical Sciences},
	addressline={Zhejiang Normal University}, 
	city={Jinhua},
	postcode={321004},
	country={China}} 

\affiliation[label2]{organization={College of Mathematical Medicine},
            addressline={Zhejiang Normal University}, 
            city={Jinhua},
            postcode={321004},
            country={China}}           
\affiliation[label3]{organization={College of  science},
	addressline={Guangdong University of Petrochemical Technology}, 
	city={Maoming},
	postcode={525000},
	country={China}}  
	
%% Abstract
\begin{abstract}
In this paper, we explore the evolution of plane curves and surfaces governed by the hyperbolic mean curvature flow. We propose a mesh-free approach based on the physics-informed neural networks (PINNs), which eliminates the need for discretization and meshing of computational domains, and is efficient in solving partial differential equations involving high dimensions. To the best of our knowledge, this is the first result on the numerical analysis by employing the PINNs for the hyperbolic geometric evolution equations in the literature. The effectiveness of this method is demonstrated through several numerical simulations by selecting diverse initial curves and surfaces, as well as both constant and non-constant initial velocities.

\end{abstract}



%% Keywords

\end{frontmatter}



\section{Introduction}
\label{sec1}
In this paper, we are concerned with the following hyperbolic mean curvature flow (HMCF)
\begin{equation}\label{eq11}
	\frac{\partial^2 F}{\partial t^2} + \beta \frac{\partial F}{\partial t} = H\overrightarrow{N} -\nabla e ,
\end{equation}
where $F:M\times[0,T)\rightarrow \mathbb{R}^{n+1}$ is an immersed $n$-dimensional hypersurface $M$ in the Euclidean space, $H$ is the mean curvature of hypersurface $M$, and $\overrightarrow{N}$ is the  inner unit normal vector. Throughout this paper, we shall consider two cases for the hypersurfaces. Firstly, we take the choice $n=1$, and the equation (\ref{eq11}) can be written as
\begin{equation}\label{eq12}
	\frac{\partial^2 \gamma}{\partial t^2} + \beta \frac{\partial \gamma}{\partial t} = H\overrightarrow{N} - (\frac{\partial^2 \gamma}{\partial s \partial t}, \frac{\partial \gamma}{\partial t})\overrightarrow{T},
\end{equation}
which describes the motion of a closed curve $\Gamma(t)$ in the plane. Here,  $\nabla e = (\frac{\partial^2 \gamma}{\partial s \partial t}, \frac{\partial \gamma}{\partial t})\overrightarrow{T}$, and the symbol $(\cdot, \cdot)$ denotes the scalar product in $\mathbb{R}^2$, $\overrightarrow{T}$ is the unit tangent vector and  $s$ is the arclength parameter.

Second, we are concerned with the case of two space dimensions, and we have the following equation
\begin{equation}\label{eq13}
	\frac{\partial^2 X}{\partial t^2} + \beta \frac{\partial X}{\partial t} = H\overrightarrow{N} - \frac{1}{2}\nabla_{S(t)}(|X_t|^2),
\end{equation}
which governs the motion for a family of closed oriented surfaces $S(t)$. Note that $\nabla e = \frac{1}{2}\nabla_{S(t)}(|X_t|^2)$ by LeFloch and Smoczyk \cite{Le1}, which involves the kinetic energy and the internal energy associated with the surface.

The HMCF has been introduced by some authors (cf. \cite{Le2, He2, Kong2, Wang2, Wang3, McCoy, Ginder} and references therein). Gurtin and Podio-Guidugli \cite{Gu} developed a hyperbolic theory for the evolution of plane curves, which was used to describe the melting-freezing waves observed at the surface of  crystals. Yau \cite{Yau} proposed the following equation related to a vibrating membrane or the motion of a surface
\begin{equation}\label{eq14}
	\frac{\partial^2 F}{\partial t^2} = H\overrightarrow{N},
\end{equation}
where $H$ is the mean curvature, $\overrightarrow{N}$ is the unit inner normal vector of the surface. He et al. \cite{He1} showed that the flow (\ref{eq14}) admits a unique short-time smooth solution and possesses the nonlinear stability defined on the Euclidean space with dimension larger than 4. For the normal HMCF
\begin{equation*}
	\frac{\partial^2 F}{\partial t^2} = H\overrightarrow{N} - (\frac{\partial^2 F}{\partial s \partial t}, \frac{\partial F}{\partial t})\overrightarrow{T},
\end{equation*}
Kong, Liu and Wang \cite{Kong1} proved  that if the initial curve is strictly convex closed curve, then there exists a class of initial velocities such that the solution of the above initial value problem exists only at a finite time interval $[0,T_{max})$, and when $t$ goes to $T_{max}$, either the solution converge to a point or shock and other propagating discontinuities are generated. Moreover, Wang et al \cite{Wang1} analyzed the asymptotic behavior of the dissipative hyperbolic curvature flow (\ref{eq12}), and showed that if the minimum of initial velocity is nonnegative, the flow will converge  to either a point or a limit curve; while if the maximum of initial velocity is positive, the flow will expand firstly and then converge  to either a point or a limit curve.

It is well-known that solutions of these HMCFs often blow up in some finite time, even though the initial conditions may be smooth. In such a situation, it is difficult to obtain analytical solutions. In practice, it is more interesting to find out the approximate solutions. Rotstein et al. \cite{Rotstein} developed a hyperbolic crystalline algorithm for the motion of closed convex polygonal curves, which generalized the standard crystalline algorithm. Kusumasari \cite{Kusu} considered the interface motion with an obstacle according to the HMCF, by employing the Hyperbolic Merriman-Bence-Osher (HMBO) algorithm as an approximation method. Deckelnick et al. \cite{Deckel1} introduced a semidiscrete finite difference method for the approximation of HMCF in the plane and proved the error bounds for natural discrete norms. Meanwhile,  Deckelnick et al. \cite{Deckel2}  addressed the numerical approximation of two variants of HMCF of surfaces in $\mathbb{R}^{3}$,  proposing both a finite element method, as well as a finite difference scheme in the case of axially symmetric surfaces. Monika \cite{Monika} analyzed the behavior of the solutions of the HMCF by means of a semi-discrete finite-volume scheme.

However, instead of using the traditional numerical schemes, alternative approaches have emerged recently using deep learning algorithms to solve forward and inverse problems for partial differential equations (PDEs) \cite{Blech}, such as the multilayer perceptron \cite{Lagaris}, the deep galerkin method \cite{Sirignano}, and the deep Ritz method \cite{E}. Specifically, a physics-informed neural networks (PINNs) was proposed by Raissi et al. \cite{Raissi}. The essential idea of PINNs is to embed the physical information described by PDEs into the neural network through the construction of a loss function. By minimizing this loss function, the network parameters are optimized to approximate the solution of the PDEs. Compared with traditional numerical approximation methods like the finite element method \cite{Reddy}, the finite difference method \cite{Liszka}, and the finite volume difference method \cite{Levequer}, where the governing PDEs are eventually discretized over the computational domains, the major advantage of PINNs is providing a mesh-free algorithm as the differential operators in the governing PDEs are approximated by automatic differentiation \cite{Baydin}.

Recently, the application of the novel methodology of PINNs has been employed to solve a wide variety of PDEs, including Navier-Stokes equations \cite{Eivazi, Zhang, Arthurs, Zhu, Cheng}, Schr$\ddot{o}$dinger equations \cite{Song, Tamara}, advection-diffusion equations \cite{Cai, Hu, Dwivedi, He3}, and conservation laws \cite{Jagtap, Mao, Kissas}. However,  the performance of this algorithm has not yet been fully investigated in the literature for solving HMCF. In this work, we are interested in the evolution of closed plane curves and surfaces governed by the HMCF. Conventional numerical schemes for this geometric flow often suffer from severe mesh distortion, necessitating frequent and computationally expensive re-meshing, particularly during long-term simulations. Furthermore, standard explicit time-stepping methods are constrained by the strict Courant–Friedrichs–Lewy condition, which enforces strictly small time steps and consequently incurs a substantial computational burden.

To address these challenges, we propose a PINNs solver to approximate the solutions to Eqs.(\ref{eq12}) and (\ref{eq13}). Our approach adopts an mesh-free paradigm, employing Latin Hypercube Sampling strategy \cite{Stein} to collocate points randomly across the spatiotemporal domain. This strategy effectively obviates the need for domain discretization and complex grid generation. Moreover, leveraging the automatic differentiation framework allows us to compute the mean curvature exactly from the parameterized coordinate map, thereby avoiding the approximation errors inherent in discrete differential geometry operators.

To enforce the correct topology for closed geometries, we implement a soft constraint mechanism by incorporating specialized penalty terms into the loss function. Specifically, the introduction of periodic boundary conditions—and for surfaces, pole constraints—enables the neural network to autonomously learn solutions that remain continuous and smooth at boundaries and coordinate singularities.

Training PINNs for hyperbolic geometric flows is notoriously challenging due to complex loss landscapes and convergence difficulties. To mitigate these issues, we employ a hybrid optimization strategy coupled with a dynamic loss weighting scheme. This involves an initial training phase using the robust Adam optimizer, followed by a fine-tuning stage using the high-precision L-BFGS algorithm. Concurrently, we adopt a curriculum learning approach where larger penalty weights are initially assigned to the initial and boundary conditions. This guides the network to prioritize learning the fundamental geometric topology before focusing on minimizing the PDEs residual, significantly enhancing both convergence speed and solution accuracy.
To the best of our knowledge, this work presents the first numerical analysis employing PINNs specifically for hyperbolic geometric evolution equations.

This paper is organized as follows. In the next section, we describe the parametric description of the HMCF, and the main analytical properties of the evolution equations are summarized. Section 3 presents the proposed methodology in detail. In Section 4, we perform various numerical experiments on closed plane curves and surfaces. Finally, some concluding remarks  are given in Section 5.

\section{Properties enjoyed by the flow}
\label{sec2}
\subsection{Hyperbolic mean curvature flows for plane curves}
In what follows we shall employ a parametric description of the curves $\Gamma(t)$ evolving by (\ref{eq12}), i.e. $\Gamma(t)=\gamma(S^1,t)$ for some mapping $\gamma:S^1 \times [0,T) \rightarrow \mathbb{R}^2$, where $S^1$ is the periodic interval $[0,2\pi]$. We express the tangent vector $\overrightarrow{T}$, and the normal vector $\overrightarrow{N}$ by means of the parametrization and using the arclength variable $s$, for which $\frac{\partial}{\partial s}=\frac{1}{|\partial_u \gamma|}\frac{1}{\partial u}$, 
\begin{equation}\label{eq21}
	\overrightarrow{T}=\frac{\partial \gamma}{\partial s} , \quad \overrightarrow{N}=\frac{{\partial \gamma}^{\bot}}{\partial s},
\end{equation}
where the symbol ${\bot}$ denotes the anticlockwise rotation through $\frac{\pi}{2}$. By Frenet's formula
\begin{equation*}
	\frac{\partial \overrightarrow{T}}{\partial s}=H\overrightarrow{N}, \quad \frac{\partial \overrightarrow{N}}{\partial s}=-H\overrightarrow{T},
\end{equation*}
we have
\begin{equation*}
	H\overrightarrow{N}=\frac{1}{|\partial_u\gamma|}
	\partial_u(\frac{\partial_u\gamma}{|\partial_u\gamma|}).
\end{equation*}
Furthermore,
\begin{equation*}
	(\frac{\partial^2 \gamma}{\partial s \partial t}, \frac{\partial \gamma}{\partial t})\overrightarrow{T}=\frac{1}{|\partial_u\gamma|^2}(\frac{\partial^2 \gamma}{\partial u \partial t}, \frac{\partial \gamma}{\partial t})\partial_u \gamma.
\end{equation*}
By (\ref{eq12}), we obtain the hyperbolic geometric evolution equations for plane curves
\begin{equation}\label{eq22}
	\begin{aligned}
		\frac{\partial^2 \gamma}{\partial t^2} + \beta \frac{\partial \gamma}{\partial t} &= \frac{1}{|\partial_u\gamma|}\partial_u(\frac{\partial_u\gamma}{|\partial_u\gamma|})
		-\frac{1}{|\partial_u\gamma|^2}(\frac{\partial^2 \gamma}{\partial u \partial t}, \frac{\partial \gamma}{\partial t})\partial_u \gamma & in \; S^1\times[0,T),\\
		\gamma(u,0)&=\gamma_0 &in \;\quad S^1,\\
		\gamma_t(u,0)&=\gamma_1 &in \;\quad S^1.
	\end{aligned}
\end{equation}
Here, $\gamma_0$ is the initial closed curve, and $\gamma_1$ is the initial velocity.
\begin{prop}\label{pro1}
	Assume the initial velocity is normal. Then, the hyperbolic mean curvature flow (\ref{eq12}) is normal.
\end{prop}
\begin{proof}
	By using (\ref{eq12}) and (\ref{eq21}), we deduce that
	\begin{eqnarray*}
		\frac{\partial}{\partial t}(\frac{\partial \gamma}{\partial t},\frac{\partial \gamma}{\partial u})&=&(\frac{\partial^2 \gamma}{\partial t^2},\frac{\partial \gamma}{\partial u})+(\frac{\partial \gamma}{\partial t},\frac{\partial^2 \gamma}{\partial u \partial t})\\
		&=&(H\overrightarrow{N} - (\frac{\partial^2 \gamma}{\partial s \partial t}, \frac{\partial \gamma}{\partial t})\overrightarrow{T}-\beta \frac{\partial \gamma}{\partial t},\frac{\partial \gamma}{\partial u})+(\frac{\partial \gamma}{\partial t},\frac{\partial^2 \gamma}{\partial u \partial t})\\
		&=&|\frac{\partial \gamma}{\partial u}|(H\overrightarrow{N} - (\frac{\partial^2 \gamma}{\partial s \partial t}, \frac{\partial \gamma}{\partial t})\overrightarrow{T}, \overrightarrow{T})+(\frac{\partial \gamma}{\partial t},\frac{\partial^2 \gamma}{\partial u \partial t})-\beta (\frac{\partial \gamma}{\partial t},\frac{\partial \gamma}{\partial u})\\
		&=&-(\frac{\partial^2 \gamma}{\partial u \partial t},\frac{\partial \gamma}{\partial t})+(\frac{\partial \gamma}{\partial t},\frac{\partial^2 \gamma}{\partial u \partial t})-\beta (\frac{\partial \gamma}{\partial t},\frac{\partial \gamma}{\partial u})\\
		&=&-\beta (\frac{\partial \gamma}{\partial t},\frac{\partial \gamma}{\partial u}).
	\end{eqnarray*}
	In view of the third equation in (\ref{eq22}), we have $(\frac{\partial \gamma}{\partial t},\frac{\partial \gamma}{\partial u})=0$, which implies (\ref{eq12}) is normal. The proof is completed. 
\end{proof}

\subsection{Hyperbolic mean curvature flow for surfaces}
Let us consider a parametric description of the evolving surfaces, i.e. $S(t)=X(M,t)$ for some mapping $X:M \times [0,T) \rightarrow \mathbb{R}^3$. We assume that $\{\frac{\partial X}{\partial u_1}, \frac{\partial X}{\partial u_2}\}$ is a basis of the tangent space $T_XS(t)$. Define the metric on $S(t)$ by
\begin{equation*}
	g_{ij}=\frac{\partial X}{\partial u_1}\cdot\frac{\partial X}{\partial u_2}, \quad i,j=1,2,
\end{equation*}
and let $g^{ij}$ be the components of the inverse matrix of $(g_{ij})$. We then have the following formula for the tangential gradient of a function $f$ and the mean curvature vector \cite{Deckel3},
\begin{equation*}
	\bigtriangledown_{S(t)}f=\sum^2_{i,j=1}g^{i,j}\frac{\partial f}{\partial u_i}\cdot\frac{\partial X}{\partial u_j},
\end{equation*}
\begin{equation*}
	H\overrightarrow{N}=\bigtriangleup_{S(t)}X=\frac{1}{\sqrt g}\sum^2_{i,j=1}\frac{\partial}{\partial u_i}(g^{ij}\sqrt g \frac{\partial X}{\partial u_j}),
\end{equation*}
where $\bigtriangleup_{S(t)}X$ is the Laplace-Beltrami operator, $g=\det(g_{ij})$.

Recalling (\ref{eq13}), we shall consider the following equations for surfaces
\begin{equation}\label{eq23}
	\begin{aligned}
		\frac{\partial^2 X}{\partial t^2} + \beta \frac{\partial X}{\partial t} &=
		\frac{1}{\sqrt g}\sum^2_{i,j=1}\frac{\partial}{\partial u_i}(g^{ij}\sqrt g \frac{\partial X}{\partial u_j})+\sum^2_{i,j=1}g^{i,j}\frac{\partial X}{\partial t}\frac{\partial X_t}{\partial u_i}\cdot\frac{\partial X}{\partial u_j}& on \; M\times[0,T),\\
		X(u_1, u_2, 0)&=X_0 &on \;\quad M,\\
		X_t(u_1, u_2, 0)&=X_1 &on \;\quad M,
	\end{aligned}
\end{equation}
where $X_0$ is the initial closed surface, and $X_1$ is the initial velocity.
\begin{prop}\label{pro2}
If the initial velocity is normal, then the hyperbolic mean curvature flow (\ref{eq13}) is also normal.
\end{prop}
\begin{proof}
	By observing that $\{\frac{\partial X}{\partial u_1}, \frac{\partial X}{\partial u_2}\}$ is a basis of the tangent space $T_XS(t)$, we can conclude that (\ref{eq13}) is normal if we can show
	\begin{equation}\label{eq24}
		\frac{\partial X}{\partial t} \cdot \frac{\partial X}{\partial u_i}=0,  \quad\quad for \; i=1,2\;\quad and \quad\; 0\leq t \leq T.
	\end{equation}
	By (\ref{eq13}), we have
	\begin{eqnarray*}
		\frac{\partial}{\partial t}(\frac{\partial X}{\partial t} \cdot \frac{\partial X}{\partial u_i})
		&=& \frac{\partial^2 X}{\partial t^2}\cdot\frac{\partial X}{\partial u_i}+\frac{\partial X}{\partial t}\cdot\frac{\partial^2 X}{\partial u_it}\\
		&=&(H\overrightarrow{N} - \frac{1}{2}\nabla_{S(t)}(|X_t|^2)-\beta \frac{\partial X}{\partial t})\cdot\frac{\partial X}{\partial u_i}+\frac{\partial X}{\partial t}\cdot\frac{\partial^2 X}{\partial u_it}\\
		&=&- \frac{1}{2}\nabla_{S(t)}(|X_t|^2)\cdot\frac{\partial X}{\partial u_i}+\frac{\partial X}{\partial t}\cdot\frac{\partial^2 X}{\partial u_it}-\beta \frac{\partial X}{\partial t}\cdot\frac{\partial X}{\partial u_i}\\
		&=&-\beta \frac{\partial X}{\partial t}\cdot\frac{\partial X}{\partial u_i},
	\end{eqnarray*}
	since
	\begin{eqnarray*}
		\frac{1}{2}\nabla_{S(t)}(|X_t|^2)\cdot\frac{\partial X}{\partial u_i}
		&=&\frac{1}{2}\sum^2_{k,l}g^{kl}\frac{\partial (|X_t|^2)}{\partial u_k}\cdot\frac{\partial X}{\partial u_l}\cdot\frac{\partial X}{\partial u_i}
		=\sum^2_{k,l}g^{kl}\cdot\frac{\partial X}{\partial t}\cdot\frac{\partial^2 X}{\partial tu_k}\cdot g_{li}\\
		&=&\sum^2_{k,l}\delta_{ki}(\frac{\partial X}{\partial t}\cdot\frac{\partial^2 X}{\partial tu_k})
		=\frac{\partial X}{\partial t}\cdot\frac{\partial^2 X}{\partial u_it},
	\end{eqnarray*}
	where $\delta_{ki}$ is the Kronecker symbol.
	With the help of the third equation of (\ref{eq23}), we find that (\ref{eq24}) holds. The proof is completed. 
\end{proof}

\section{The Algorithm of PINNs}
\label{sec3}
In this section, we formulate the method of PINNs  for the HMCF. Let $\mathcal{N}: \mathbb{R}^{d_{in}}\rightarrow \mathbb{R}^{d_{out}}$ be a full connected neural network, for an space-time coordinate vector $\hat{x}=(x_1, x_2, \ldots, x_{d-1}, t)$, where $d^{in}$ and $d^{out}$ represent the dimensionality  of the input and output vector, respectively, we have
\begin{equation*}
	\mathcal{N}(\hat{x}, \theta)= h_L\circ h_{L-1}\circ h_{L-2}\circ \cdots \circ h_1(\hat{x}),\quad for \quad \hat{x}\in \mathbb{R}^{d},
\end{equation*}
where $h_{\ell} : \mathbb{R}^{M_{\ell-1}} \rightarrow \mathbb{R}^{M_{\ell}}, M_{\ell}\in \mathbb{N}^{+}, {\ell}=0, 1, \ldots, L$ is a nonlinear function, and can be represented as
\begin{equation*}
	h_{\ell}(x_{\ell}) := \phi(W_{\ell}x_{\ell}+b_{\ell}).
\end{equation*}
Here, $M_{\ell}$ is the number of neurons of the $\ell-th$ layer, $L$ is called the depth, the function $\phi$ denotes a predefined activation function, and $W_{\ell} \in \mathbb{R}^{M_{\ell}}\times\mathbb{R}^{M_{\ell-1}}$ and $b_{\ell} \in \mathbb{R}^{M_{\ell}}$ are the weight and bias parameters, that are assembled in the vector $\theta$.

To enable the training of our neural network, we reformulate the PDEs discussed in Section  \ref{sec2} into residual forms. In the case of closed plane curves, the periodic boundary conditions hold, i.e.,
\begin{equation}\label{eq31}
	\gamma(0,t) = \gamma(2\pi,t), \; \gamma_u(0,t) = \gamma_u(2\pi,t).
\end{equation}
Then, we rewrite the HMCF into the residual form as
\begin{equation}\label{eq32}
	\mathcal{L}_{f_c}^c(\theta_c) = \frac{1}{N_{f_c}}\sum_{i=1}^{N_{f_c}}|f(u_{f_c}^i,t_{f_c}^i)|^2,
\end{equation}
where  $f_c(u,t)$ denotes the left-hand-side of the first equation in (\ref{eq22}), i.e.,
\begin{equation}\label{eq33}
	f_c := \frac{\partial^2 \gamma}{\partial t^2} + \beta \frac{\partial \gamma}{\partial t} - \frac{1}{|\partial_u\gamma|}\partial_u(\frac{\partial_u\gamma}{|\partial_u\gamma|})
	+\frac{1}{|\partial_u\gamma|^2}(\frac{\partial^2 \gamma}{\partial u \partial t}, \frac{\partial \gamma}{\partial t})\partial_u \gamma.
\end{equation}
Similarly, the residual corresponding to the initial and boundary conditions is given by 
\begin{equation}\label{eq34}
	\mathcal{L}_0^c(\theta_c) = \frac{1}{N_0}\sum_{i=1}^{N_0}(|\gamma^i(u_0^i,0)-\gamma_0^i|^2 + |\gamma_t^i(u_0^i,0)-\gamma_1^i|^2),
\end{equation}
and
\begin{equation}\label{eq35}
	\mathcal{L}_b^c(\theta_c) = \frac{1}{N_b}\sum_{i=1}^{N_b}(|\gamma^i(0,t_b^i)-\gamma^i(2\pi,t_b^i)|^2 + |\gamma_u^i(0,t_b^i)-\gamma_u^i(2\pi,t_b^i)|^2).
\end{equation}
Here, $\{u_{f_c}^i,t_{f_c}^i\}_{i=1}^{N_{f_c}}$ represents the collocation points on $f_c(u,t)$, $\{u_0^i,\gamma_0^i,\gamma_1^i\}_{i=1}^{N_0}$ denotes the initial data, and $\{t_b^i\}_{i=1}^{N_b}$ corresponds to the collocation points on the boundary.

Regarding surfaces such as spheres and ellipsoids, we will discuss them in Section \ref{sec4} and assume the periodic boundary conditions also hold, i.e.,
\begin{equation}\label{eq36}
	X(u_1,0,t) = X(u_1,2\pi,t), \; X_{u_2}(u_1,0,t) = X_{u_2}(u_1,2\pi,t).
\end{equation}
To enforce closure and smoothness at the north and south poles of the surface, we need to impose additional boundary conditions, i.e.,
\begin{equation}\label{eq37}
	X(0,u_2,t) = r(t), \; X(\pi,u_2,t) = -r(t).
\end{equation}
and
\begin{equation}\label{eq38}
	X_{u_1}(0,u_2,t) = 0, \; X_{u_1}(\pi,u_2,t) = 0, \; X_{u_2}(0,u_2,t) = 0, \; X_{u_2}(\pi,u_2,t) = 0,
\end{equation}
where $r(t)$ denotes the position vector of the surface at the north and south poles at time $t$. Eq. (\ref{eq37}) is designed to ensure convergence of the surface to a single point at the north and south poles, while Eq. (\ref{eq38}) guarantees the smoothness of the surface in the vicinity of the poles. Then, those residuals take the form
\begin{equation}\label{eq39}
	\mathcal{L}_{f_s}^s(\theta_s) = \frac{1}{N_{f_s}}\sum_{i=1}^{N_{f_s}}|f(u_{1f_s}^i,u_{2f_s}^i,t_{f_s}^i)|^2,
\end{equation}
\begin{equation}\label{eq310}
	\mathcal{L}_0^s(\theta_s) = \frac{1}{N_0}\sum_{i=1}^{N_0}(|X^i(u_{10}^i,u_{20}^i,0)-X_0^i|^2 + |X_t^i(u_{10}^i,u_{20}^i,0)-X_1^i|^2),
\end{equation}
\begin{equation}\label{eq311}
	\mathcal{L}_b^s(\theta_s) = \frac{1}{N_b}\sum_{i=1}^{N_b}(|X^i(u_1^i,0,t_b^i)-X^i(u_1^i,2\pi,t_b^i)|^2 + |X_{u_2}^i(u_1^i,0,t_b^i)-X_{u_2}^i(u_1^i,2\pi,t_b^i)|^2),
\end{equation}
and
\begin{equation}\label{eq312}
	\begin{split}
		\mathcal{L}_p^s(\theta_s)
		&= \frac{1}{N_p}\sum_{i=1}^{N_p}(|X^i(0,u_2^i,t_b^i)-r^i(t_b^i)|^2 + |X^i(\pi,u_2^i,t_b^i)+r^i(t_b^i)|^2 \\
		&+ |X_{u_1}^i(0,u_2^i,t_b^i)|^2 + |X_{u_1}^i(\pi,u_2^i,t_b^i)|^2 + |X_{u_2}^i(0,u_2^i,t_b^i)|^2 + |X_{u_2}^i(\pi,u_2^i,t_b^i)|^2),
	\end{split}
\end{equation}
where $f_s(u,t)$ represents the left-hand-side of the first equation in (\ref{eq23}), i.e.,
\begin{equation}\label{eq313}
	f_s := \frac{\partial^2 X}{\partial t^2} + \beta \frac{\partial X}{\partial t} -
	\frac{1}{\sqrt g}\sum^2_{i,j=1}\frac{\partial}{\partial u_i}(g^{ij}\sqrt g \frac{\partial X}{\partial u_j})-\sum^2_{i,j=1}g^{i,j}\frac{\partial X}{\partial t}\frac{\partial X_t}{\partial u_i}\cdot\frac{\partial X}{\partial u_j}.
\end{equation}
Here, $\{u_{1f_s}^i,u_{2f_s}^i,t_{f_s}^i\}_{i=1}^{N_{f_s}}$ denotes the collocation points on $f_s(u_1,u_2,t)$, $\{u_{10}^i,u_{20}^i,X_0^i,X_1^i\}_{i=1}^{N_0}$ represents the initial data, $\{u_1^i,t_b^i\}_{i=1}^{N_b}$ corresponds to the collocation points on the boundary, and $\{u_2^i,t_b^i,r_b^i\}_{i=1}^{N_p}$ is the set of collocation points on north and south poles.

Combining all the losses for the HMCF for plane curves and surfaces, we obtain the loss functions
\begin{equation}\label{eq314}
	\mathcal{L}_{curve}(\theta_c) = \omega_{f_c}^c\mathcal{L}_{f_c}^c(\theta_c) + \omega_0\mathcal{L}_0^c(\theta_c)+ \omega_b\mathcal{L}_b^c(\theta_c),
\end{equation}
\begin{equation}\label{eq315}
	\mathcal{L}_{surface}(\theta_s) = \omega_{f_s}^s\mathcal{L}_{f_s}^s(\theta_s) + \omega_0^s\mathcal{L}_0^s(\theta_s) + \omega_b^s\mathcal{L}_b^s(\theta_s) + \omega_p^s\mathcal{L}_p^s(\theta_s),
\end{equation}
where $\omega_{f_c}^c$, $\omega_0^c$, and $\omega_b^c$, are the corresponding weight parameters for the governing PDEs, initial conditions, and boundary conditions. In the context of the surface, we utilize $\omega_{f_s}^s$, $\omega_0^s$, $\omega_b^s$, and $\omega_p^s$ to assign weights to the loss components associated with the PDEs, initial conditions, boundary conditions, and boundary conditions at north and south poles.

The optimal neural network parameters $\theta_c^*$ are found by minimizing the loss function $\mathcal{L}_{curve}(\theta_c)$ :
\begin{equation*}
	\theta_c^* = arg\min \limits_{\theta_c} \mathcal{L}_{curve}(\theta_c).
\end{equation*}
The resulting solutions provided by networks satisfy the HMCF described in equation (\ref{eq22}) with initial condition and boundary condition. Similarly, for the loss $\mathcal{L}_{surface}(\theta_s)$, we have
\begin{equation*}
	\theta_s^* = arg\min \limits_{\theta_s} \mathcal{L}_{surface}(\theta_s),
\end{equation*}
and the HMCF for surface in (\ref{eq23}), along with the appropriate initial and boundary conditions, as well as the boundary conditions at the north and south poles, can be effectively solved.

For the plane curve evolution problem, we employ a multi-stage hybrid optimization strategy designed to leverage the distinct advantages of various optimizers and learning rates. This comprehensive training scheme is presented in Algorithm \ref{alg:1}.

\begin{algorithm}[h!]
	\caption{Training Scheme for Plane Curve Evolution}
	\label{alg:1} 
	\begin{algorithmic}
		\State \textbf{Input:} Neural network $\mathcal{N}$ with parameters $\theta_c$; training data; Adam steps $N_{\text{Adam}}^{(k)}$; Adam learning rate $\eta_{\text{Adam}}^{(k)}$; L-BFGS steps $N_{\text{LBFGS}}$.
		\State \textbf{Output:} Optimized neural network parameters $\theta_c^*$.
		\State
		\State \textbf{Phase 1 Adam Optimization for Stage 1}
		\For{$i = 1$ to $N_{\text{Adam}}^{(1)}$}
		\If{$i < 2,000$}
		\State Set weights: $w_{f_c}^c \gets 1, w_0^c \gets 100, w_b^c \gets 100$;
		\Else
		\State Set weights: $w_{f_c}^c \gets 1, w_0^c \gets 1, w_b^c \gets 1$;
		\State Compute total loss: $\mathcal{L}_{\text{curve}}(\theta_c) = w_{f_c}^c \mathcal{L}_{f_c}^c(\theta_c) + w_0^c \mathcal{L}_0^c(\theta_c) + w_b^c \mathcal{L}_b^c(\theta_c)$;
		\State Update network parameters $\theta_c$ using Adam Optimizer with learning rate $\eta_{\text{Adam}}^{(1)}$;
		\EndIf
		\EndFor
		\State
		\State \textbf{Phase 2 Adam Optimization for Stage 2}
		\For{$i = 1$ to $N_{\text{Adam}}^{(2)}$}
		\State Set weights: $w_{f_c}^c \gets 1, w_0^c \gets 1, w_b^c \gets 1$;
		\State Compute total loss: $\mathcal{L}_{\text{curve}}(\theta_c) = w_{f_c}^c \mathcal{L}_{f_c}^c(\theta_c) + w_0^c \mathcal{L}_0^c(\theta_c) + w_b^c \mathcal{L}_b^c(\theta_c)$;
		\State Update network parameters $\theta_c$ using Adam Optimizer with learning rate $\eta_{\text{Adam}}^{(2)}$;
		\EndFor
		\State
		\State \textbf{Phase 3 L-BFGS Fine-Tuning}
		\For{$i = 1$ to $N_{\text{LBFGS}}$}
		\State Set weights: $w_{f_c}^c \gets 1, w_0^c \gets 1, w_b^c \gets 1$;
		\State Compute total loss: $\mathcal{L}_{\text{curve}}(\theta_c) = w_{f_c}^c \mathcal{L}_{f_c}^c(\theta_c) + w_0^c \mathcal{L}_0^c(\theta_c) + w_b^c \mathcal{L}_b^c(\theta_c)$;
		\State Update network parameters $\theta_c$ using LBFGS Optimizer;
		\EndFor
	\end{algorithmic}
\end{algorithm}

Our optimization process begins by initializing the neural network weights using the Xavier initialization scheme. The training then proceeds through three sequential phases. The first phase utilizes the Adam optimizer for $N_{Adam}^{(1)}$ steps with a learning rate of $\eta_{Adam}^{(1)}=10^{-3}$. To mitigate the difficulty of learning complex PDEs dynamics from a random state, we implement a dynamic loss weighting scheme. During the initial 2,000 steps, the weights for the initial condition loss $(\mathcal{L}_0^c(\theta_c))$ and the boundary condition loss $(\mathcal{L}_b^c(\theta_c))$ are amplified by a factor of 100 (i.e., $w_0^c=w_b^c=100$). This warm-up strategy enforces strict adherence to the spatiotemporal constraints before prioritizing the minimization of the PDE residual. Subsequently, the weights revert to unity to facilitate balanced optimization.

Following the initial exploration, we continue with the Adam optimizer for an additional $N_{Adam}^{(2)}$ steps but reduce the learning rate to $\eta_{Adam}^{(2)}=10^{-4}$. This decay allows the network to fine-tune the parameters and escape potential oscillations around local minima encountered during the aggressive exploration of Phase 1.

The final phase employs the L-BFGS optimizer. As a quasi-Newton method, L-BFGS approximates the Hessian matrix to incorporate second-order curvature information. This step is crucial for minimizing the residual to a low tolerance, achieving high-precision convergence once the solution is within the basin of the global minimum.

This coarse-to-fine strategy combines the advantages of both high and low learning rates, enabling rapid convergence while ensuring the accuracy and stability of the convergence.

Next, we shall focus our attention on the surface. The surface evolution problem is particularly challenging due to the computational complexity of high dimensions. To solve this problem, we employ a hybrid optimization strategy designed to effectively train the PINNs. Our training scheme is detailed in Algorithm \ref{alg:2}.

\begin{algorithm}[h!] 
	\caption{Training Scheme for Surface Evolution}
	\label{alg:2} 
	\begin{algorithmic}
		\State \textbf{Input:} Neural network $\mathcal{N}$ with parameters $\theta_s$; training data; Adam steps $N_{\text{Adam}}$; L-BFGS steps $N_{\text{LBFGS}}$;
		Adam learning rate scheduler (OneCycleLR); gradient clipping threshold $\lambda_{\text{clip}}$.
		\State \textbf{Output:} Optimized neural network parameters $\theta_s^*$.
		\State
		\State \textbf{Phase 1: Adam Optimization}
		\For{$i = 1$ to $N_{\text{Adam}}$}
		\State Sample training data;
		\If{$i < 10,000$}
		\State Set weights: $w_{f_s}^s \gets 1, w_0^s \gets 1,000, w_b^s \gets 1,000, w_p^s \gets 1,000$;
		\ElsIf{$i < 20,000$}
		\State Set weights: $w_{f_s}^s \gets 1, w_0^s \gets 0.1 \times w_0^s, w_b^s \gets 0.1 \times w_b^s, w_p^s \gets 0.1 \times w_p^s$;
		\Else
		\State Set weights: $w_{f_s}^s \gets 1, w_0^s \gets 1, w_b^s \gets 1, w_p^s \gets 1$;
		\State Compute total weighted loss: $\mathcal{L}_{\text{surface}}(\theta_s) = w_{f_s}^s \mathcal{L}_{f_s}^s(\theta_s) + w_0^s \mathcal{L}_0^s(\theta_s) + w_b^s \mathcal{L}_b^s(\theta_s) + w_p^s \mathcal{L}_p^s(\theta_s)$;
		\State Calculate gradients: $\nabla \mathcal{L}_{\text{surface}}(\theta_s)$;
		\State Clip gradients using $\lambda_{\text{clip}}$;
		\State Update network parameters $\theta_s$ using Adam Optimizer;
		\State Step OneCycleLR scheduler update learning rate;
		\EndIf
		\EndFor
		\State
		\State \textbf{Phase 2: L-BFGS Fine-Tuning}
		\State Sample a fixed set of training data for L-BFGS.
		\For{$i = 1$ to $N_{\text{LBFGS}}$}
		\State Set weights: $w_{f_s}^s \gets 1, w_0^s \gets 1, w_b^s \gets 1, w_p^s \gets 1$;
		\State Compute total weighted loss: $\mathcal{L}_{\text{surface}}(\theta_s) = w_{f_s}^s \mathcal{L}_{f_s}^s(\theta_s) + w_0^s \mathcal{L}_0^s(\theta_s) + w_b^s \mathcal{L}_b^s(\theta_s) + w_p^s \mathcal{L}_p^s(\theta_s)$;
		\State Update network parameters $\theta_s$ using LBFGS Optimizer;
		\EndFor
	\end{algorithmic}
\end{algorithm}

The training process is broadly divided into two phases: an initial Adam optimization phase followed by an L-BFGS fine-tuning phase. We first apply the Xavier scheme to initialize the neural network parameters $\theta_s$. An Adam optimizer is then initialized, coupled with a OneCycleLR learning rate scheduler to adjust the learning rate throughout this phase dynamically. In each Adam iteration, we first sample a fresh set of training points from the domain for the PDE residual ($\mathcal{L}_{f_s}^s(\theta_s)$), initial conditions ($\mathcal{L}_0^s(\theta_s)$), boundary conditions ($\mathcal{L}_b^s(\theta_s)$), and boundary conditions at pole($\mathcal{L}_p^s(\theta_s)$). This practice of resampling at each step is crucial for preventing the model from overfitting to a specific set of collocation points, thereby promoting the learning of a solution that generalizes across the entire continuous domain. The individual loss components are then computed. During this phase, we initially set $w_{f_s}^s = 1$, $w_0^s = w_b^s = w_p^s = 1000$, a strong emphasis is placed on enforcing the initial, boundary, and pole conditions. This aggressive weighting helps to rapidly stabilize the solution at critical points in the domain. Then, we update $w_0^s$, $w_b^s$, and $w_p^s$ by a factor of 0.1 when the training has run for 10,000 steps, allowing the PDE residual to gain more influence while still maintaining a robust enforcement of initial and boundary conditions. After 20,000 steps, all weights are set to unity (i.e., $w_{f_s}^s = w_0^s = w_b^s = w_p^s = 1$), indicating that the network has sufficiently learned the basic constraints, and the optimization can focus on minimizing the overall physical residual. To ensure training stability and prevent gradient explosion, we apply gradient clipping with a threshold $\lambda_{clip}$ after the gradient is calculated.

Upon completion of the Adam phase, the network parameters are loaded from the saved best state, providing an excellent starting point for high-precision optimization. An L-BFGS optimizer  is then employed. For L-BFGS, a fixed set of training points is sampled, and we fix unit weights(i.e., $w_{f_s}^s = w_0^s = w_b^s = w_p^s = 1$). The L-BFGS optimizer is intended to achieve fine-grained local convergence by leveraging second-order derivative information.

This hybrid optimization strategy effectively combines aggressive initial exploration and constraint enforcement with subsequent high-precision refinement, proving essential for robustly training PINNs for complex surface evolution problems.

Our methods are implemented across all experiments using the Python programming language coupled with the PyTorch framework. A Linux server equipped with a single NVIDIA RTX A6000 Graphics Processing Units is employed to execute these experiments.

\section{Numerical experiments}
\label{sec4}
This section provides detailed numerical experiments on HMCF. We consider the evolution of plane curves and surfaces with different initial conditions. To quantify the accuracy and effectiveness of the PINNs method, we introduce the relative $\mathbb{L}_2$ error as follows
\begin{equation*}
	Error = \frac{||u-u^*||_2}{||u||_2},
\end{equation*}
where the vectors $u$ and $u^*$ denote the reference solution and the PINNs solution, respectively.
\subsection{Hyperbolic mean curvature flow for curves}
Our first set of numerical experiments is for the evolution of an initially circular curve when $\beta = 0$. We consider the initial velocity to be constant and in the normal direction along the curve.

Let the initial curve and velocity be
\begin{equation}\label{eq41}
	\begin{aligned}
		& \gamma_0(u) = r_0(\cos u, \sin u), \\
		& \gamma_1(u) = -r_1\overrightarrow{N}_0 = r_1(\cos u, \sin u)
	\end{aligned}
\end{equation}
for $r_0 \in \mathbb{R}_+$, $r_1 \in \mathbb{R}$, $\overrightarrow{N}_0$ the inner normal vector of the initial circle.

Then, Eqs. (\ref{eq22}) can be rewritten as a second order ordinary differential equation for the time-dependent radius $r=r(t)$ \cite{Kong1}:
\begin{equation}\label{eq42}
	\left\{\begin{array}{ll}
		r_{tt} = -\frac{1}{r} \quad in \quad (0,T), \vspace{3mm}\\
		r(0) = r_0 ,\vspace{3mm}\\
		r_t(0) = r_1.
	\end{array}\right.
\end{equation}

\begin{lemma}\label{lemma1}
	The radially symmetric analytical solution of  the initial-value problem (\ref{eq42}) for $r_1=0$ is given by
	\begin{equation*}
		r(t) = r_0 exp\left(-\left(erf^{-1}\left(t\sqrt{2/r_0^2\pi}\right)\right)^2\right), \quad for \quad t\in(0,T), \quad T = r_0\sqrt{\pi/2}.
	\end{equation*}
	Whenever $r_1>0$, the solution is
	\begin{equation*}
		r(t) = r_0 e^\frac{r_1^2}{2}exp\left(-\left[erf^{-1}\left(-t e^{-\frac{r_1^2}{2}}\sqrt{2/r_0^2\pi}+erf(\frac{r_1}{\sqrt{2}})\right)\right]^2\right)
	\end{equation*}
	for $t\in[0,T_s]$, where
	\begin{equation*}
		T_s = \sqrt{\frac{\pi}{2}}r_0 e^\frac{r_1^2}{2}erf\left(\frac{r_1}{\sqrt{2}}\right).
	\end{equation*}
	For $t\in[T_s,T)$, the solution is given as the zero velocity solution with the initial radius equal to $r(T_s)=r_0e^\frac{r_1^2}{2}$.
\end{lemma}
\begin{proof}
	The proof is omitted here, see reference \cite{Monika}.
\end{proof}

For the initial value problem (\ref{eq22}) with the boundary conditions (\ref{eq31}), the PINNs solution is obtained on the time domain $[0,1.2]$ with the training data spanning the time interval $[0,1.1]$ for $r_1=0$, and $[0,3.4]$ with the training data spanning the time interval $[0,3.3]$ for $r_1=1$, respectively. Throughout the remainder of the numerical experiments within this subsection, we assume the initial conditions at $N_0=200$ points and periodic boundary conditions at $N_b=200$ points, as well as minimize the PDE residual at $N_f=20,000$ collocation points. All sampled points were generated using the Latin Hypercube Sampling strategy. We represent the latent solution by a 7-layer deep neural network with 50 neurons per hidden layer and a hyperbolic tangent activation function. The training scheme is the multi-stage hybrid optimization strategy as discussed in Section \ref{sec3}.

\begin{figure}[htpb] 
	\centering 
	\begin{subfigure}[b]{0.45\linewidth} 
		\centering 
		\includegraphics[width=\linewidth]{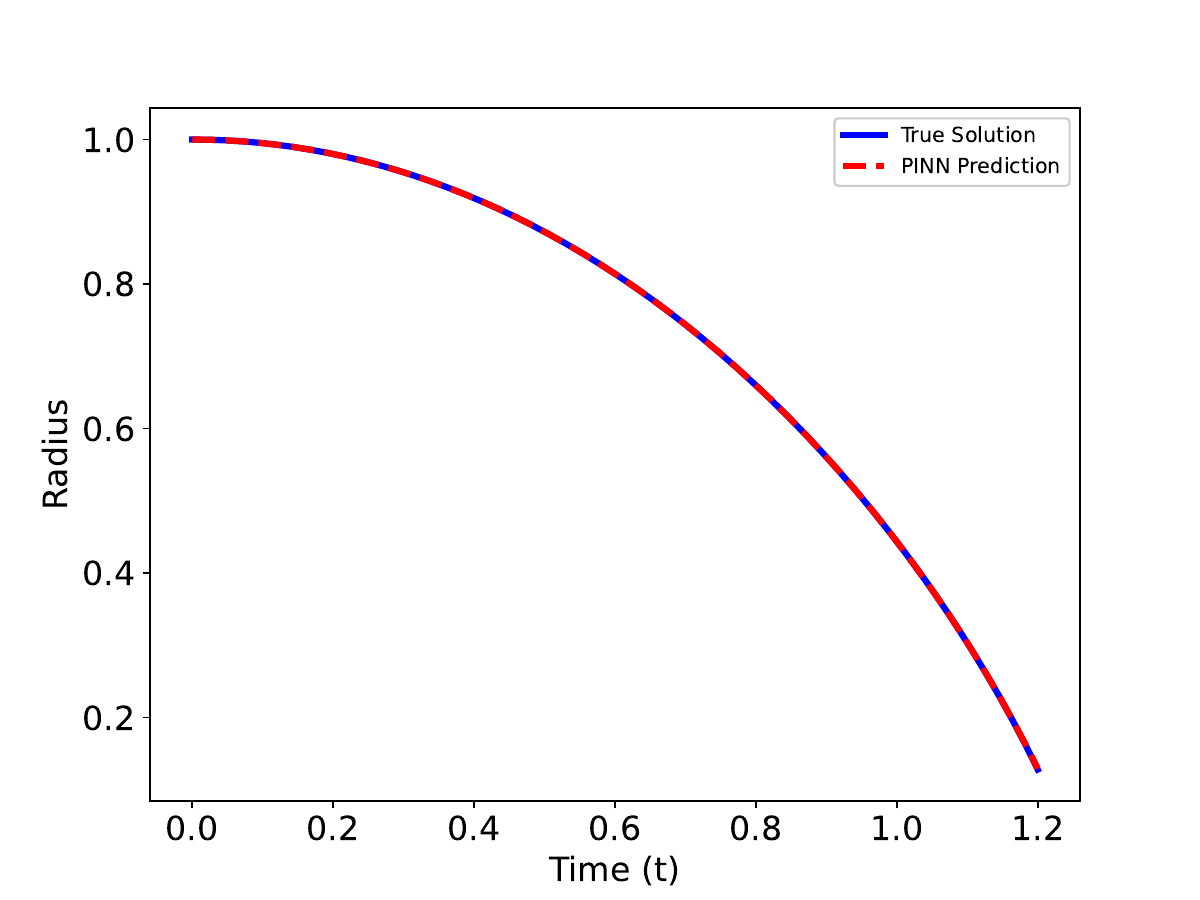}
	\end{subfigure}
	\hfill
	\begin{subfigure}[b]{0.45\linewidth}
		\centering
		\includegraphics[width=\linewidth]{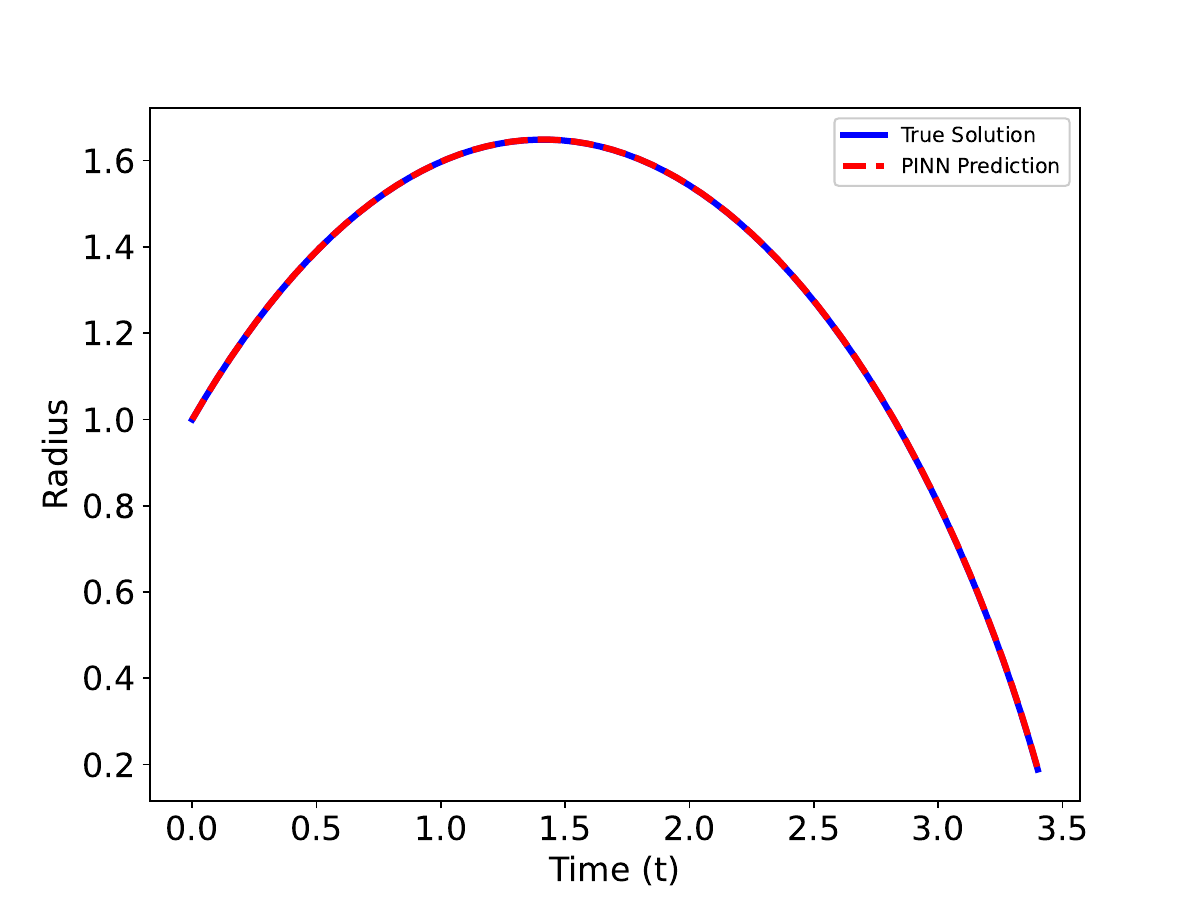} 
	\end{subfigure}
	
	{\captionsetup{font=small}
		\caption{\textit{Comparison of the PINNs and analytical solution of the HMCF starting from a unit circle. Examples with $r_1 = 0$ on the left, and $r_1 = 1$ on the right. The relative $\mathbb{L}_2$ errors are 1.91e-4 on the left, and 3.68e-4 on the right, respectively. The size of the deep neural network is $7\times50$.}}
		\label{fig1}
	}
\end{figure}

Figure \ref{fig1} shows the comparison between the mean radius of PINNs and analytical solution given by Lemma 4.1 for the initial radius $r_0=1$. The PINNs solution has a near-perfect agreement with the analytical solution, with a relative $\mathbb{L}_2$ error of 1.91e-4 for $r_1=0$, and 3.68e-4 for $r_1=1$. Note that when the initial velocity $r_1=0$, the radius monotonically decreases and reaches zero. On the other hand, when the initial velocity $r_1=1$, the radius increases first and then decreases and eventually reaches zero within a finite time.

To further analyze the performance of our method, we quantify the prediction accuracy through a series of virtual experiments, including exploring various neural network architectures and examining the impact of different numbers of training and collocation points. We also study the effects on the optimizer. These self-ablation tests aim to systematically determine the optimal parameters for HMCF.

\begin{table}[h!]
	\centering
	\small
	\begin{threeparttable}
		\begin{tabular}{c|cccc}
			\multicolumn{5}{c}{\textnormal{(a) $r_1=0$}} \\[1mm]
			\Xhline{0.7pt}
			\diagbox{Layers}{Neurons} & 10 & 25 & 50 & 100 \\
			\hline
			1 & 9.64e-1 & 9.62e-1 & 8.94e-1 &  9.12e-1\\
			3 & 9.57e-1 & 9.15e-1 & 9.47e-1 &  9.15e-1\\
			5 & 9.99e-1 & 2.67e-4 & 2.38e-4 &  9.11e-1\\
			7 & 9.31e-1 & 6.10e-4 & 1.91e-4 &  8.93e-1\\
			\Xhline{0.7pt}
		\end{tabular}
		\vspace{2mm}
		\begin{tabular}{c|cccc}
			\multicolumn{5}{c}{\textnormal{(b) $r_1=1$}} \\[1mm]
			\Xhline{0.7pt}
			\diagbox{Layers}{Neurons} & 10 & 25 & 50 & 100 \\
			\hline
			1 & 1.29e+0 & 1.58e+0 & 3.50e+0 & 1.91e+0 \\
			3 & 8.13e-1 & 2.32e+0 & 8.31e-3 & 9.77e-4 \\
			5 & 9.85e-1 & 5.68e-3 & 8.56e-4 & 5.86e-4 \\
			7 & 5.83e-1 & 7.41e-1 & 3.68e-4 & 4.24e-1 \\
			\Xhline{0.7pt}
		\end{tabular}
		\caption{\small HMCF: Relative $\mathbb{L}_2$ error between PINNs and analytical solution for different numbers of hidden layers neurons in each layer. Here, the total number of training and collocation points is fixed to $N_0 = N_b = 200$ and $N_f = 20,000$, respectively.}
		\label{tab1}
	\end{threeparttable}
\end{table}

Analysis of Table \ref{tab1} reveals that network architecture plays a critical role in the accuracy of the solution. Here, the total number of training and collocation points is fixed to be $N_0=N_b=200$ and $N_f=20,000$, respectively. The initial entries (first row and column) demonstrate that shallow networks and few neurons lack the necessary representative capacity to model the solution, which manifests as underfitting and a high relative $\mathbb{L}_2$ error. Progressing towards the center of the table, a clear pattern emerges that increasing network depth and width generally leads to a systematic reduction in the relative $\mathbb{L}_2$ error, yielding more precise approximations. We note that for the same number of residual points, the deep neural network with 7 layers and 100 neurons per layer ($7 \times 100$) exhibits a larger relative $\mathbb{L}_2$ error than its less deep
$5 \times 100$ counterpart. Such a phenomenon suggests that larger models require more residual points for training generally.

\begin{figure}[htpb] 
	\centering 
	\begin{subfigure}[b]{0.23\linewidth} 
		\centering 
		\includegraphics[width=\linewidth]{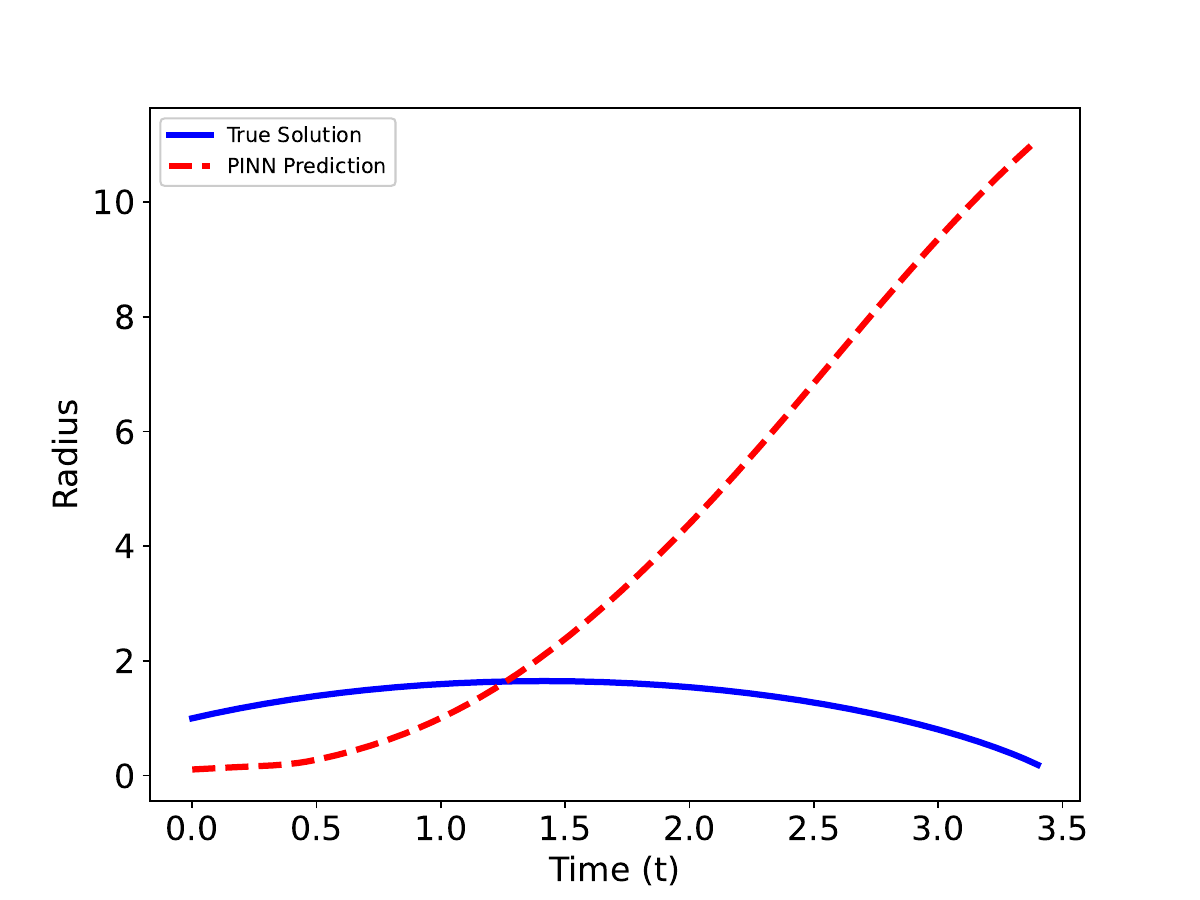}
		\caption{$1\times50$}
	\end{subfigure}
	\hfill
	\begin{subfigure}[b]{0.23\linewidth}
		\centering
		\includegraphics[width=\linewidth]{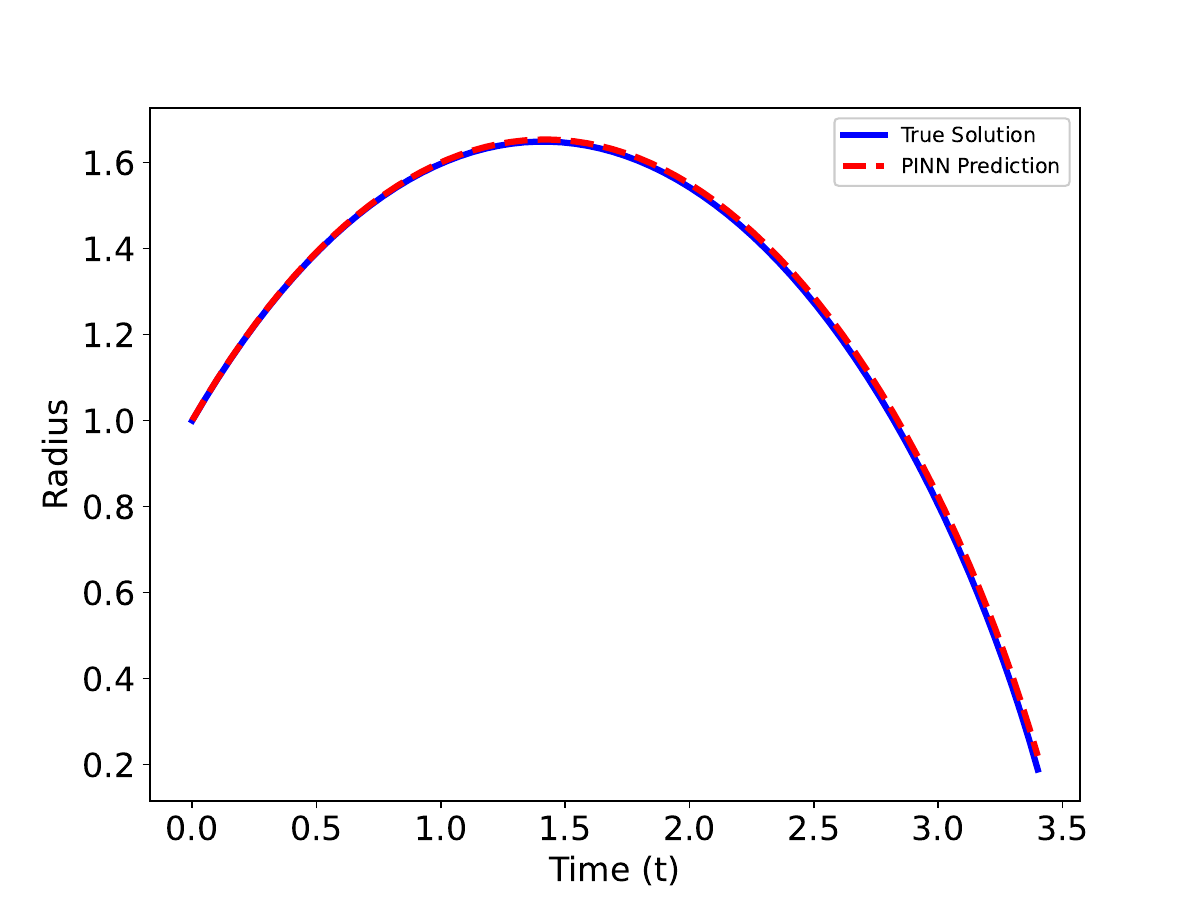} 
		\caption{$3\times50$}
	\end{subfigure}
	\hfill
	\begin{subfigure}[b]{0.23\linewidth}
		\centering
		\includegraphics[width=\linewidth]{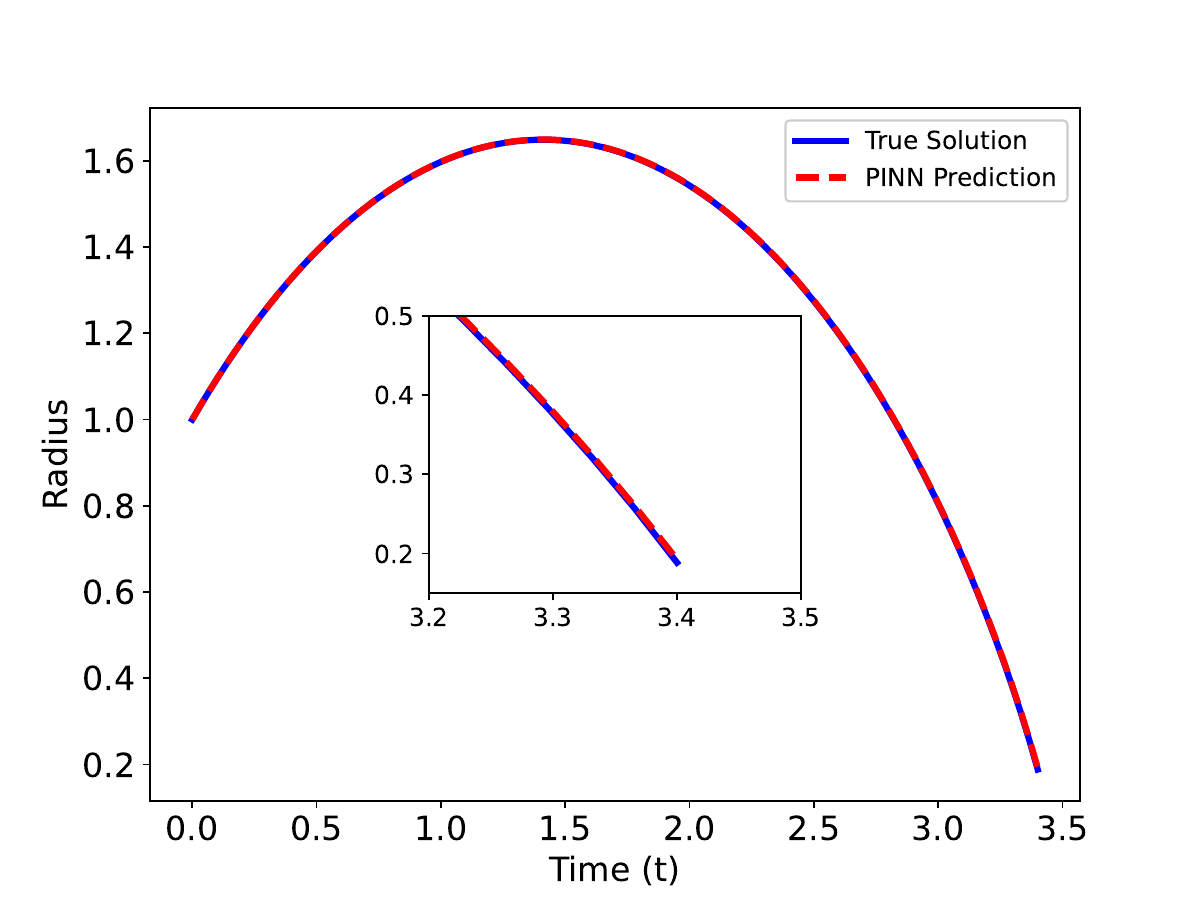}
		\caption{$5\times50$} 
	\end{subfigure}
	\hfill
	\begin{subfigure}[b]{0.23\linewidth}
		\centering
		\includegraphics[width=\linewidth]{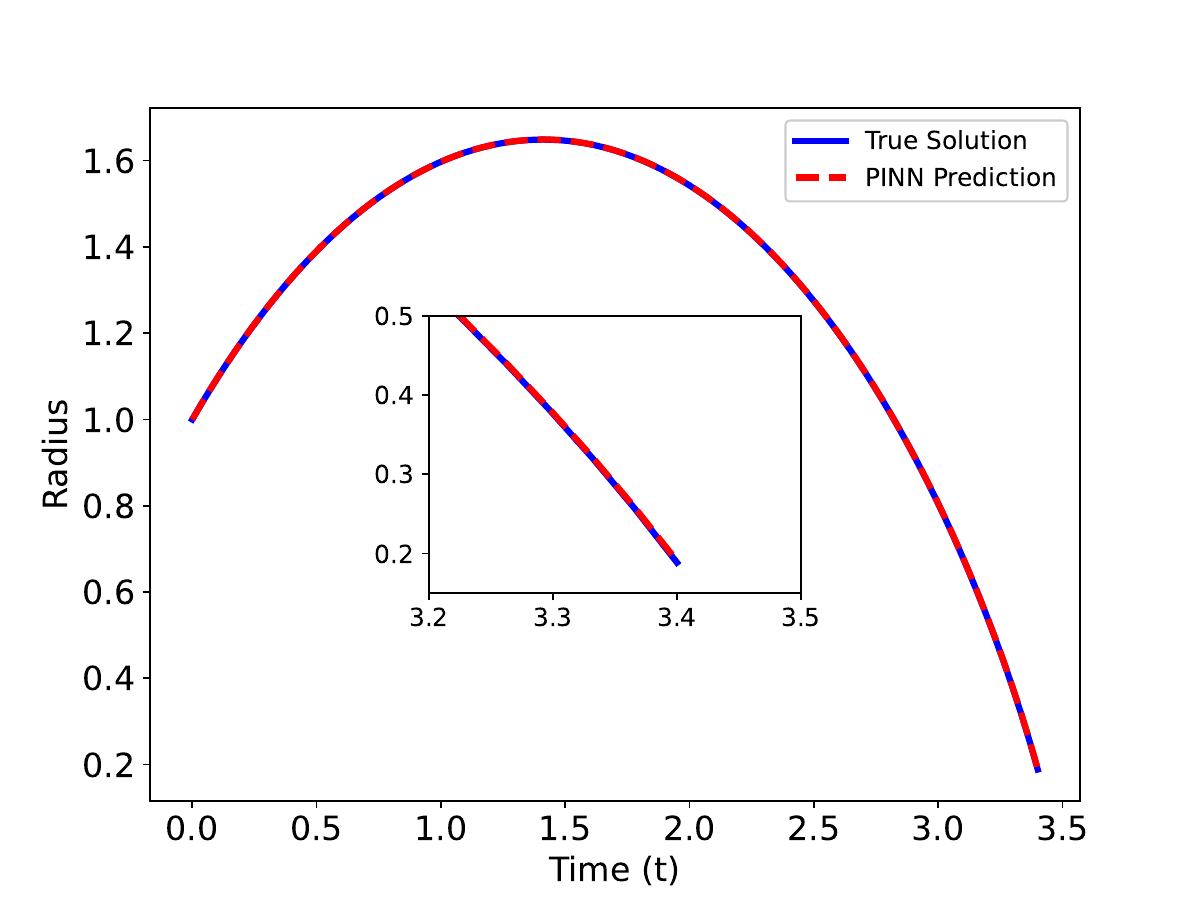} 
		\caption{$7\times50$}
	\end{subfigure}
	\hfill
	\begin{subfigure}[b]{0.23\linewidth}
		\centering
		\includegraphics[width=\linewidth]{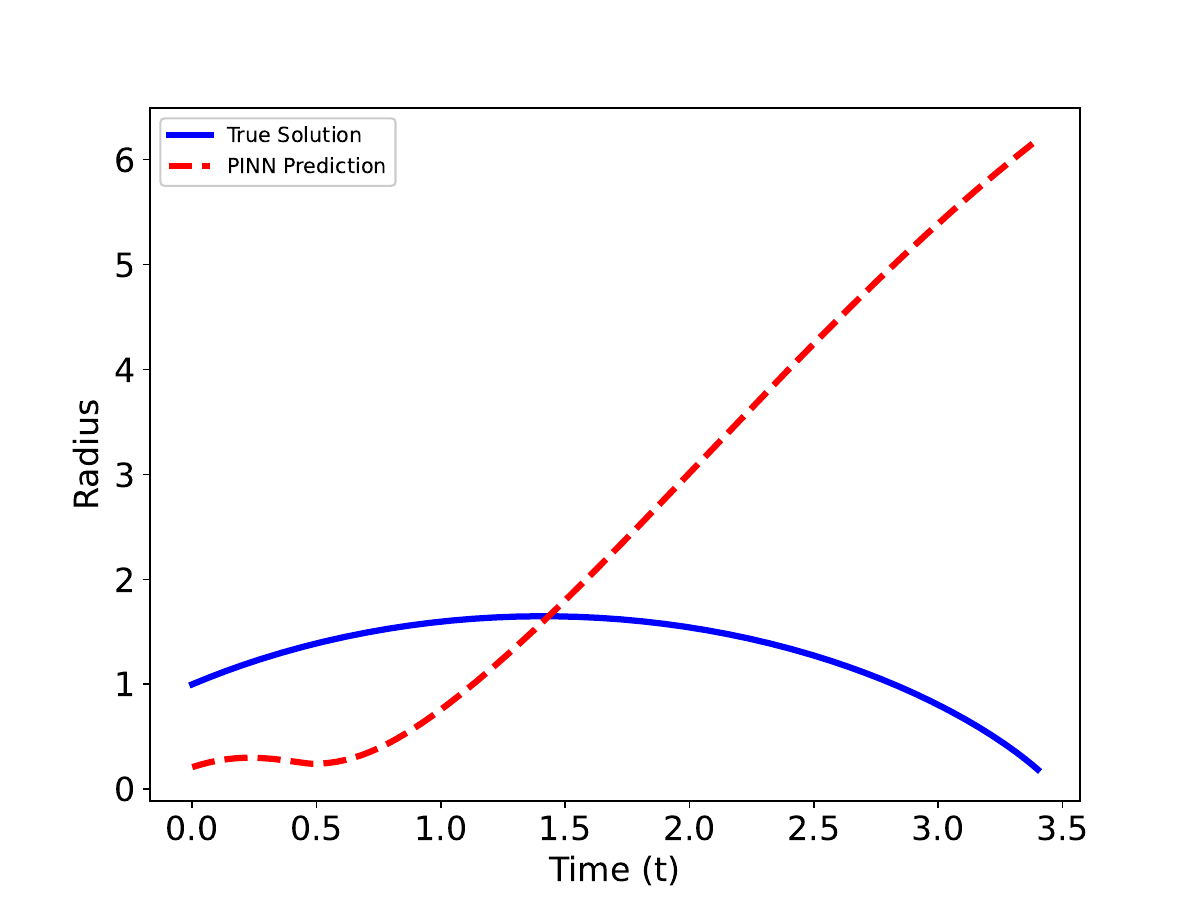}
		\caption{$1\times100$} 
	\end{subfigure}
	\hfill
	\begin{subfigure}[b]{0.23\linewidth}
		\centering
		\includegraphics[width=\linewidth]{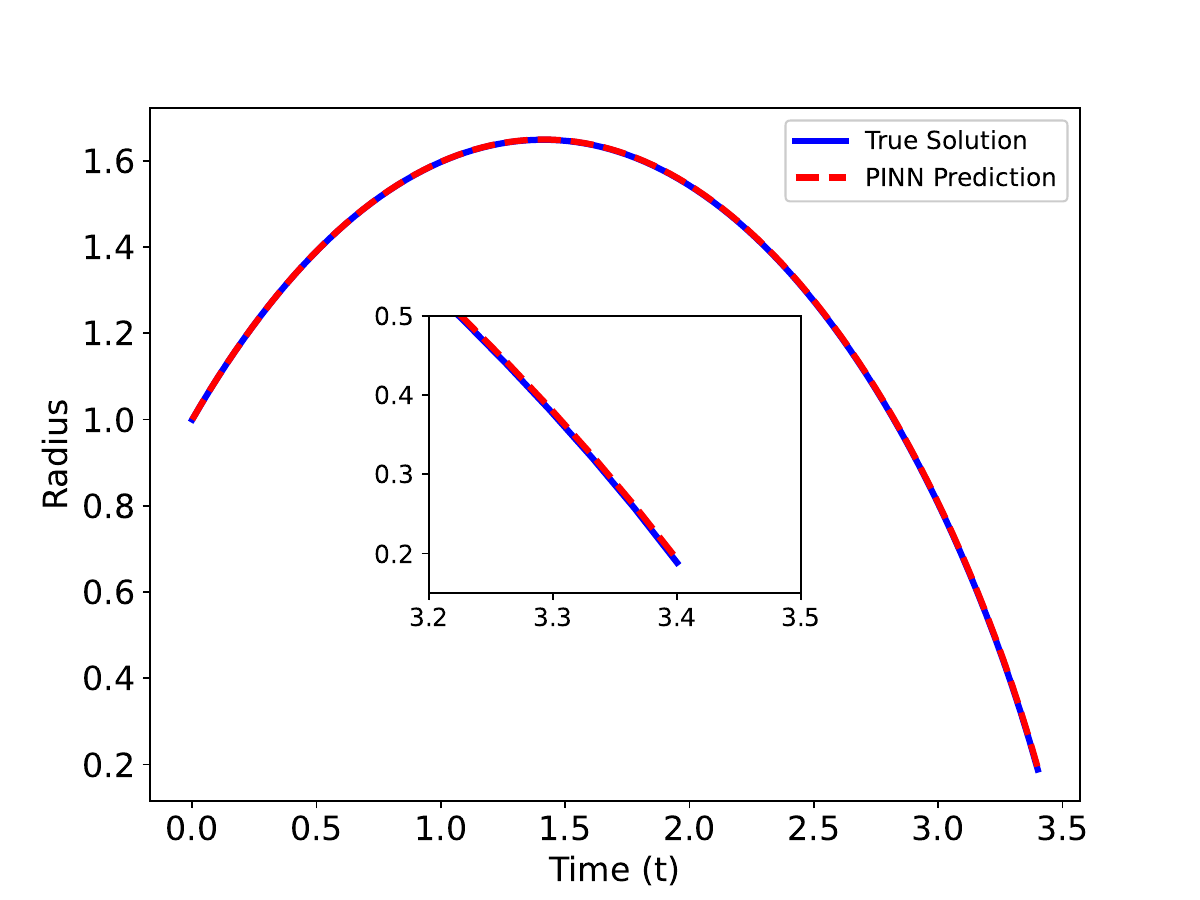}
		\caption{$3\times100$}  
	\end{subfigure}
	\hfill
	\begin{subfigure}[b]{0.23\linewidth}
		\centering
		\includegraphics[width=\linewidth]{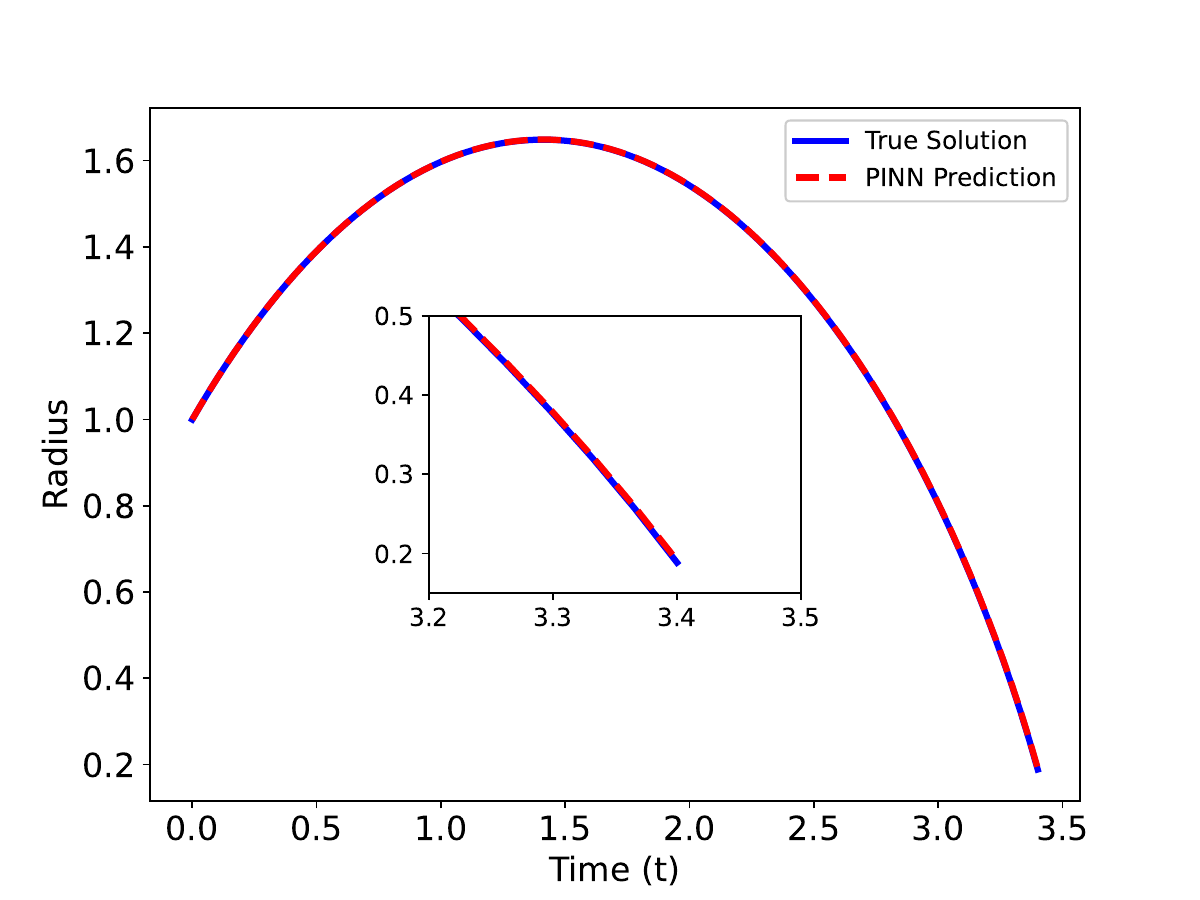} 
		\caption{$5\times100$} 
	\end{subfigure}
	\hfill
	\begin{subfigure}[b]{0.23\linewidth}
		\centering
		\includegraphics[width=\linewidth]{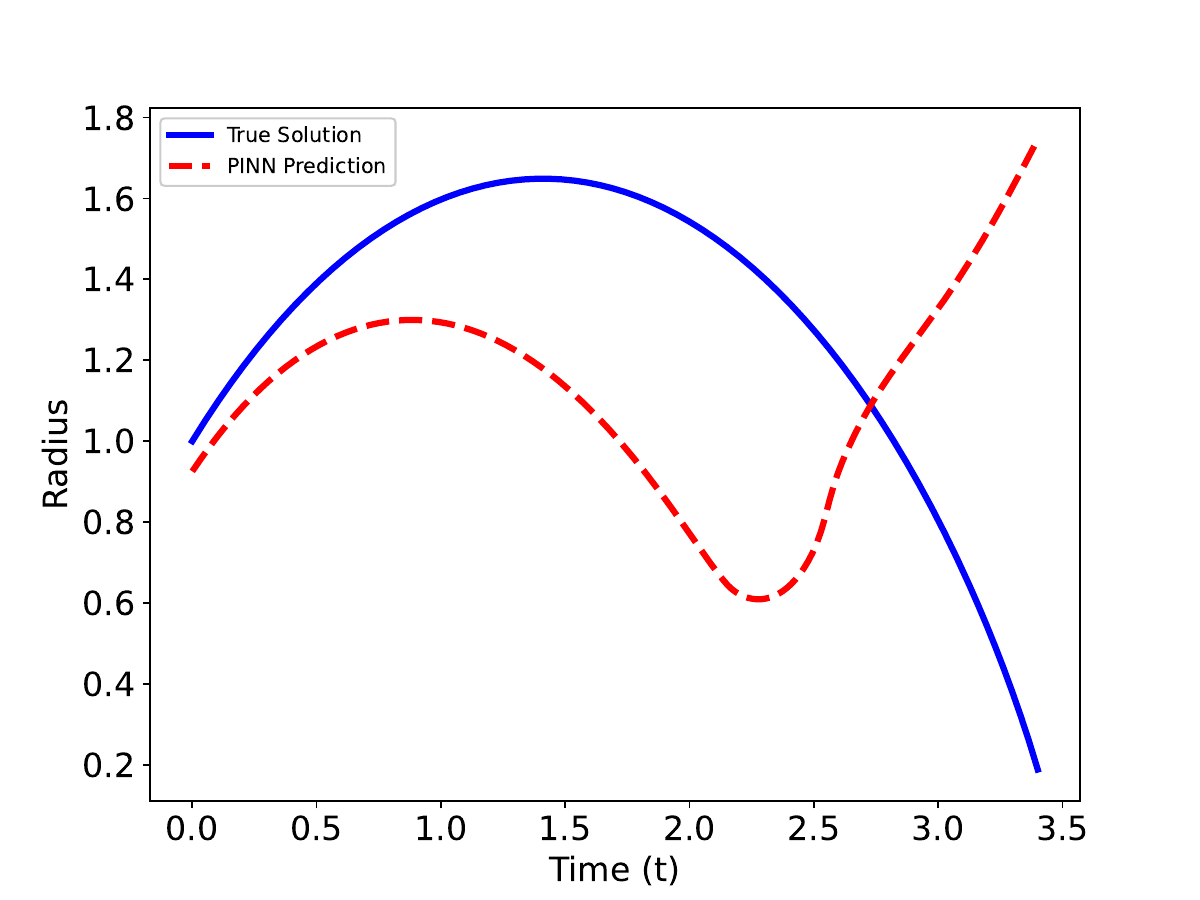} 
		\caption{$7\times100$} 
	\end{subfigure}
	
	{\caption{\textit{Comparisons of the PINNs and analytical solution of the HMCF starting from a unit circle. We show the evolution of radius over time for the initial velocity $r_1 = 1$, above the deep neural network size is $1 \times 50$, $3 \times 50$, $5 \times 50$ and $7 \times 50$ (from left to right), blew the deep neural network size is $1 \times 100$, $3 \times 100$, $5 \times 100$ and $7 \times 100$ (from left to right).}}
	\label{fig2}
	}
\end{figure}

Figure \ref{fig2} shows the PINNs and analytical solution for different sizes of deep neural network for $r_1=1$ . In Figures 2(a)-2(d), the PINNs solution is obtained with a $1 \times 50$, $3 \times 50$, $5 \times 50$ and $7 \times 50$ deep neural network, while the solution in Figures 2(e)-2(h) is obtained with a bigger $1 \times 100$, $3 \times 100$, $5 \times 100$ and $7 \times 100$ deep neural network. It is obvious that the solution with 50 neurons per layer provides an excellent fit to the analytical solution as the network depth increases, while the solution with 100 neurons per layer fails to converge to the analytical solution when the network depth reaches 7.

\begin{table}[h!]
	\centering
	\small
	\begin{threeparttable}
		\begin{tabular}{c|cccccc}
			\multicolumn{7}{c}{\textnormal{(a) $r_1=0$}} \\[1mm]
			\Xhline{0.7pt}
			\diagbox{$N_0=N_b$}{$N_f$} & 1,000 & 3,000 & 5,000 & 7,000 & 10,000 & 20,000 \\
			\hline
			30 & 8.46e-1 & 2.25e-4 & 1.65e-4 & 1.60e-4 & 2.38e-4 & 2.77e-4  \\
			50 & 2.87e-4 & 8.45e-1 & 2.61e-4 & 8.47e-1 & 3.06e-4 & 4.55e-4  \\
			70 & 3.49e-4 & 3.60e-4 & 3.33e-4 & 8.60e-1 & 8.53e-1 & 3.10e-4  \\
			100 & 8.48e-1 & 2.58e-4 & 2.37e-4 & 2.32e-4 & 8.71e-1 & 2.74e-4  \\
			200 & 1.80e-4 & 2.25e-4 & 1.52e-3 & 2.37e-4 & 2.47e-4 & 1.91e-4  \\
			\Xhline{0.7pt}
		\end{tabular}
		\vspace{2mm}
		\begin{tabular}{c|cccccc}
			\multicolumn{7}{c}{\textnormal{(b) $r_1=1$}} \\[1mm]
			\Xhline{0.7pt}
			\diagbox{$N_0=N_b$}{$N_f$} & 1,000 & 3,000 & 5,000 & 7,000 & 10,000 & 20,000 \\
			\hline
			30 & 9.98e-1 & 5.87e-1 & 3.73e-1 & 3.92e-1 & 1.48e-4 & 6.48e-4  \\
			50 & 5.79e-1 & 6.95e-4 & 9.14e-4 & 3.81e-1 & 1.16e-3 & 4.92e-4  \\
			70 & 4.48e-1 & 3.99e-4 & 6.35e-4 & 7.04e-1 & 7.78e-1 & 3.53e-4  \\
			100 & 1.11e+0 & 7.68e-1 & 8.88e-4 & 3.66e-4 & 9.75e-4 & 1.07e-4  \\
			200 & 8.03e-1 & 8.67e-1 & 6.14e-1 & 3.89e-1 & 8.26e-4 & 3.68e-4  \\
			\Xhline{0.7pt}
		\end{tabular}
		\caption{\small HMCF: Relative $\mathbb{L}_2$ error between PINNs and analytical solution for different numbers of initial and boundary training data $N_0$ and $N_b$, and different numbers of collocation points $N_f$. Here, the network architecture is fixed to 7 layers with 50 neurons per hidden layer.}
		\label{tab2}
	\end{threeparttable}
\end{table}

In Table \ref{tab2} we report the relative $\mathbb{L}_2$ error for varying numbers of initial and boundary training data $N_0$ and $N_b$ and different numbers of collocation points $N_f$. The neural network model is composed of 7 hidden layers, each comprising 50 neurons. With a sufficient number of training data, the model achieves a significantly lower predictive error, reaching an order of magnitude of $10^{-4}$. 

\begin{table}[h!]
	\centering
	\small
	\begin{threeparttable}
		\begin{tabular}{cccccc}
			\multicolumn{6}{c}{\textnormal{(a) $r_1=0$}} \\[1mm]
			\Xhline{0.7pt}
			&\ Adam & Adam & Adam/L-BFGS & Adam/L-BFGS & Ours \\
			\hline
			Learning rate & 1e-3 & 1e-4 & 1e-3/1e-1 & 1e-4/1e-1 & 1e-3/1e-4/1e-1 \\
			Steps & 80,000 & 80,000 & 80,000/500 & 80,000/500 & 20,000/60,000/500\\
			Relative $\mathbb{L}_2$ error & 1.73e-3 & 8.91e-1 & 2.56e-4 & 8.73e-1 & 1.91e-4 \\
			\Xhline{0.7pt}
		\end{tabular}
		\vspace{2mm}
		\begin{tabular}{cccccc}
			\multicolumn{6}{c}{\textnormal{(b) $r_1=1$}} \\[1mm]
			\Xhline{0.7pt}
			&\ Adam & Adam & Adam/L-BFGS & Adam/L-BFGS & Ours \\
			\hline
			Learning rate & 1e-3 & 1e-4 & 1e-3/1e-1 & 1e-4/1e-1 & 1e-3/1e-4/1e-1 \\
			Steps & 80,000 & 80,000 & 80,000/500 & 80,000/500 & 20,000/60,000/500\\
			Relative $\mathbb{L}_2$ error & 1.20e-2 & 1.09e+0 & 1.17e-3 & 1.23e+0 & 3.68e-4\\
			\Xhline{0.7pt}
		\end{tabular}
		\caption{\small HMCF: Relative $\mathbb{L}_2$ error between PINNs and analytical solution for different optimizers and learning rates. Here, the network architecture is fixed to 7 layers with 50 neurons per hidden layer, the total number of training and collocation points is fixed to $N_0 = N_b = 200$ and $N_f = 20,000$, respectively.}
		\label{tab3}
	\end{threeparttable}
\end{table}

\begin{figure}[htpb] 
	\centering 
	\begin{subfigure}[b]{0.45\linewidth} 
		\centering 
		\includegraphics[width=\linewidth]{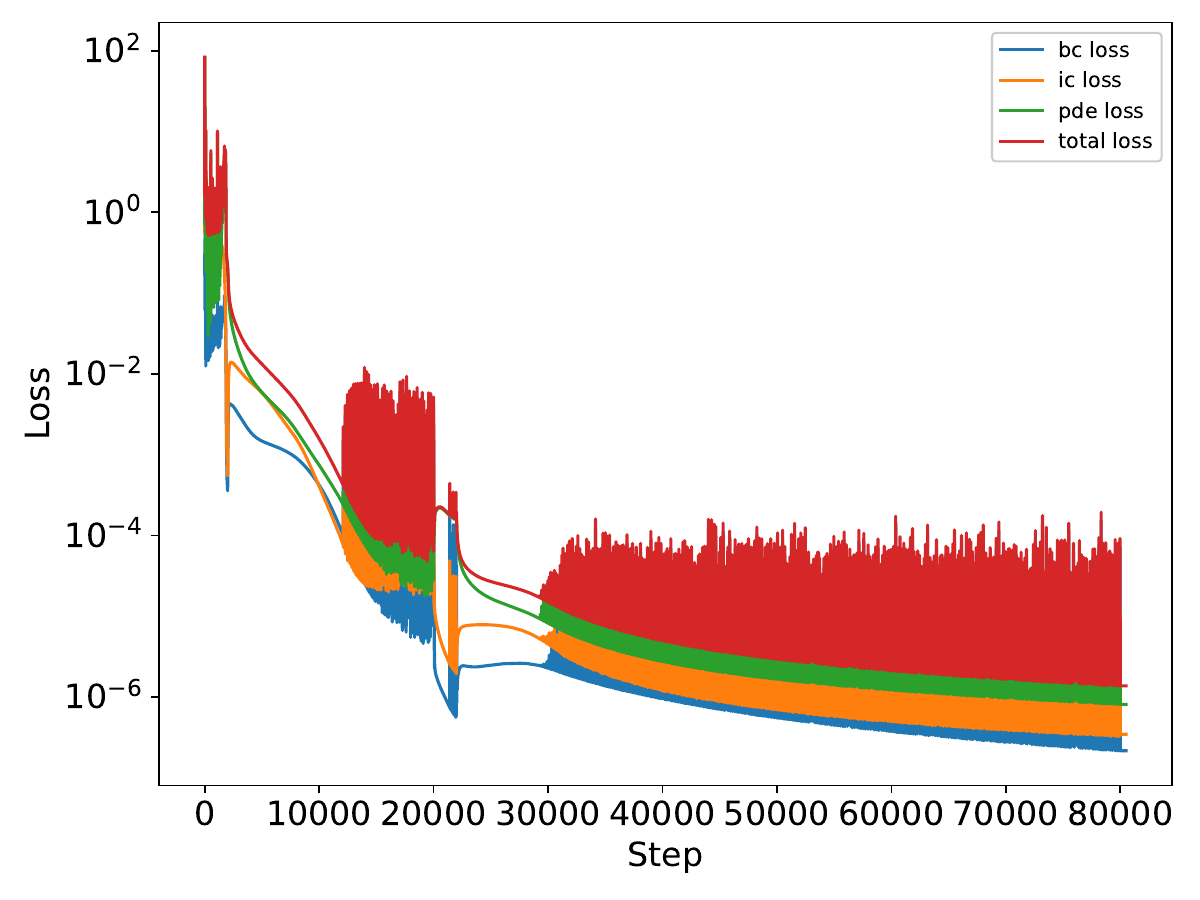}
		\caption{$r_1=0$} 
	\end{subfigure}
	\hfill
	\begin{subfigure}[b]{0.45\linewidth}
		\centering
		\includegraphics[width=\linewidth]{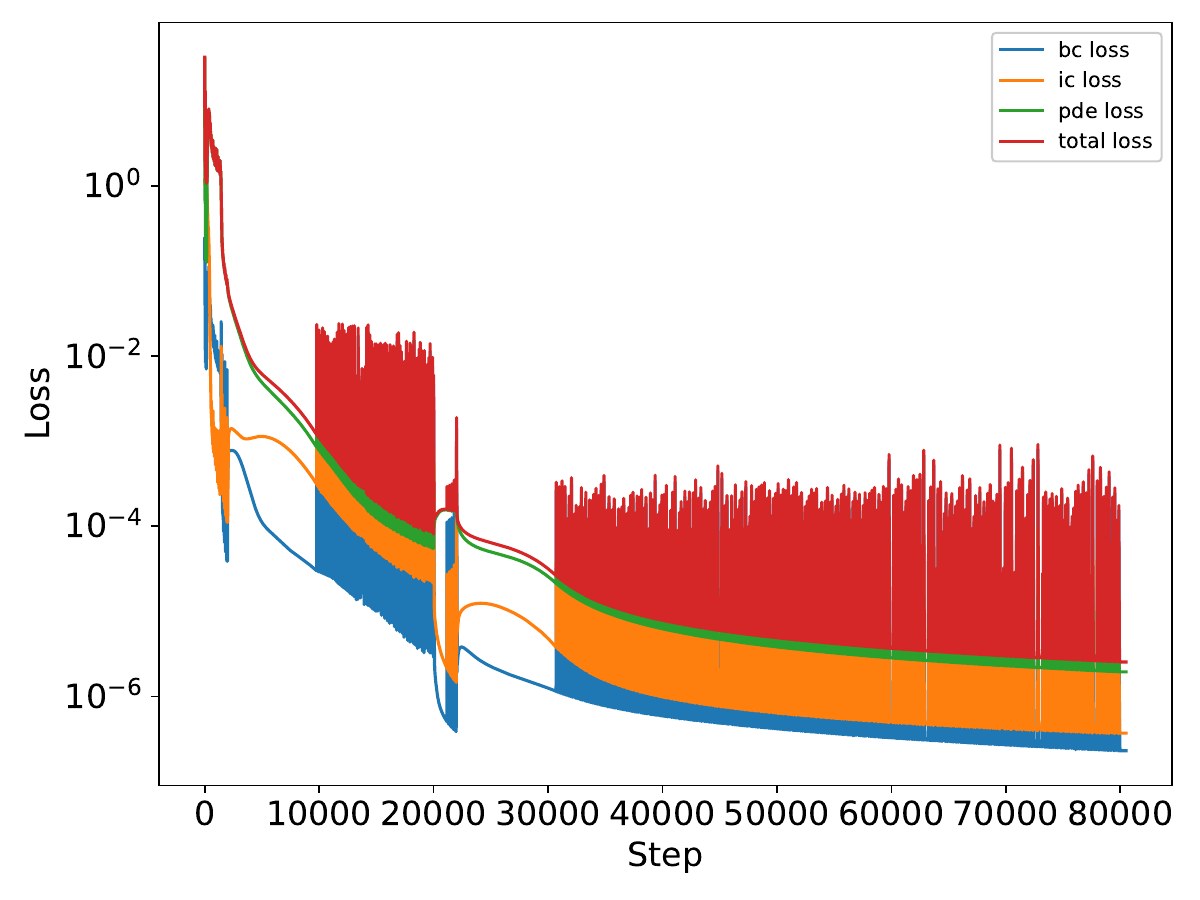} 
		\caption{$r_1=1$}
	\end{subfigure}
	{\captionsetup{font=small}
		\caption{\textit{Loss functions for the HMCF with (a) $r_1=0$ and (b) $r_1=1$. The deep neural network size is $7\times50$.}}
		\label{fig3}}
\end{figure}

The proposed architecture features a multi-stage hybrid optimization strategy. While the Adam optimizer at a relatively high learning rate allows the network parameters to approach the minimum
loss region rapidly, it suffers from oscillations near the minimum, potentially impeding convergence to a high-precision solution. In contrast, a lower learning rate enables the optimizer to perform fine-tuning, taking smaller, more stable steps to settle into a more optimal point within the local minimum. By utilizing Adam optimizer with these two learning rate sequentially, we aim to leverage the benefits of both. To assess this approach, we conducted experiments comparing three training strategies: our multi-stage hybrid optimization strategy, the use of the Adam optimizer with different learning rates exclusively, and a two-stage approach combining Adam with L-BFGS. We fix the total number of training and collocation points  $N_0=N_b=200$ and $N_f=20,000$, respectively, applying a neural network architecture of 7 hidden layers with 50 neurons each layer. Results, presented in Table \ref{tab3}, demonstrate that our proposed multi-stage hybrid optimization strategy achieves an outstanding accuracy and showcase its superior approximation capability.

Moreover, Figure \ref{fig3} visually encapsulates these findings. The total loss and the individual loss integrate a rapid exploration phase (Adam, 1e-3), a stabilizing fine-tuning phase (Adam, 1e-4), and a final high-fidelity refinement stage (L-BFGS), ensuring both efficient global search and precise local convergence.

For the next computations, we apply the multi-stage hybrid optimization strategy to other convex initial curves with different initial velocities. We now consider the HMCF for the initial curve given by the parametrization
\begin{equation*}
	\gamma_0(u) = (1.5cosu, sinu), \quad u\in[0,2\pi].
\end{equation*}

\begin{figure}[htpb] 
	\centering 
	\begin{subfigure}[b]{0.31\linewidth} 
		\centering 
		\includegraphics[width=\linewidth]{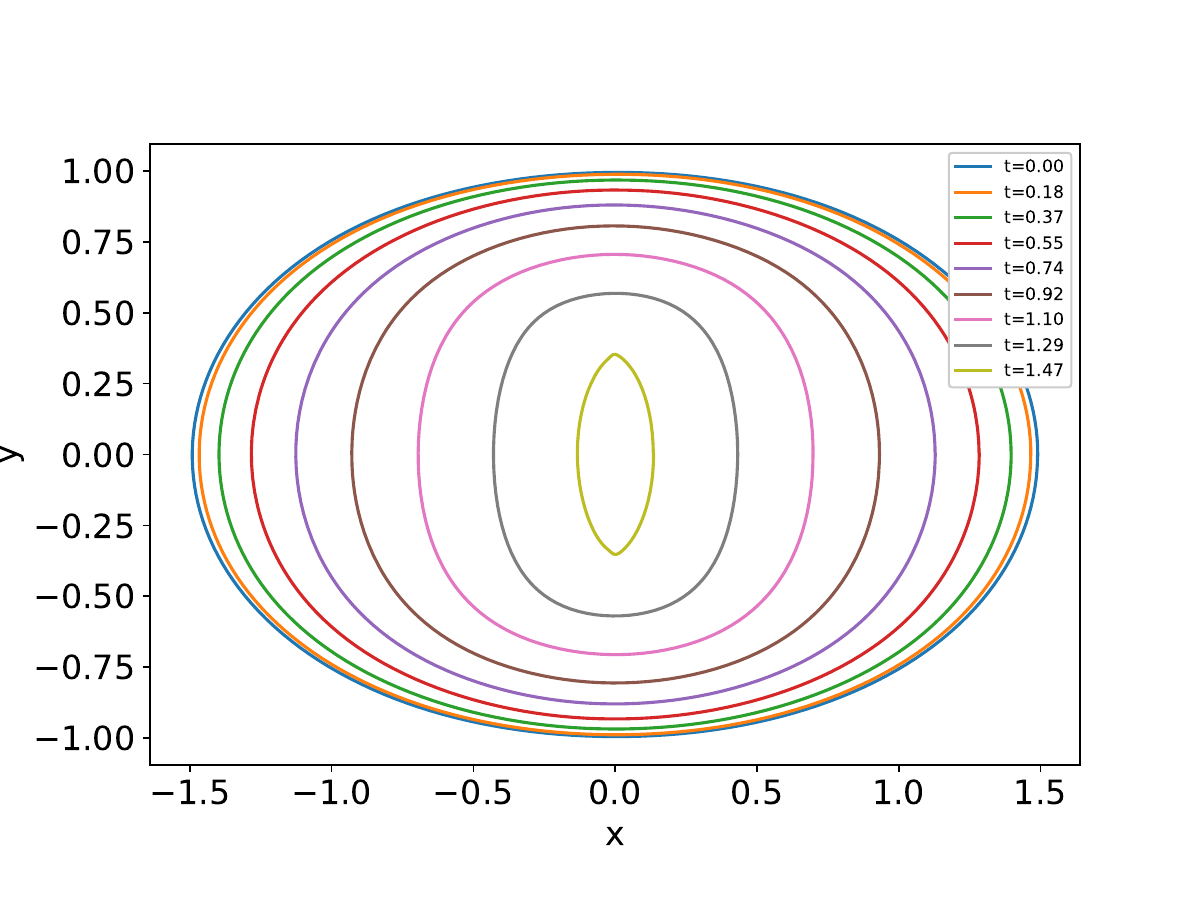}
	\end{subfigure}
	\hfill
	\begin{subfigure}[b]{0.31\linewidth}
		\centering
		\includegraphics[width=\linewidth]{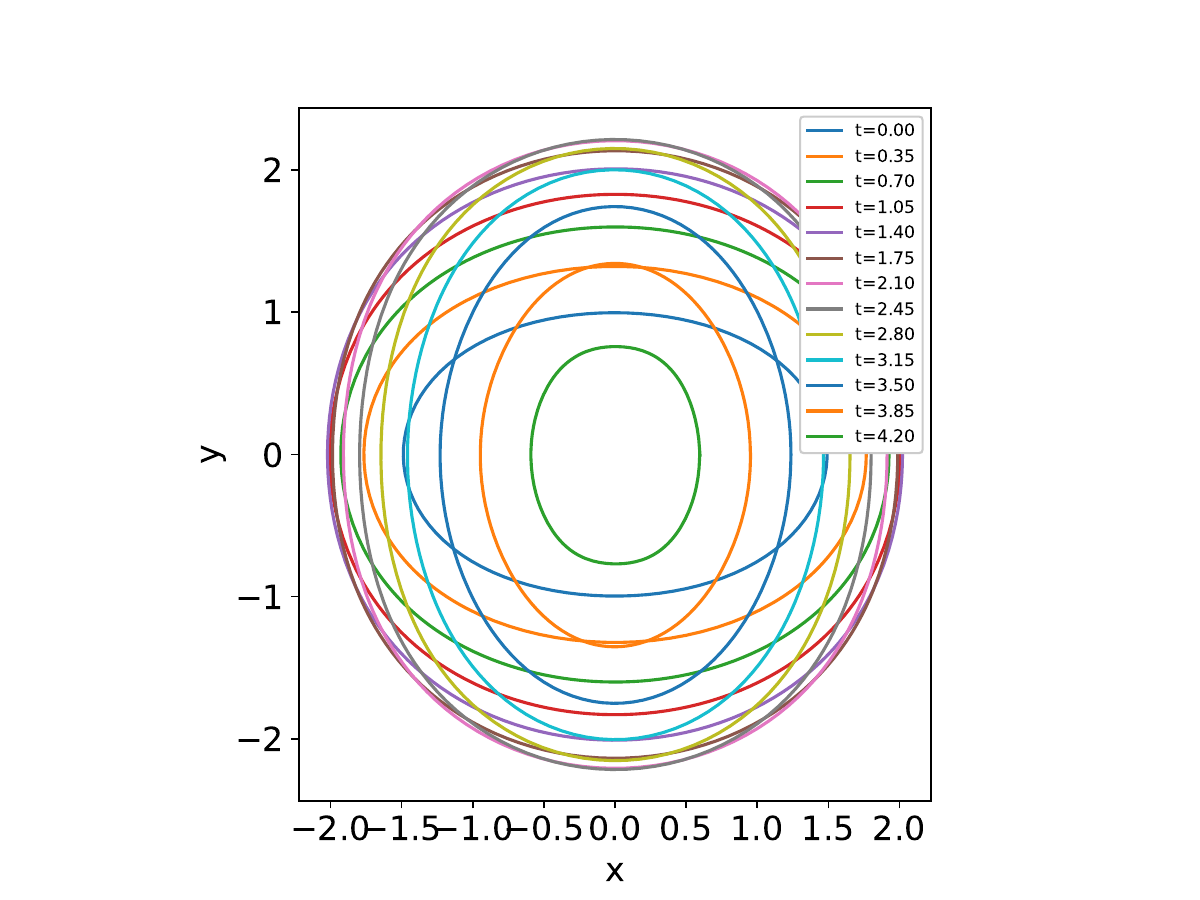} 
	\end{subfigure}
	\hfill
	\begin{subfigure}[b]{0.31\linewidth}
		\centering
		\includegraphics[width=\linewidth]{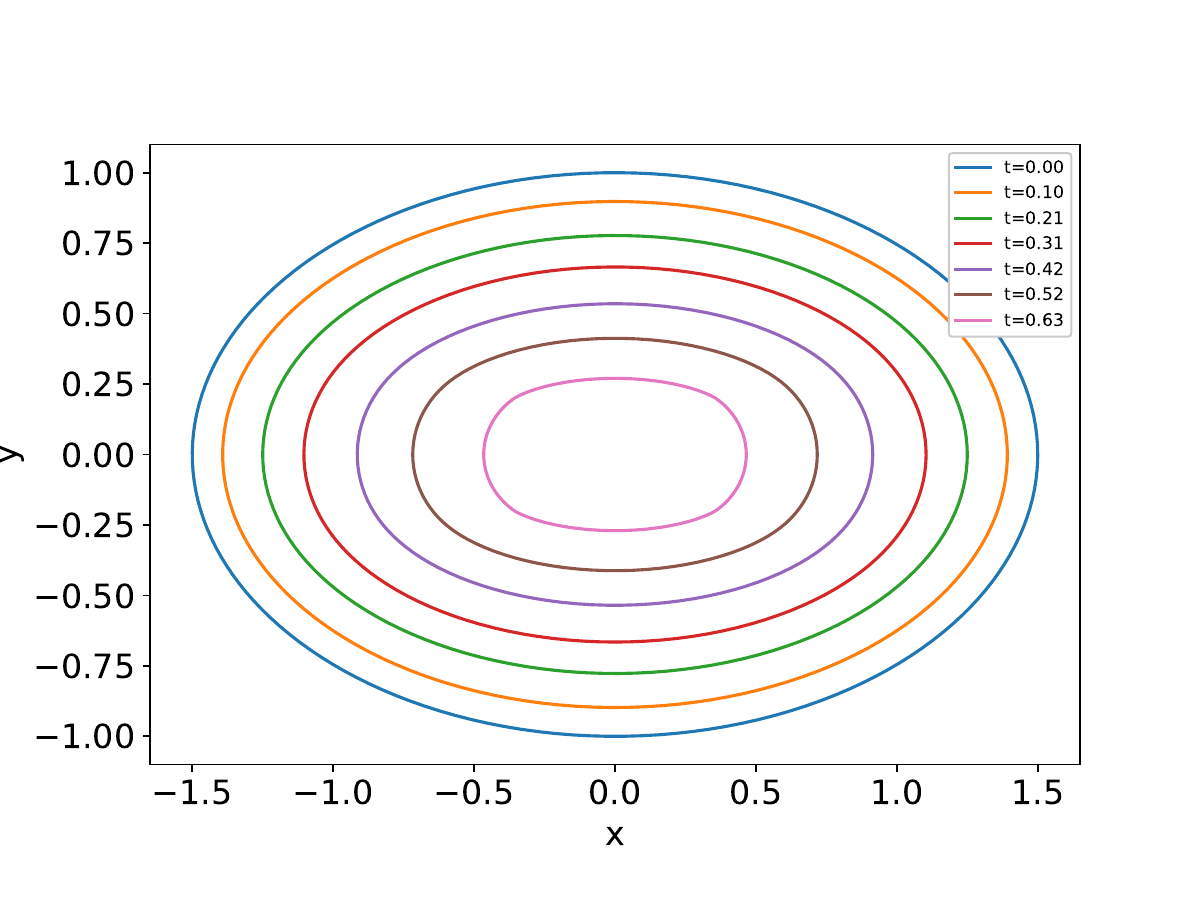} 
	\end{subfigure}
	
	{\captionsetup{font=small}
		\caption{\textit{HMCF starting from an ellipse. We show the evolution for $r_1=0$, $r_1=1$, and $r_1=-1$ (from left to right).}}
		\label{fig4}}
\end{figure}

In Figure \ref{fig4}, we can see that the evolution of the ellipse under HMCF is critically dependent on its initial velocity $r_1$. For the case of zero initial velocity depicted in Figure \ref{fig4}(a), the motion is entirely governed by the inward acceleration driven by curvature. The inherent non-uniformity of the ellipse's curvature - maximal at the major axis vertices and minimal at the minor axis vertices - induces a differential acceleration field, causing the rate of shrinking along the major axis to significantly exceed that along the minor axis. This leads to a rapid decrease in eccentricity, with the ellipse passing through a transient, nearly circular state before its geometry inverts, transforming into a vertically oriented ellipse as the original horizontal axis becomes shorter than the vertical one. This new configuration, with its correspondingly inverted curvature distribution, then shrinks inward while largely preserving its vertical orientation, and eventually the curve appears to form two kinks.

In contrast, an initial outward velocity $r_1=1$ shown in Figure \ref{fig4}(b), instigates a competition between the initial outward motion and the persistent inward acceleration. The non-uniform acceleration field results in a non-synchronous reversal of motion, the major axis vertices with the biggest curvature decelerate most rapidly and are the first to move inward, while the minor axis vertices with the smallest curvature expand further before reversing course. This asynchronous induces a deviation from the original elliptical geometry, ultimately giving rise to four pronounced corners.

Finally, for an initial inward velocity $r_1=-1$ illustrated in Figure \ref{fig4}(c), the initial velocity and acceleration vectors are in the same direction, both pointing inward. This synergy produces accelerated shrinkage, resulting in the shortest evolution time among the three cases.

\begin{figure}[htpb] 
	\centering 
	\begin{subfigure}[b]{0.45\linewidth} 
		\centering 
		\includegraphics[width=\linewidth]{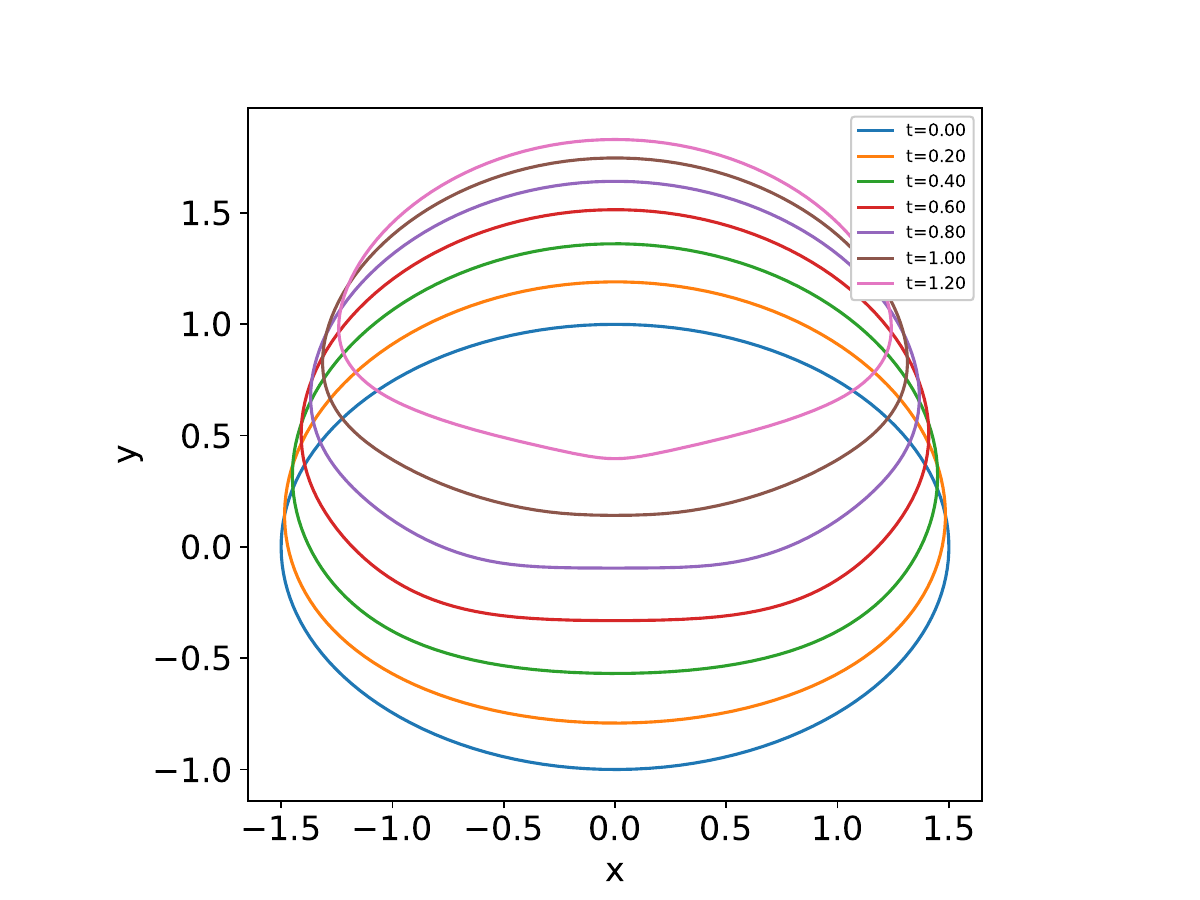}
	\end{subfigure}
	\hfill
	\begin{subfigure}[b]{0.45\linewidth}
		\centering
		\includegraphics[width=\linewidth]{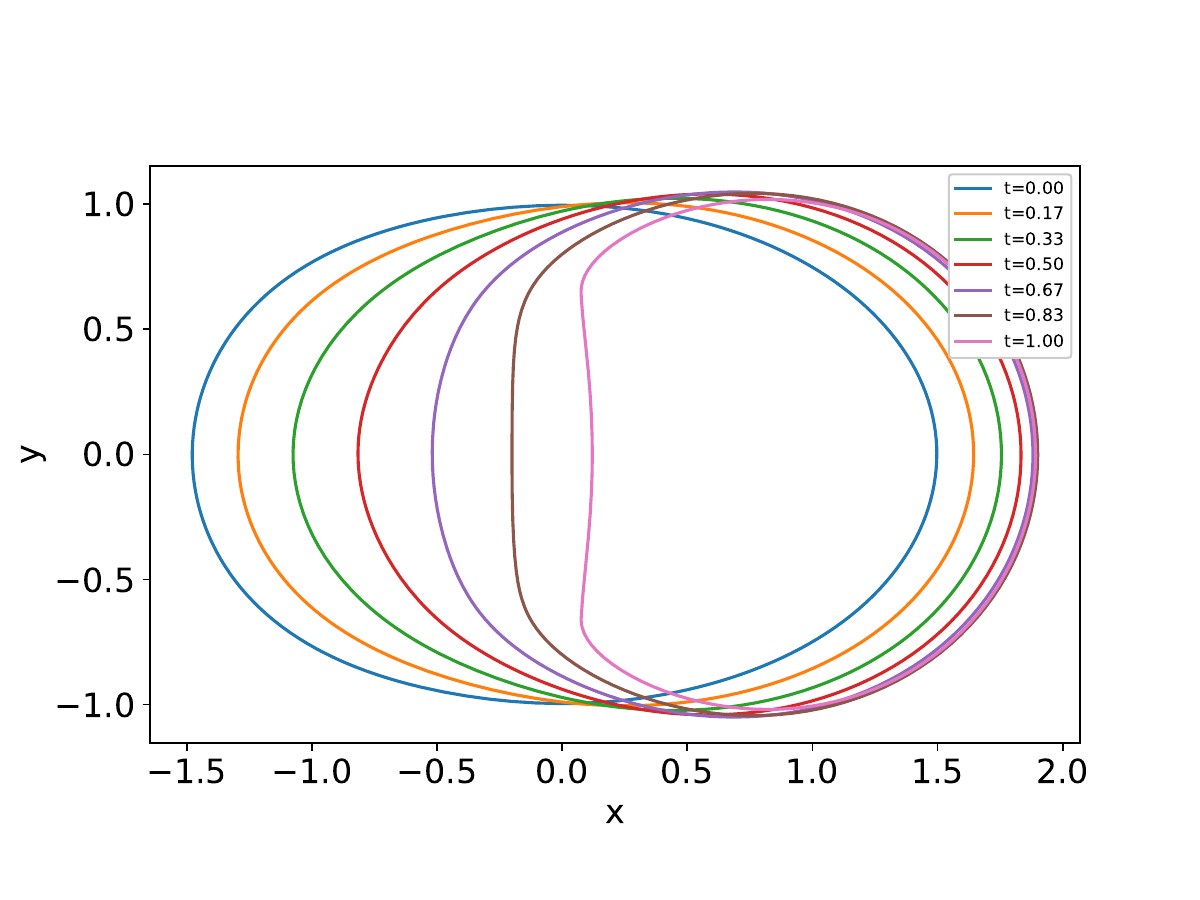} 
	\end{subfigure}
	
	{\captionsetup{font=small}
		\caption{\textit{HMCF starting from an ellipse. We show the evolution over time for $r_1=sinu$ (left) and $r_1=cosu$(right).}}
		\label{fig5}}
\end{figure}

We are also interested in numerical simulation with a nonconstant initial velocity. As shown in Figure \ref{fig5}(a), the evolution is driven by $r_1=sinu$, which introduces a pronounced top-bottom asymmetry. In the lower half of the ellipse ($\pi<u<2\pi$), $sinu$ is negative, providing an initial velocity directed inward. This velocity acts in synergy with the intrinsic inward acceleration, resulting in the bottom of the curve moving upward at a remarkable rate. In the upper half ($0<u<\pi$), the positive $sinu$ yields an outward initial velocity that opposes the inward acceleration from curvature. Consequently, the initial outward momentum dominates, causing the top of the curve to expand. At the lateral vertices of the ellipse, the initial velocity is zero. These points, however, possess the maximum curvature and thus experience the strongest inward acceleration. This results in a persistent shrinkage of the curve's horizontal dimension throughout the evolution.

Figure \ref{fig5}(b) illustrates the evolution for an initial velocity $r_1=cosu$, inducing a strong left-right asymmetry. The left half of the ellipse ($\frac{\pi}{2}<u<\frac{3\pi}{2}$) is characterized by an inward initial velocity ($cosu<0$), which is maximized at the leftmost vertex ($u=\pi$). At this vertex, the curvature is also maximal. The synergistic effect of the maximal initial inward velocity and the maximal inward acceleration results in a catastrophic inward collapse of the left side. The plots confirm that the displacement of the left boundary is substantially greater than that of the right. By the final observed time, $t=1$, the left portion of the curve has degenerated into a near-vertical line segment. In contrast, the right half of the ellipse ($-\frac{\pi}{2}<u<\frac{\pi}{2}$) has an initial outward velocity ($cosu>0$) that opposes the inward acceleration. In the early evolutionary stages, this outward momentum largely counteracts the shrinkage, preserving the curve's geometry. This stability is in stark contrast to the rapid collapse of the left side. As the evolution progresses, the persistent inward acceleration eventually overcomes this initial momentum, leading to a slight inward shrinkage of the right boundary. Finally, at the top and bottom vertices, where the initial velocity is zero, the motion is dictated solely by the local curvature. Consequently, the overall vertical dimension of the curve exhibits a monotonic decrease throughout the evolution.

\begin{figure}[htpb] 
	\centering 
	\begin{subfigure}[b]{0.31\linewidth} 
		\centering 
		\includegraphics[width=\linewidth]{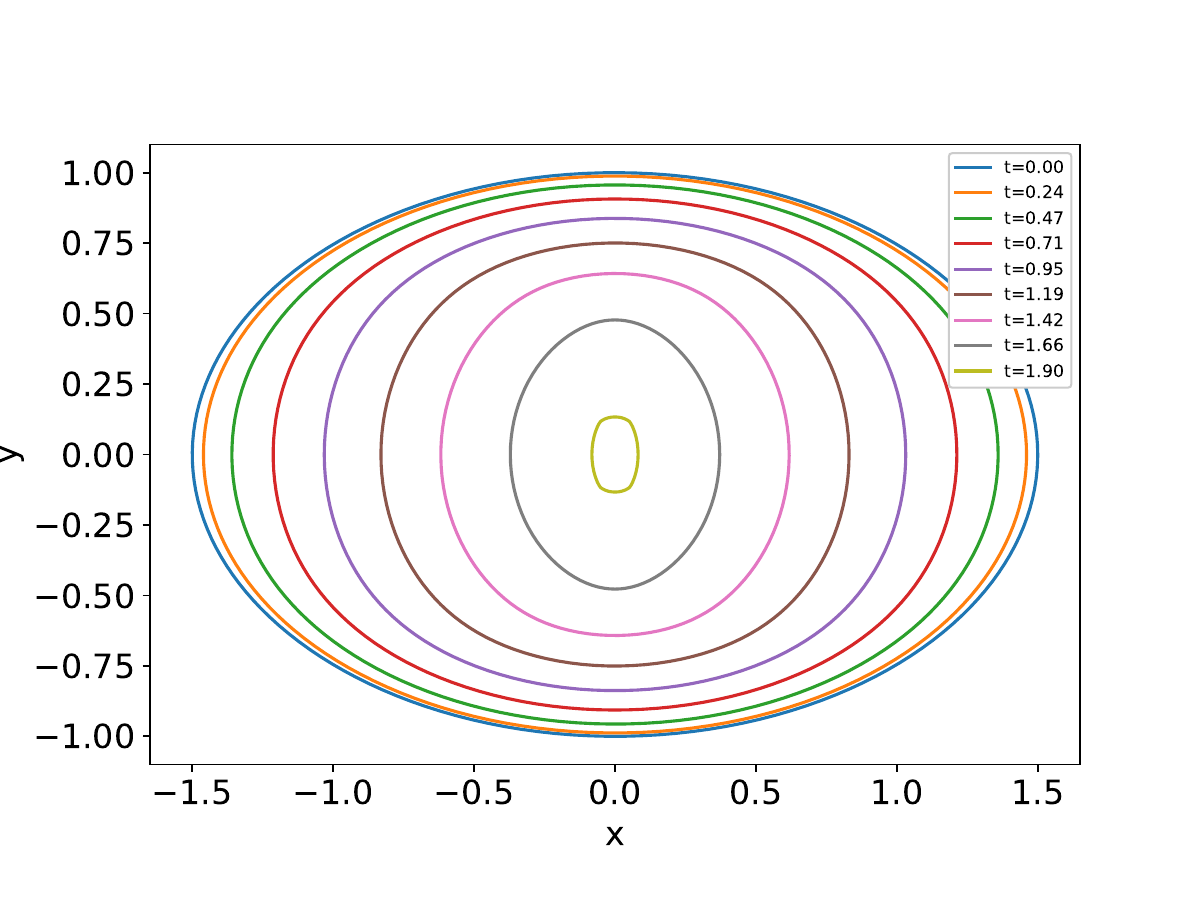}
		\caption{$r_1=0$}
	\end{subfigure}
	\hfill
	\begin{subfigure}[b]{0.31\linewidth}
		\centering
		\includegraphics[width=\linewidth]{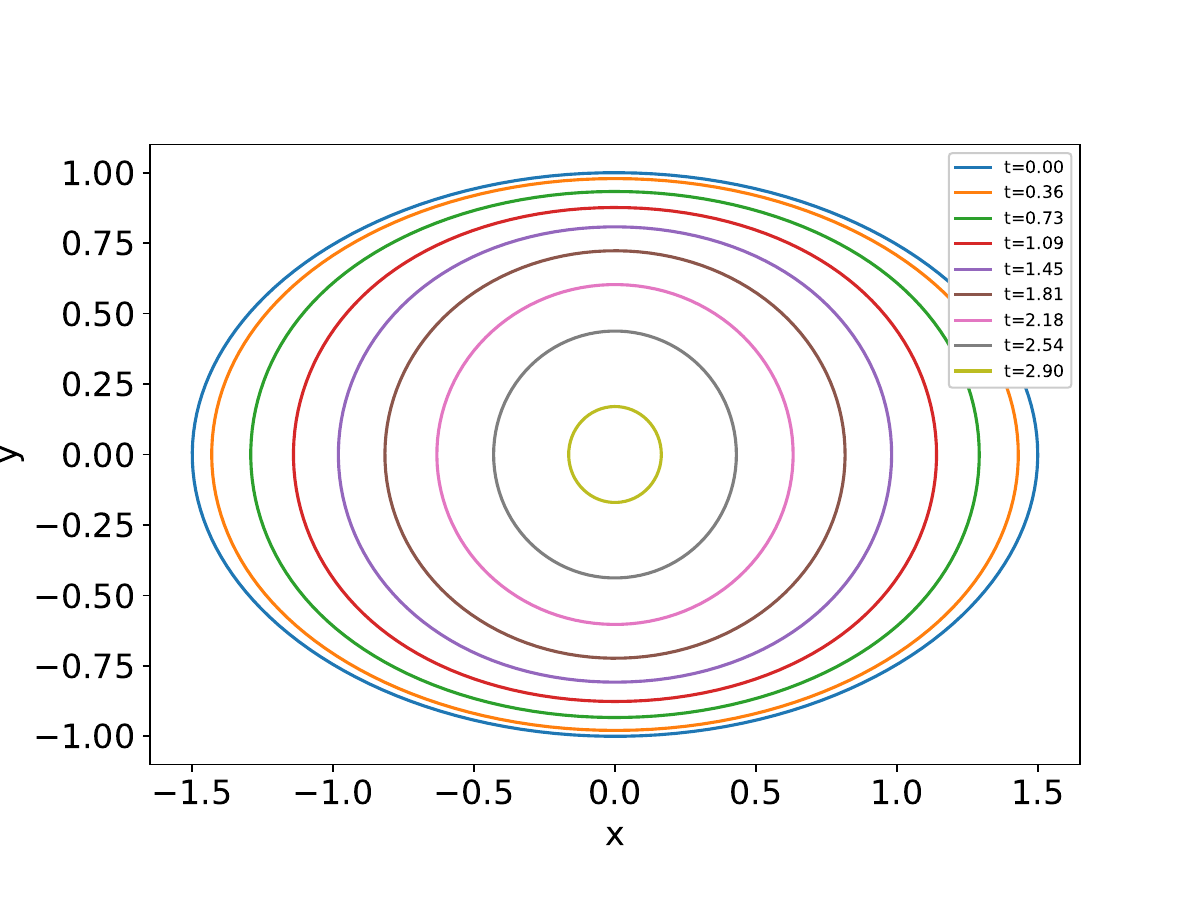} 
		\caption{$r_1=0$}
	\end{subfigure}
	\hfill
	\begin{subfigure}[b]{0.31\linewidth}
		\centering
		\includegraphics[width=\linewidth]{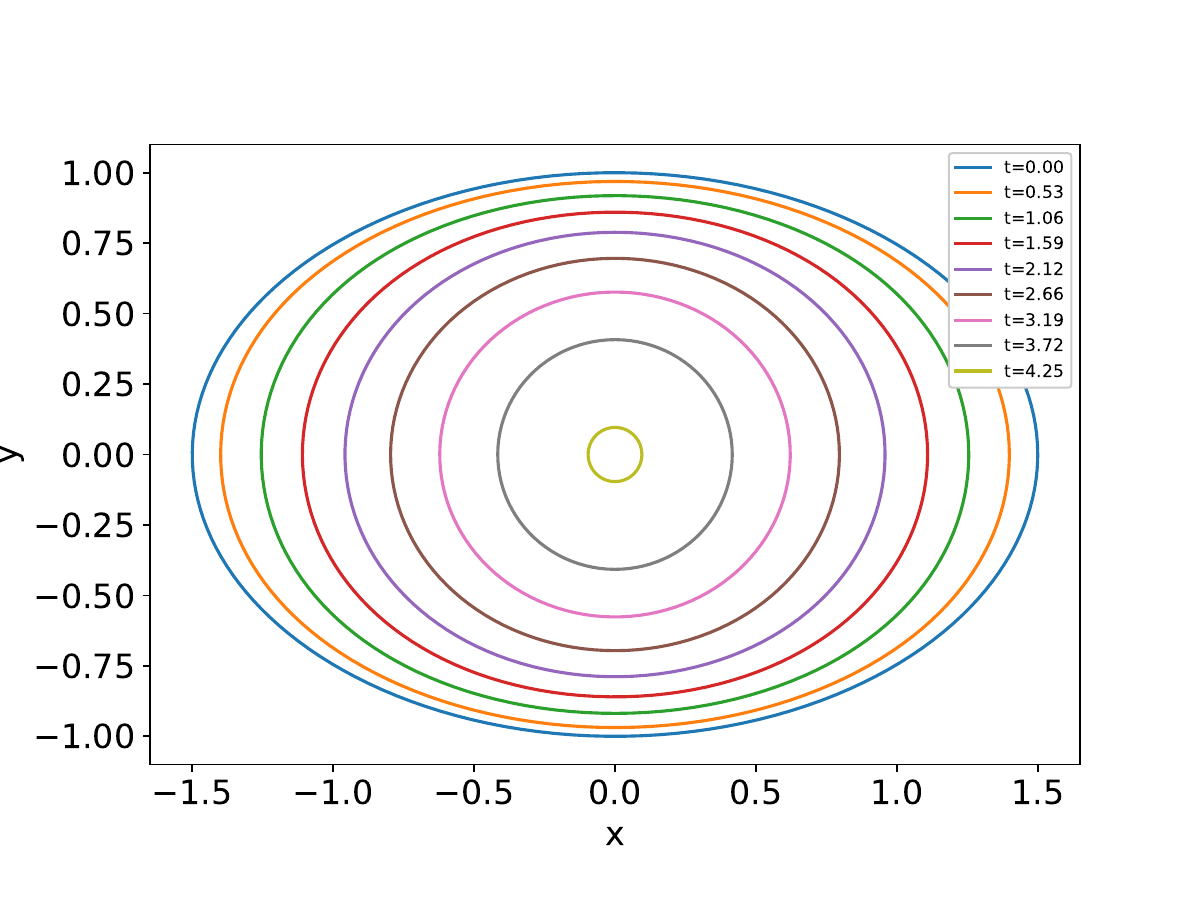}
		\caption{$r_1=0$} 
	\end{subfigure}
	\hfill
	\begin{subfigure}[b]{0.31\linewidth}
		\centering
		\includegraphics[width=\linewidth]{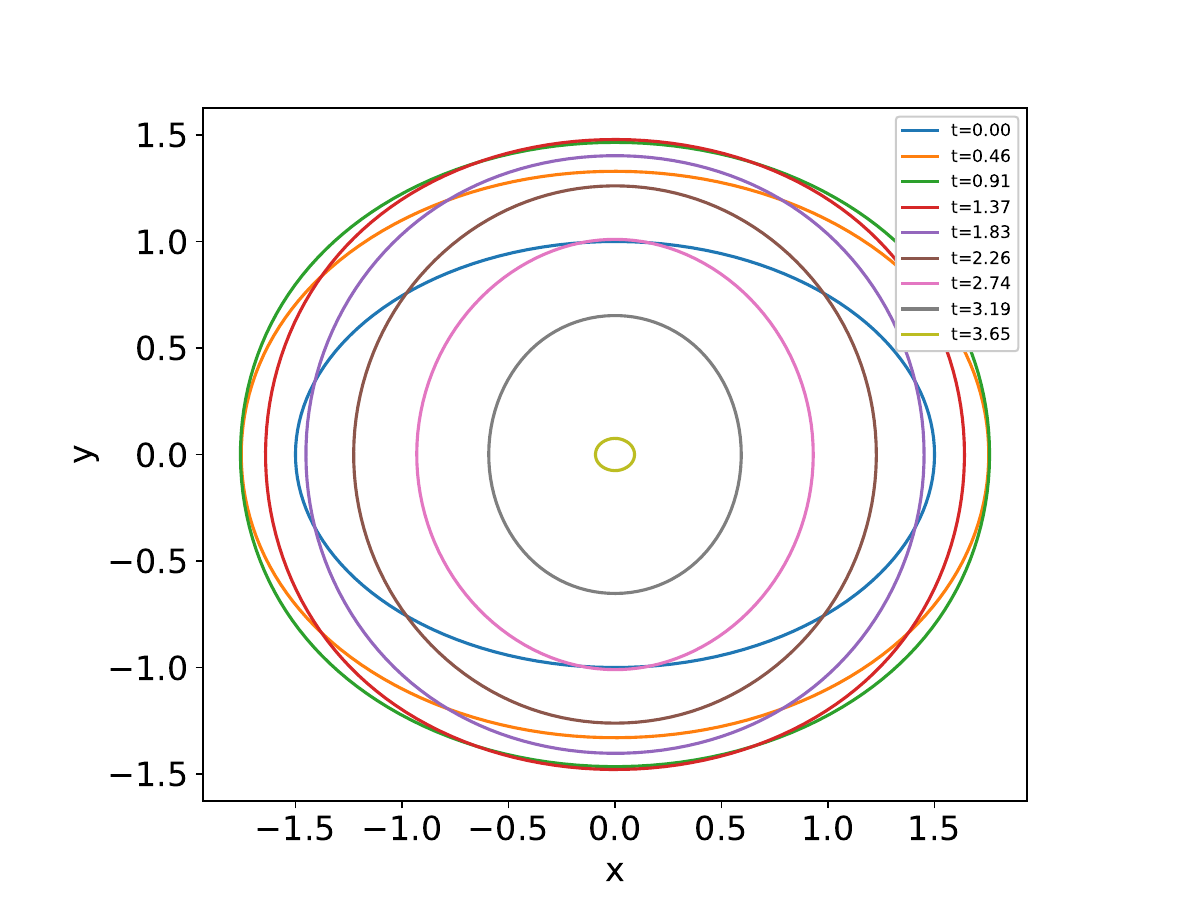} 
		\caption{$r_1=1$}
	\end{subfigure}
	\hfill
	\begin{subfigure}[b]{0.31\linewidth}
		\centering
		\includegraphics[width=\linewidth]{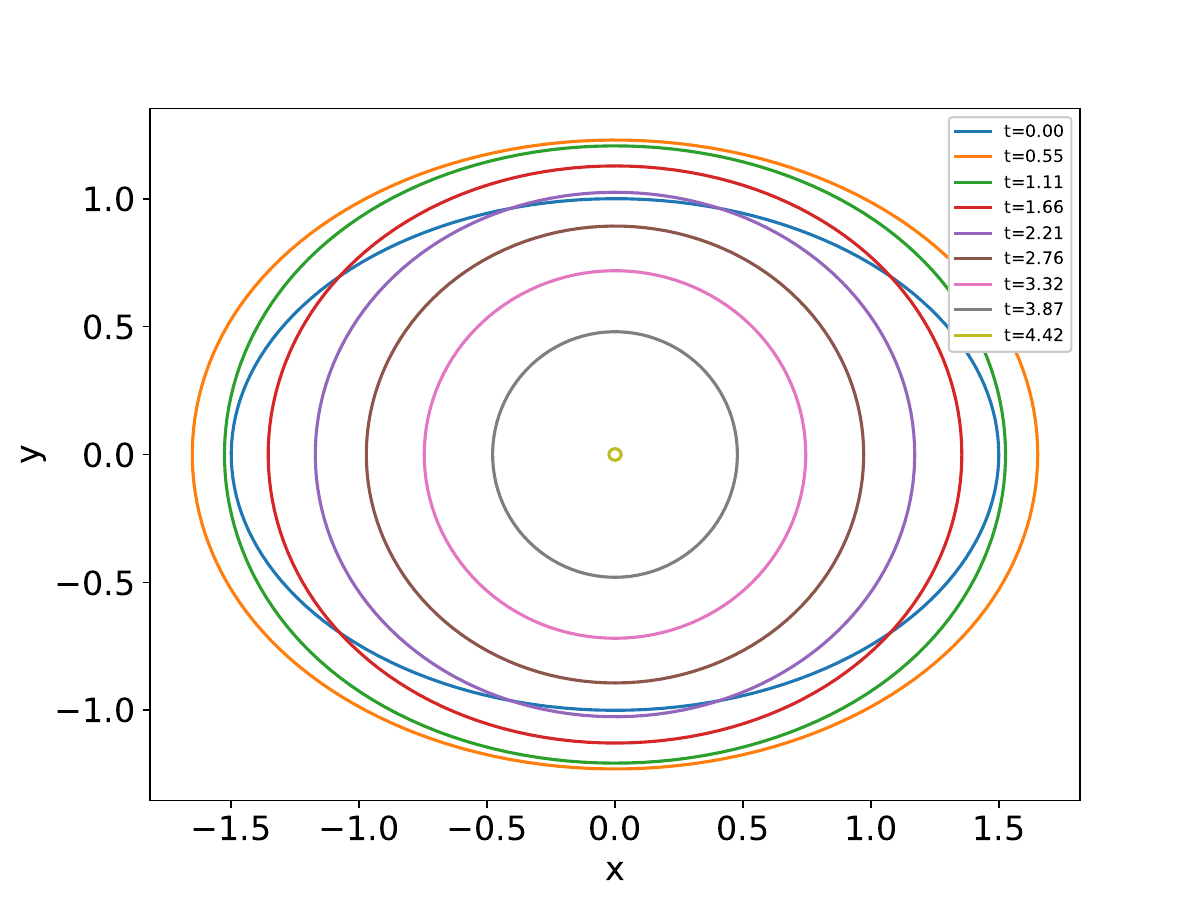}
		\caption{$r_1=1$} 
	\end{subfigure}
	\hfill
	\begin{subfigure}[b]{0.31\linewidth}
		\centering
		\includegraphics[width=\linewidth]{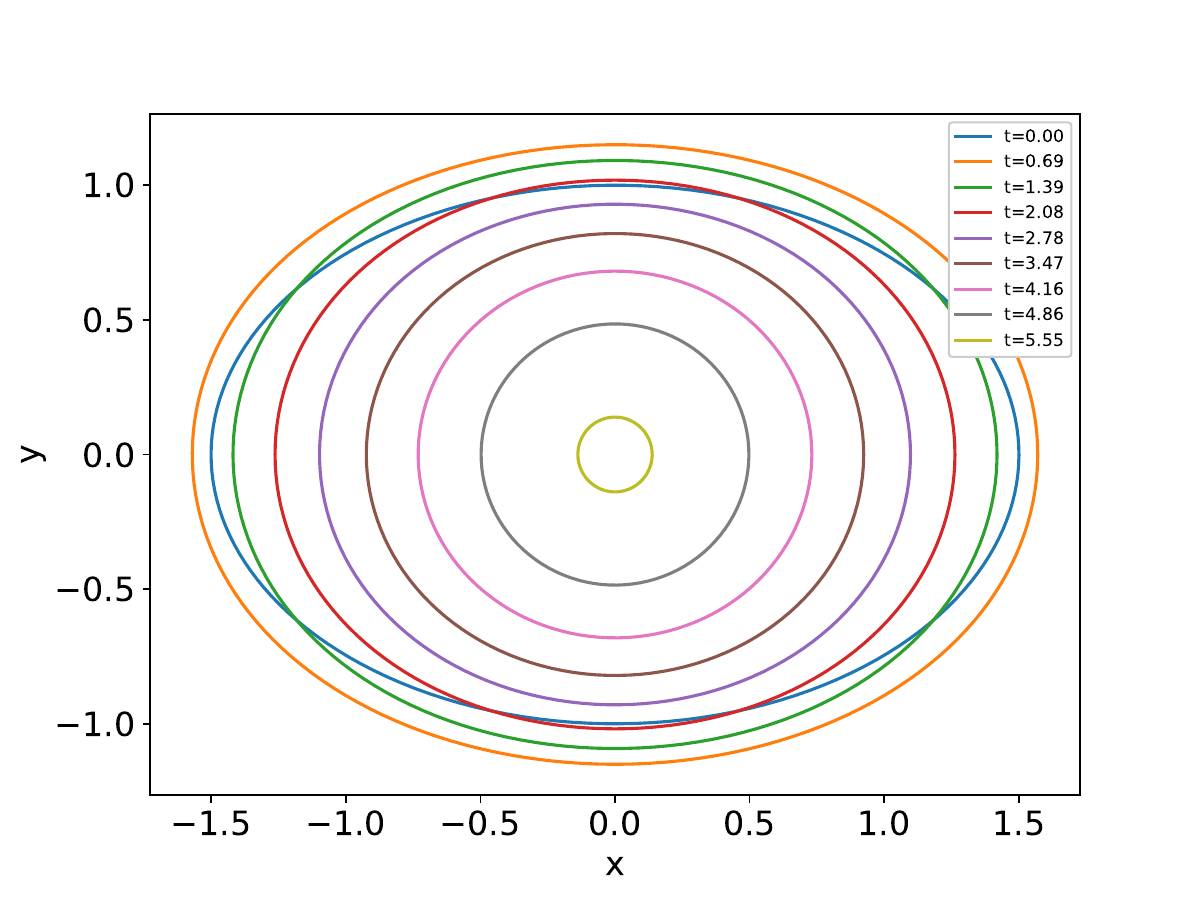}
		\caption{$r_1=1$}  
	\end{subfigure}
	\hfill
	\begin{subfigure}[b]{0.31\linewidth}
		\centering
		\includegraphics[width=\linewidth]{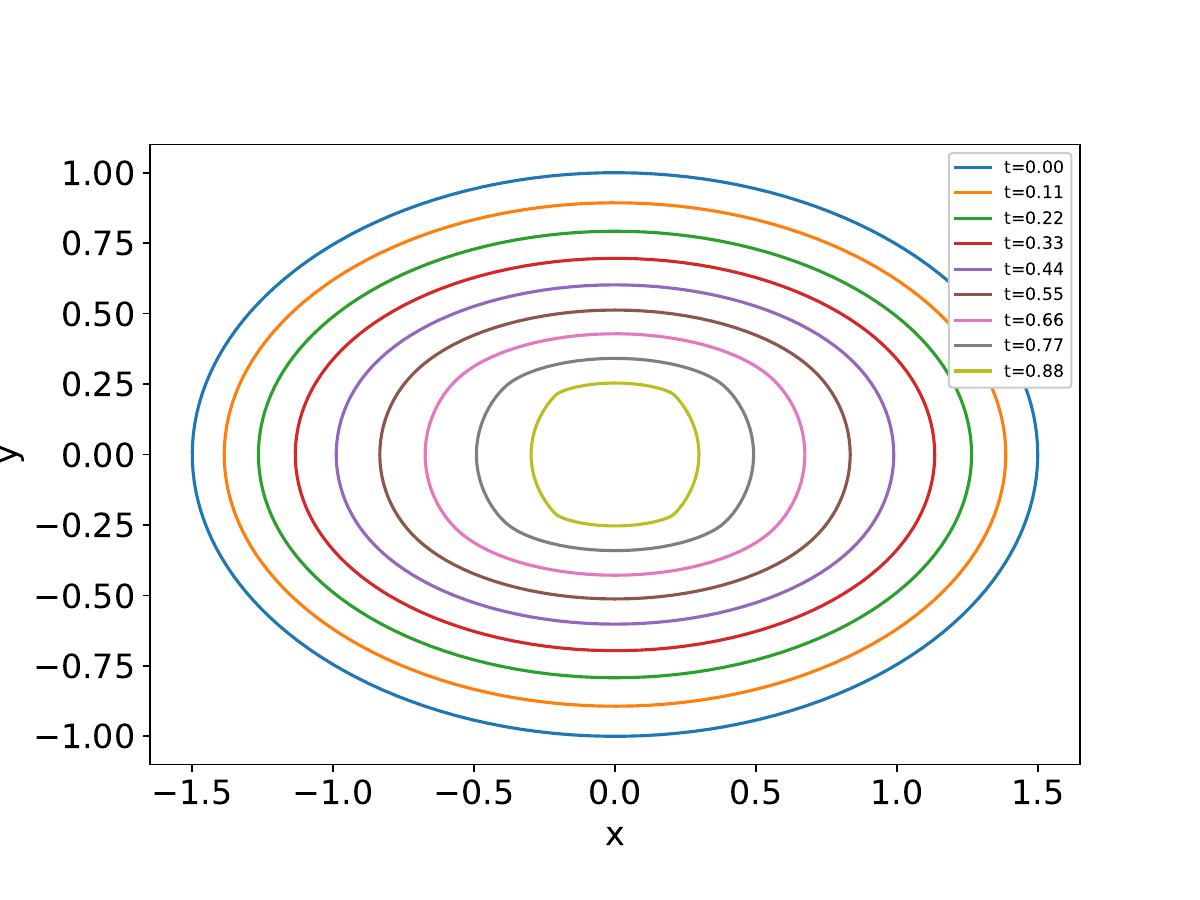} 
		\caption{$r_1=-1$} 
	\end{subfigure}
	\hfill
	\begin{subfigure}[b]{0.31\linewidth}
		\centering
		\includegraphics[width=\linewidth]{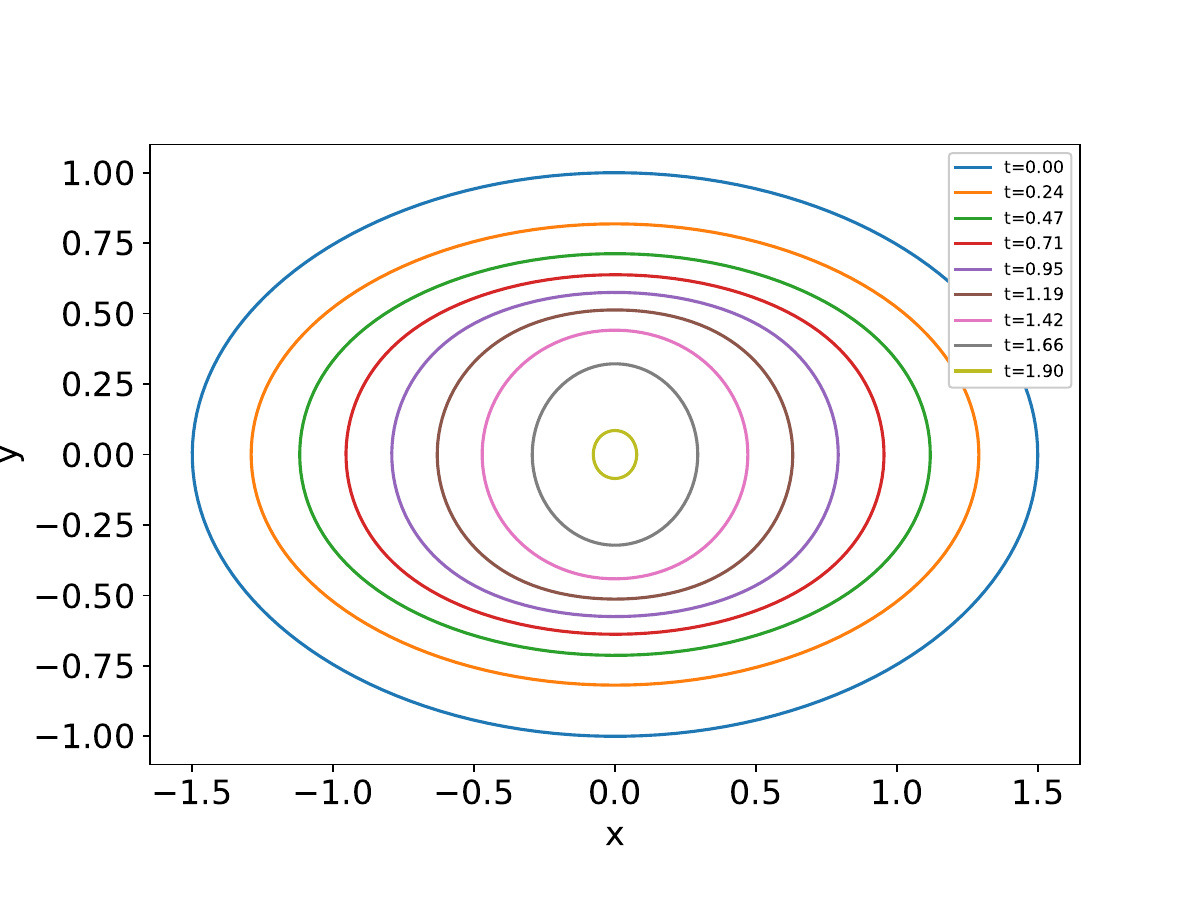} 
		\caption{$r_1=-1$} 
	\end{subfigure}
	\hfill
	\begin{subfigure}[b]{0.31\linewidth} 
		\centering 
		\includegraphics[width=\linewidth]{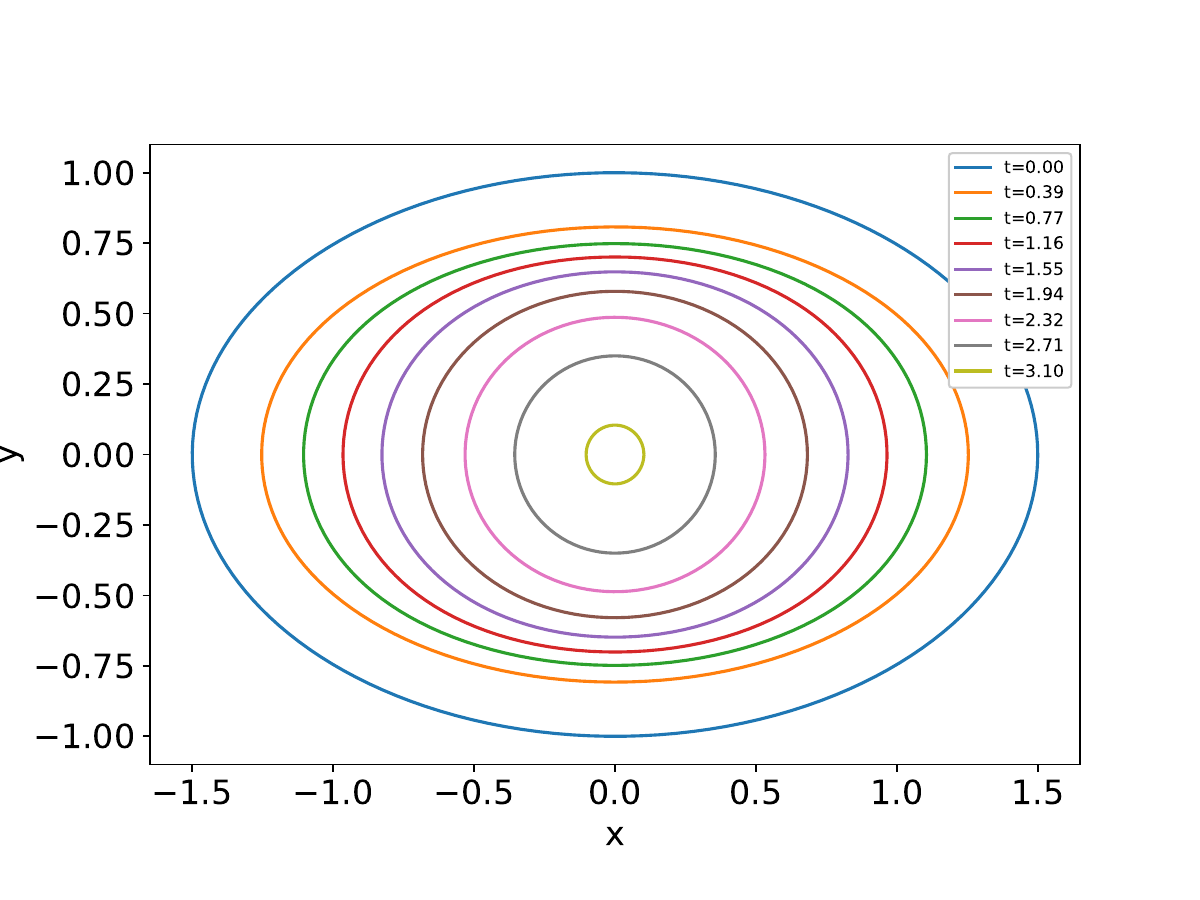}
		\caption{$r_1=-1$}
	\end{subfigure}
	
	{\captionsetup{font=small}
		\caption{\textit{HMCF starting from an ellipse with a constant initial velocity. The value of $\beta$ is 1 (first column), 3 (second column), and 5 (third column).}}
		\label{fig6}
}
\end{figure}

Figure \ref{fig6} investigates the effect of the dissipative coefficient $\beta$ in the evolution of an initial ellipse under HMCF,  by considering three distinct initial velocities $r_1$ for $\beta $ with values of 1, 3, and 5. For a zero initial velocity $r_1=0$ (first row), where the dynamics is purely curvature driven, a low dissipative coefficient $\beta = 1$ permits the characteristic hyperbolic axis-flipping instability, leading to an inversion of the ellipse's orientation. As $\beta$ increases, this axis-flipping phenomenon is gradually suppressed, promoting a more isotropic shrinkage towards a circular geometry.

In the case of an initial outward velocity $r_1 = 1$(second row), a low $\beta $ value facilitates a pronounced initial expansion, followed by the formation of four high curvature regions resulting from an asynchronous reversal of motion. High value of $\beta$ not only diminishes the magnitude of this initial expansion, but also ensures a smooth, convex evolution by inhibiting shock formation.

Conversely, for an initial inward velocity $r_1 = -1$(third row), where initial velocity and curvature driven acceleration are co-aligned, the evolution is consistently a rapid, smooth shrink. Here, $\beta$ plays purely a role as a decelerating force, systematically prolonging the collapse time as its value increases.

Overall, $\beta$ serves as a critical stabilizing and decelerating parameter, and it suppresses the characteristic hyperbolic instabilities of axis-flipping and shock formation while monotonically increasing the evolution time, causing the  dynamics of the system to approach the stable, diffusion-like behavior of a parabolic mean curvature flow for large $\beta$.

\begin{figure}[htpb] 
	\centering 
	\begin{subfigure}[b]{0.31\linewidth} 
		\centering 
		\includegraphics[width=\linewidth]{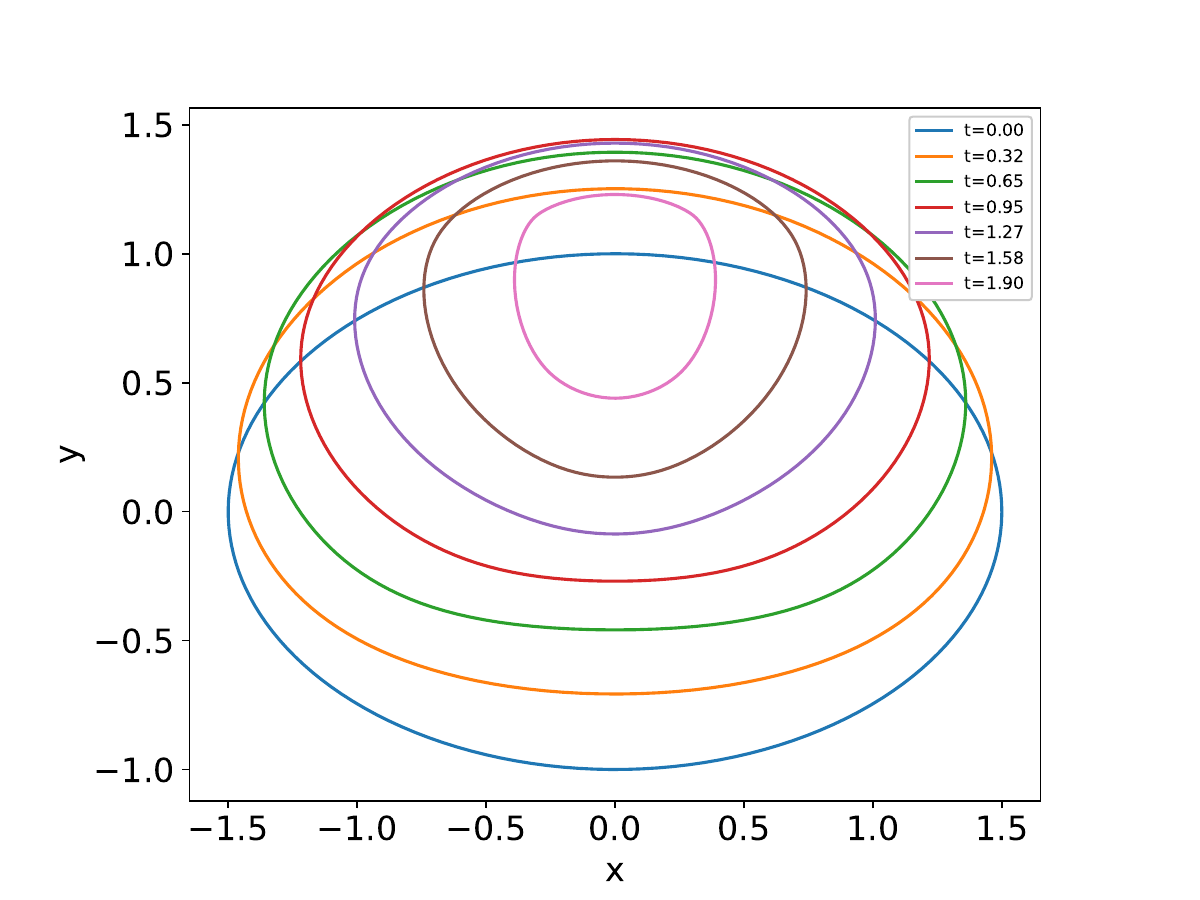}
		\caption{$r_1=sinu$}
	\end{subfigure}
	\hfill
	\begin{subfigure}[b]{0.31\linewidth}
		\centering
		\includegraphics[width=\linewidth]{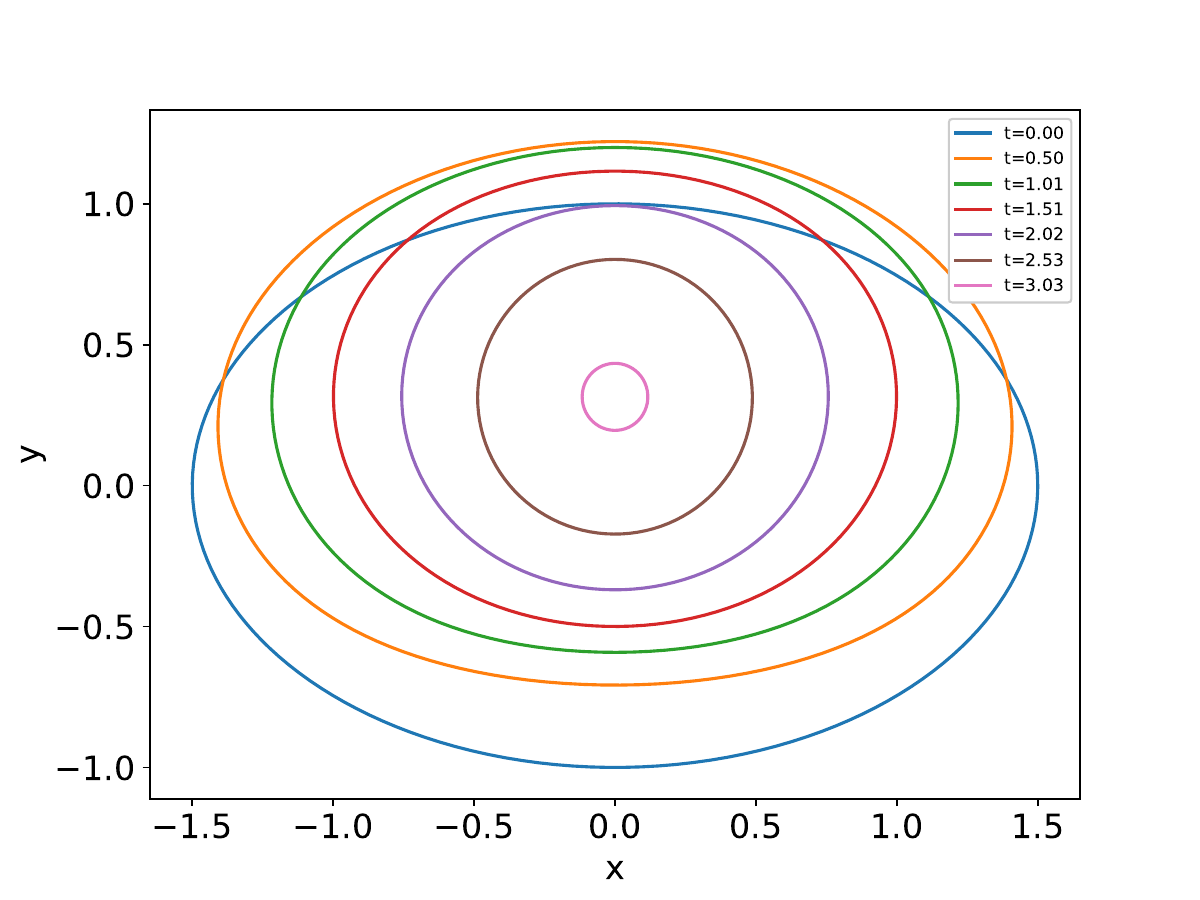} 
		\caption{$r_1=sinu$}
	\end{subfigure}
	\hfill
	\begin{subfigure}[b]{0.31\linewidth}
		\centering
		\includegraphics[width=\linewidth]{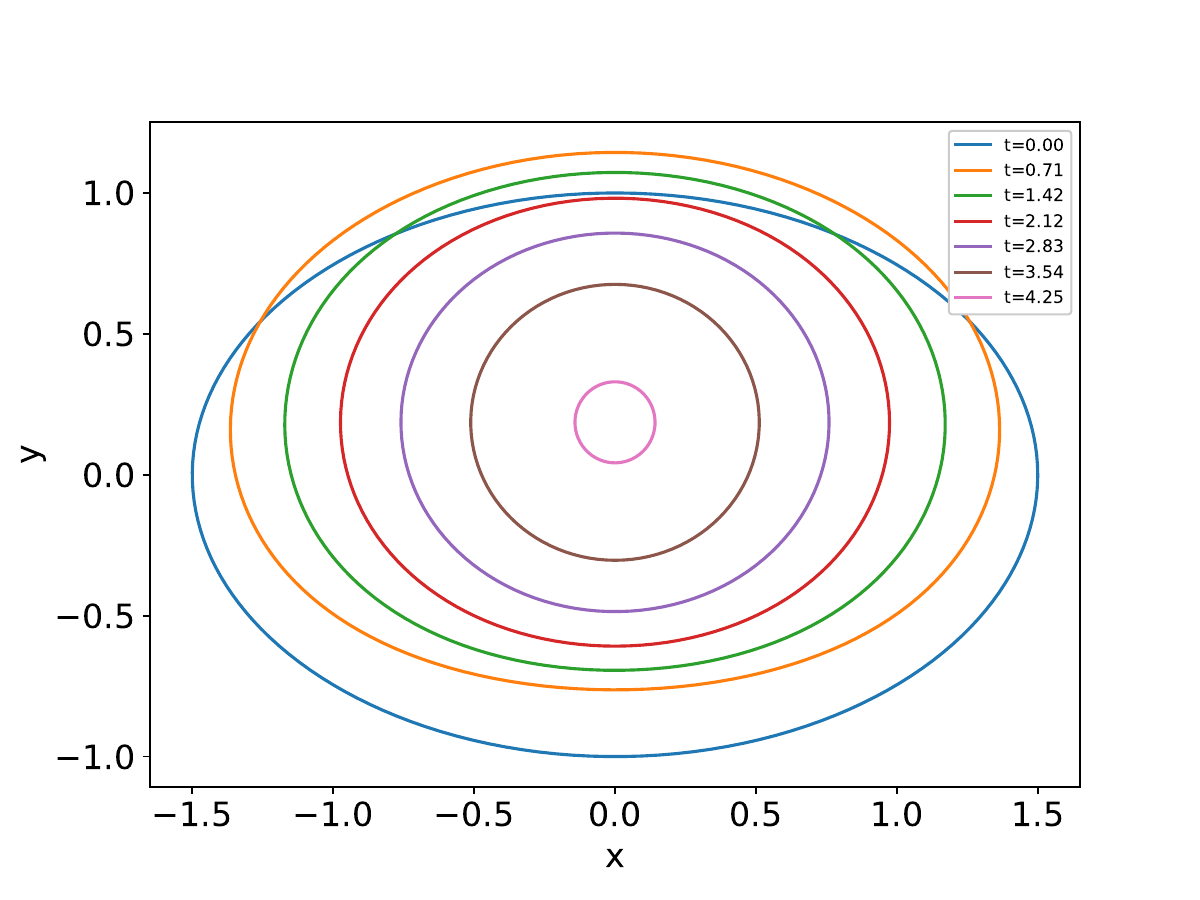}
		\caption{$r_1=sinu$} 
	\end{subfigure}
	\hfill
	\begin{subfigure}[b]{0.31\linewidth}
		\centering
		\includegraphics[width=\linewidth]{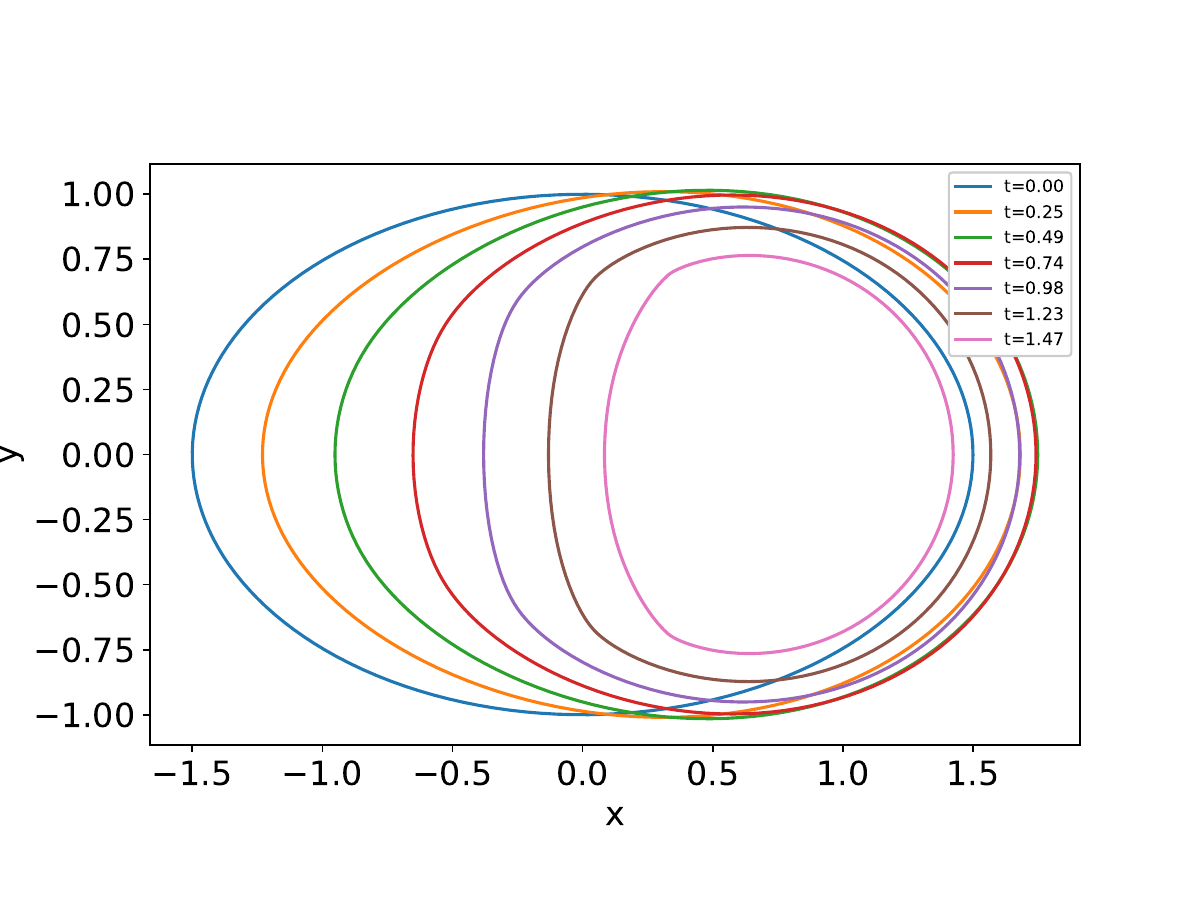} 
		\caption{$r_1=cosu$}
	\end{subfigure}
	\hfill
	\begin{subfigure}[b]{0.31\linewidth}
		\centering
		\includegraphics[width=\linewidth]{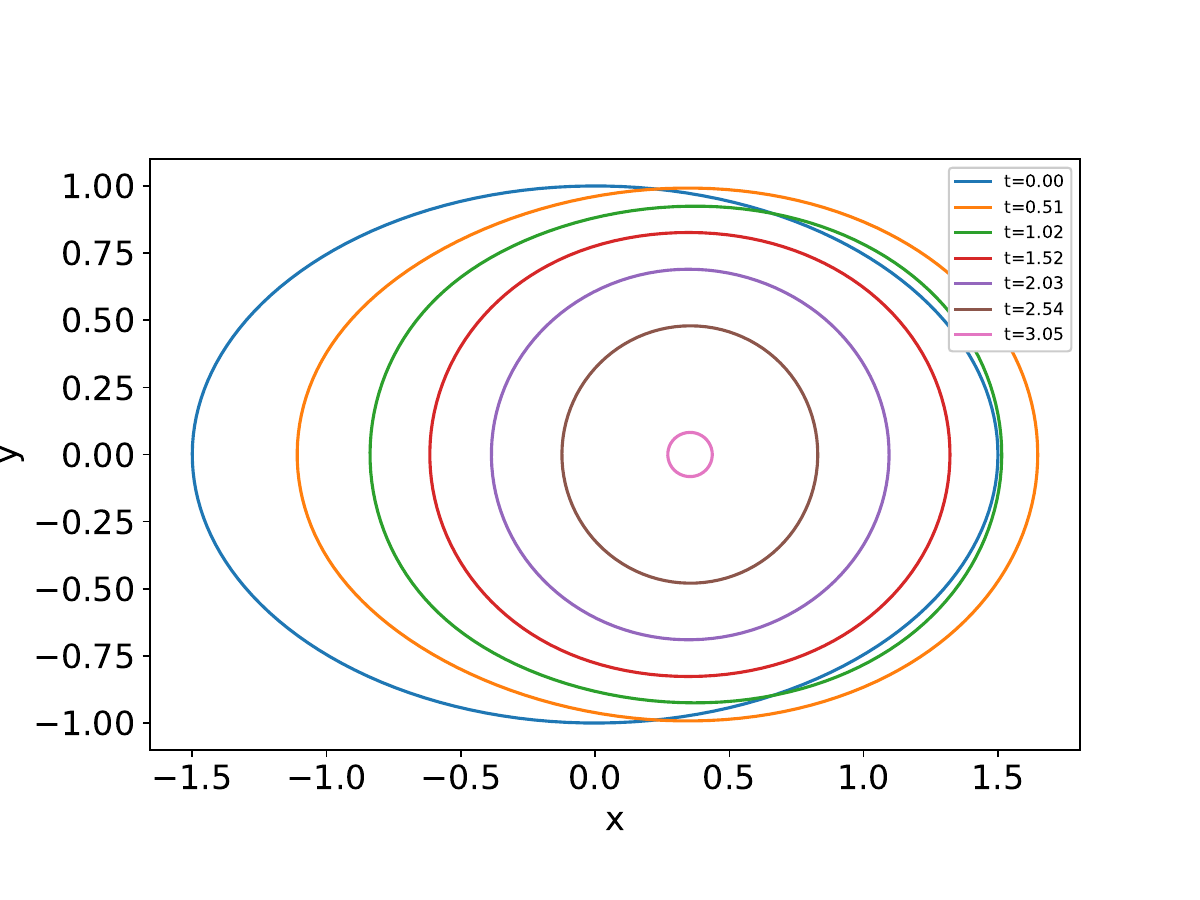}
		\caption{$r_1=cosu$} 
	\end{subfigure}
	\hfill
	\begin{subfigure}[b]{0.31\linewidth}
		\centering
		\includegraphics[width=\linewidth]{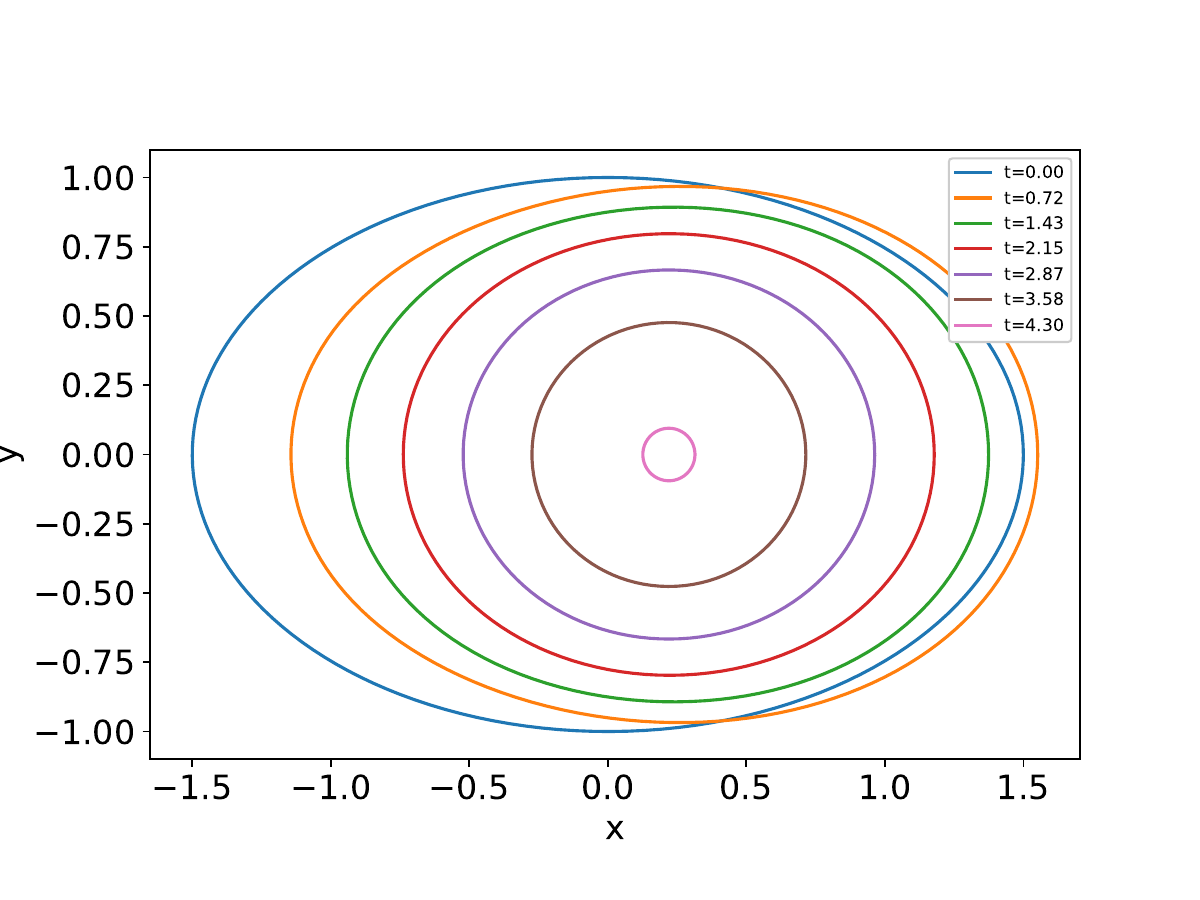}
		\caption{$r_1=cosu$}  
	\end{subfigure}

	{\captionsetup{font=small}
		\caption{\textit{HMCF starting from an ellipse with a nonconstant initial velocity. The value of $\beta$ is 1 (first column), 3 (second column), and 5 (third column).}}
		\label{fig7}
	}
\end{figure}

For cases where the initial velocity is nonconstant, Figure \ref{fig7} illustrates the pivotal role of the dissipative coefficient $\beta$ in mediating the transition from hyperbolic to parabolic behavior in the evolution of a curve. The first row corresponds to an initial velocity of $sinu$, and the second row has the initial velocity $cosu$, while the values of $\beta$ from left to right are 1, 3, and 5, respectively. For a low dissipative coefficient ($\beta=1$), the  hyperbolic features are prominent, retaining a strong memory of the initial velocity field. This is evident in the pronounced top-bottom asymmetry for $r_1=sinu$ and the severe left-right asymmetry for $r_1=cosu$, where a regularized, shock-like feature forms on the left side. As $\beta$ increases to an intermediate value ($\beta=3$), the dissipative effect becomes more significant, largely suppressing these initial asymmetries and ensuring a globally smooth evolution, preventing the formation of shock-like structures. At a high dissipative value ($\beta=5$), the parabolic nature of the flow is completely dominant. The influence of the initial velocity is rapidly dissipated, and for both $r_1=sinu$ and $r_1=cosu$, the evolution converges to that of the standard mean curvature flow, a smooth, symmetrical shrinkage towards a single point. In essence, $\beta$ acts as a control parameter, determining whether the system  is governed by its initial dynamics (hyperbolic regime) or its intrinsic geometry (parabolic regime).

\subsection{Hyperbolic mean curvature flow for surfaces}
We begin with a convergence experiment for the evolution of a sphere when $\beta=0$, in order to validate our proposed method.  The initial surface is given by parametrization
\begin{equation}\label{eq43}
	X_0(u_1, u_2) :=
	\begin{pmatrix}
		r_0 \sin u_1 \cos u_2 \\
		r_0 \sin u_1 \sin u_2 \\
		r_0 \cos u_1
	\end{pmatrix}, \quad u_1 \in [0, \pi], \quad u_2 \in [0, 2\pi],
\end{equation}
and the initial velocity is
\begin{equation}\label{eq44}
	X_1(u_1, u_2) =-r_1\overrightarrow{N}_0 = r_1(\sin u_1 \cos u_2, \sin u_1 \sin u_2, \cos u_1),
\end{equation}
which $r_0 \in \mathbb{R}_+$, $r_1 \in \mathbb{R}$, $\overrightarrow{N}_0$ the inner normal vector of the initial sphere.

Eq. (\ref{eq23}) for the evolving sphere can be rewritten as a second order differential equation for the time-dependent radius $r=r(t)$:
\begin{equation}\label{eq45}\left\{\begin{array}{ll}
		r_{tt} = -\frac{2}{r} \quad in \quad (0,T), \\[4mm]
		r(0) = r_0 ,\\[4mm]
		r_t(0) = r_1.
	\end{array}\right.
\end{equation}
We can easily derive the analytical solution for the initial value problem (\ref{eq45}). Indeed, we have

\begin{lemma}\label{lemma2}
	The radially symmetric analytical solution of  the initial question (\ref{eq45}) for $r_1=0$ is given by 
	\begin{equation*}
		r(t) = r_0 exp\left(-\left(erf^{-1}\left(t\sqrt{4/r_0^2\pi}\right)\right)^2\right), \quad for \quad t\in(0,T), \quad T = r_0\sqrt{\pi}/2.
	\end{equation*}
	Whenever $r_1>0$, the solution is
	\begin{equation*}
		r(t) = r_0 e^\frac{r_1^2}{4}exp\left(-\left[erf^{-1}\left(-t e^{-\frac{r_1^2}{4}}\sqrt{4/r_0^2\pi}+erf(\frac{r_1}{\sqrt{2}})\right)\right]^2\right)
	\end{equation*}
	for $t\in[0,T_s]$, where
	\begin{equation*}
		T_s = \sqrt{\frac{\pi}{2}}r_0 e^\frac{r_1^2}{4}erf\left(\frac{r_1}{\sqrt{2}}\right).
	\end{equation*}
	Whenever $t\in[T_s,T)$, the solution is given as the zero velocity solution with the initial radius equal to $r(T_s)=r_0e^\frac{r_1^2}{4}$.
\end{lemma}

We obtain PINNs solutions for $r_1=0$ on the time domain $[0, 0.8]$ with the training data spanning through the time interval $[0, 0.7]$. Similarly, for $r_1=1$ solutions were obtained on $[0, 1.65]$ with training data spanning on $[0, 1.5]$. Throughout the remainder of this subsection, we set $N_0=N_b=200$, and $N_f=20000$ points with temporal and spatial coordinates sampled on each domain and boundary. All sampled points are generated using the Latin Hypercube Sampling strategy. Adam optimizer with a learning rate of $1e-3$ is employed to initialize and optimize the learnable variables of the neural network. In terms of network architecture, a depth of 6 and a width of 100 are used unless otherwise specified. The training scheme is the hybrid optimization strategy as discussed in Section \ref{sec3}.

\begin{figure}[htpb] 
	\centering 
	\begin{subfigure}[b]{0.45\linewidth} 
		\centering 
		\includegraphics[width=\linewidth]{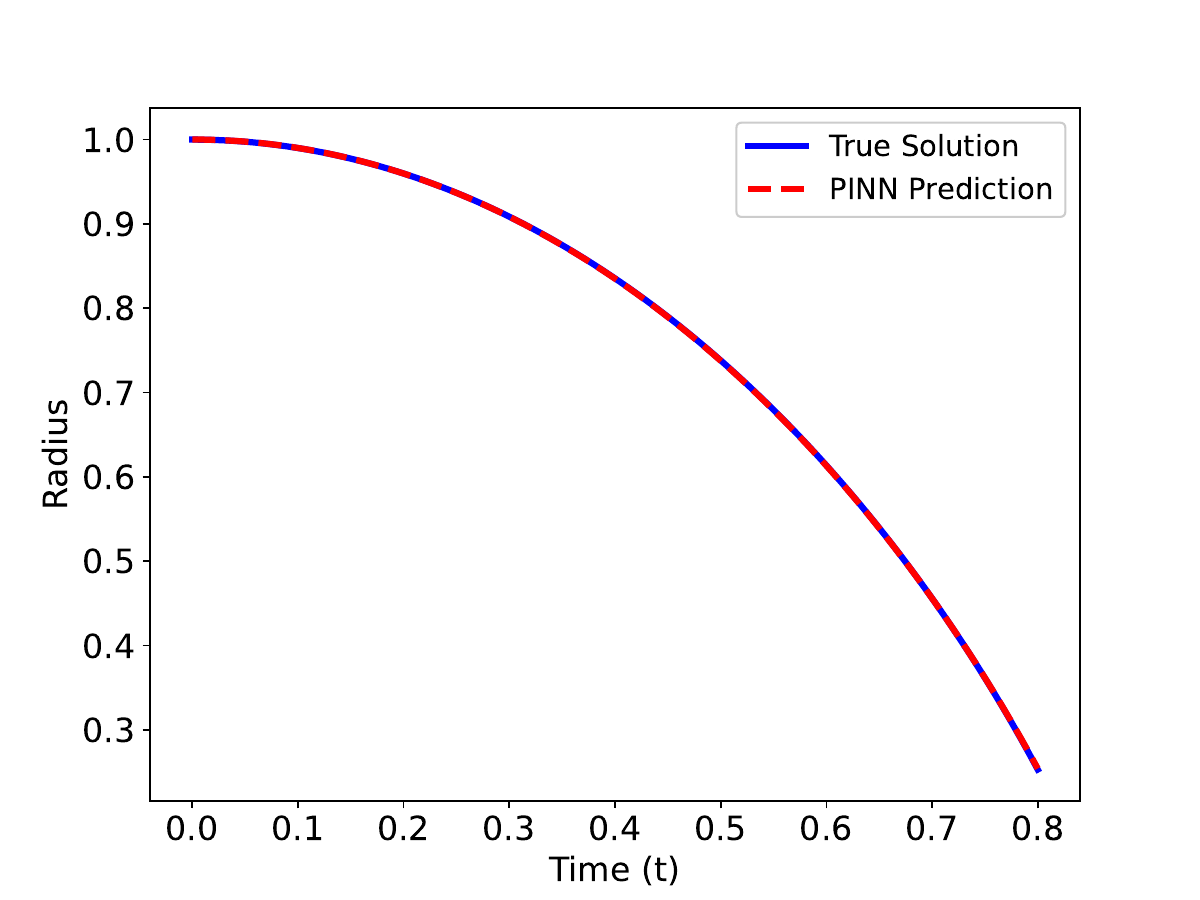}
	\end{subfigure}
	\hfill
	\begin{subfigure}[b]{0.45\linewidth}
		\centering
		\includegraphics[width=\linewidth]{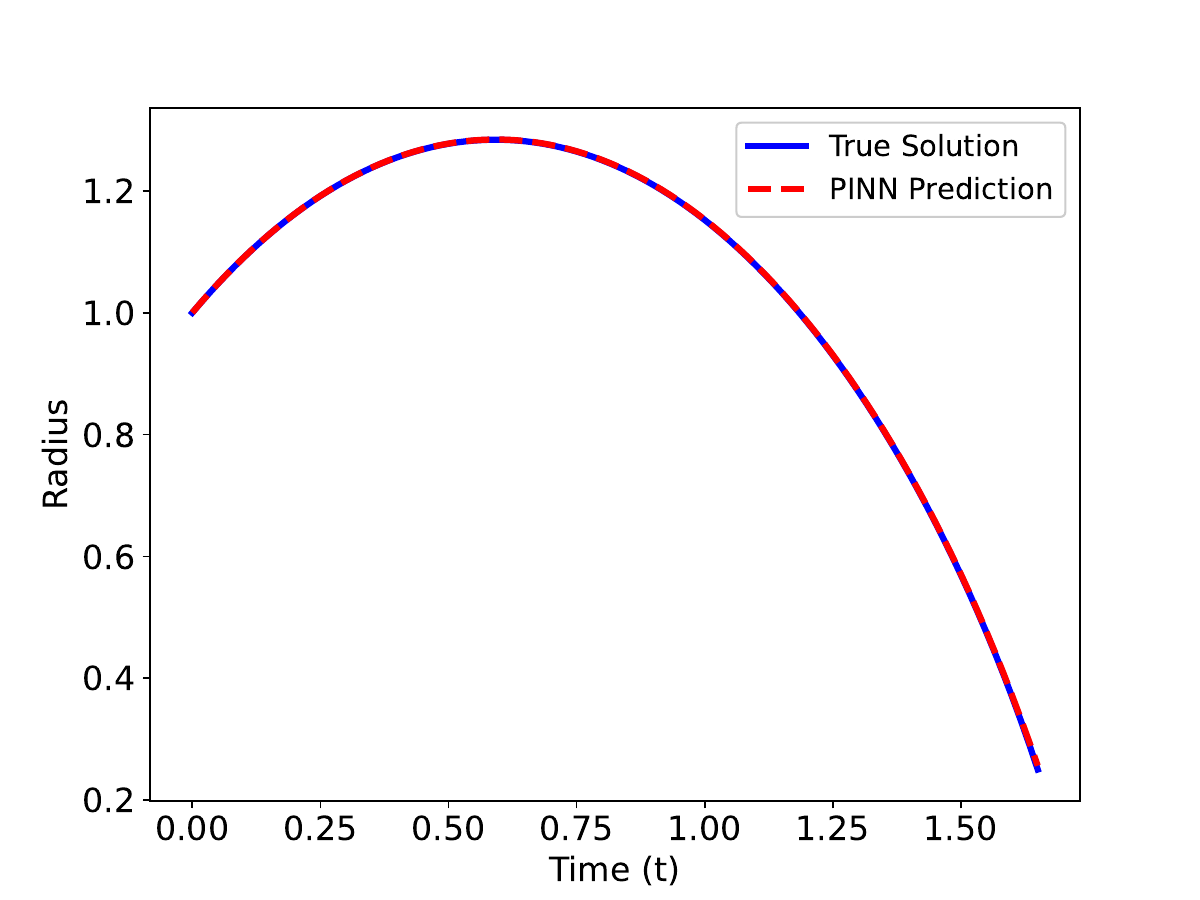} 
	\end{subfigure}
	
	{\captionsetup{font=small}
		\caption{\textit{Comparison of the PINNs and analytical solution of the HMCF starting from a unit sphere. Example with
		$r_1 = 0$ on the left, and example with $r_1 = 1$ on the right. The relative $\mathbb{L}_2$ errors are 3.53e-4 on the left, and 9.74e-4 on the right. The size of deep neural network is $6\times100$.}}
		\label{fig8}
	}
\end{figure}

Figure \ref{fig8} visually displays the variation in radius over time. For an initial unit sphere, depending on the sign of $r_1$, the family of spheres either expands at first and then shrinks, or shrinks immediately. Here, we illustrate the comparison between the PINNs solution and analytical solution given by Lemma 4.2, the relative $\mathbb{L}_2$ error is 3.53e-4 for $r_1=0$, and 9.74e-4 for $r_1=1$.

\begin{figure}[htpb] 
	\centering 
	\begin{subfigure}[b]{0.45\linewidth} 
		\centering 
		\includegraphics[width=\linewidth]{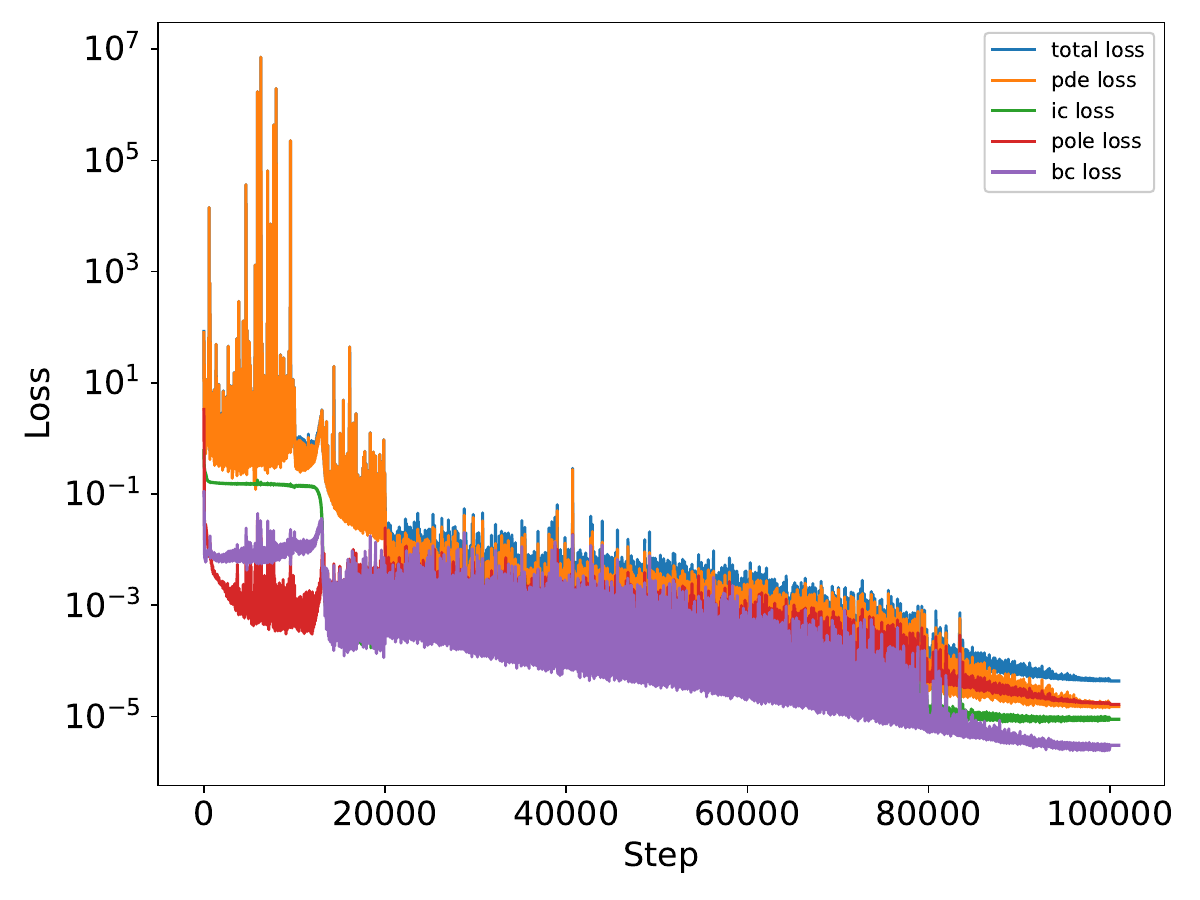}
		\caption{$r_1=0$}
	\end{subfigure}
	\hfill
	\begin{subfigure}[b]{0.45\linewidth}
		\centering
		\includegraphics[width=\linewidth]{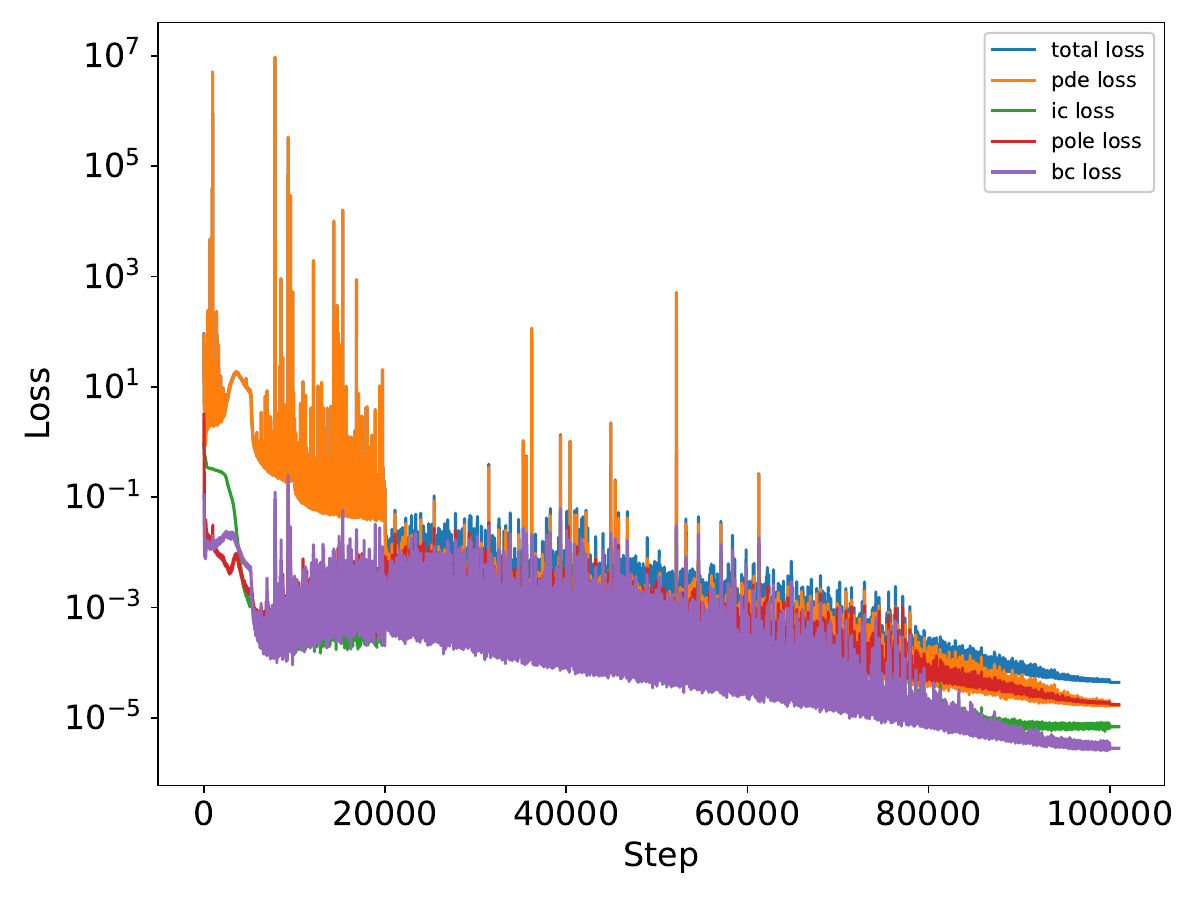} 
		\caption{$r_1=1$}
	\end{subfigure}
	
	{\captionsetup{font=small}
		\caption{\textit{Loss functions for the HMCF with (a) $r_1=0$ and (b) $r_1=1$. The size of deep neural network is $6\times100$.}}
		\label{fig9}
	}
\end{figure}

The success of the proposed model greatly benefits from the carefully designed optimization strategy, which incorporates the dynamic weighting scheme and two phases training process. In Figure \ref{fig9}, we present the evolution of the total loss and its individual components throughout the training steps. During the first 10,000 steps, the weights $w_0^s$, $w_b^s$ and $w_p^s$ are set to higher values, which results in a steep initial drop of ic loss($\mathcal{L}_0^s(\theta_s)$), pole loss ($\mathcal{L}_p^s(\theta_s)$), and bc loss ($\mathcal{L}_b^s(\theta_s)$) components, rapidly bringing them down to values below $1e-3$. This ensures the network learns to satisfy the initial and boundary conditions more earlier in the training. With decreasing in $w_0^s$, $w_b^s$ and $w_p^s$, the pde loss($\mathcal{L}_f^s(\theta_s)$) gains more influence in the total loss($\mathcal{L}_{surface}(\theta_s)$), leading to its continued significant descent. This indicates that the network's increasing ability to capture the underlying physics of the HMCF. After 20,000 steps, all weights are set to unity. By this stage, the network has sufficiently learned the initial, boundary, and pole conditions. This allows the optimization to focus on minimizing the overall physical residual, with all components contributing equally. As shown in Figure \ref{fig9}, during this period, the pde loss($\mathcal{L}_f^s(\theta_s)$) demonstrates a dramatic reduction, dropping from approximately  $1e-1$ to $1e-5$, alongside the continued, albeit more gradual, decrease in other loss components. This behavior demonstrates the effectiveness of our dynamic weighting strategy, as training initially prioritizes robust satisfaction of the matching conditions while progressively optimizing the equation residual.

During the second stage of training, characterized by fine-tuning with the L-BFGS optimizer. After training with the Adam optimizer for 100,000 steps, we further trained with the L-BFGS optimizer for 500 steps. As prominently displayed in Figure \ref{fig9} from step 100,000 onwards, the total loss, pde loss, ic loss, pole loss, and bc loss exhibit a marginal reduction, suggesting a trend toward convergence.

\begin{table}[h!]
	\centering
	\small
	\begin{threeparttable}
		\begin{tabular}{c|ccccc}
			\multicolumn{6}{c}{\textnormal{(a) $r_1=0$}} \\[1mm]
			\Xhline{0.7pt}
			& $4\times100$ & $4\times150$ & $6\times50$ & $6\times100$ & $6\times150$ \\
			\hline
			Relative $\mathbb{L}_2$ error & 6.70e-4 & 5.13e-4 & 1.39e-3 & 3.53e-4 & 4.46e-4\\
			\Xhline{0.7pt}
		\end{tabular}
		\vspace{2mm}
		\begin{tabular}{c|ccccc}
			\multicolumn{6}{c}{\textnormal{(b) $r_1=1$}} \\[1mm]
			\Xhline{0.7pt}
			& $4\times100$ & $4\times150$ & $6\times50$ & $6\times100$ & $6\times150$ \\
			\hline
			Relative $\mathbb{L}_2$ error & 2.33e-3 & 1.96e-3 & 5.22e-3 & 9.74e-4 & 1.00e-3\\
			\Xhline{0.7pt}
		\end{tabular}
		\caption{\small HMCF: Relative $\mathbb{L}_2$ error between PINNs and analytical solution for different numbers of hidden layers and different numbers of neurons per layer. Here, the total number of training and collocation points is fixed to $N_0 = N_b = 200$ and $N_f = 20,000$, respectively.}
		\label{tab4}
	\end{threeparttable}
\end{table}

Next, we turn to investigate the influence of different neural network architectures on the model performance.
In this case, we maintain the same hyperparameters as those outlined in the previous simulation problem. In Table \ref{tab4} we report the relative $\mathbb{L}_2$ error for varying numbers of hidden layers and neurons per layer. It shows that the error can be generally reduced by increasing both the number of hidden layers and the width of the layers, which indicates that our physics-informed constraints on the residual of PDEs can effectively regularize the training process and safeguard against over-fitting. However, it is worth highlighting that although deeper and wider networks provide further performance gains, these improvements become less pronounced with increased size, while simultaneously increasing computational cost. For the remainder of the subsection, we will utilize an architecture consisting of 6 layers with 100 neurons per layer, chosen for its balance between accuracy and computational efficiency.

\begin{figure}[htbp]
	\centering
	\setlength{\tabcolsep}{2pt}
	\renewcommand{\arraystretch}{1.0}
	
	\begin{tabular}{c c c c}
		& (a) $\omega_0^s=\omega_b^s=\omega_p^s=1000$  & (b) $\omega_0^s=\omega_b^s=\omega_p^s=100$ & (c) $\omega_0^s=\omega_b^s=\omega_p^s=1$  \\[2pt]
		\raisebox{1.1cm}[0pt][0pt]{$t = 0.0$} &
		\includegraphics[width=0.28\linewidth]{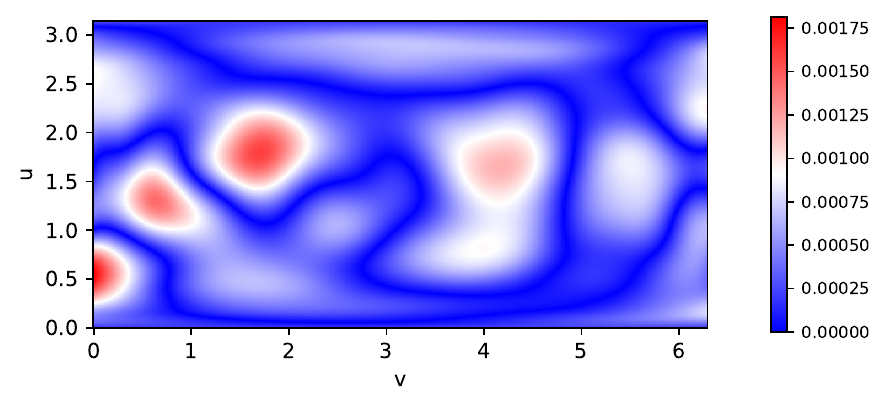} &
		\includegraphics[width=0.28\linewidth]{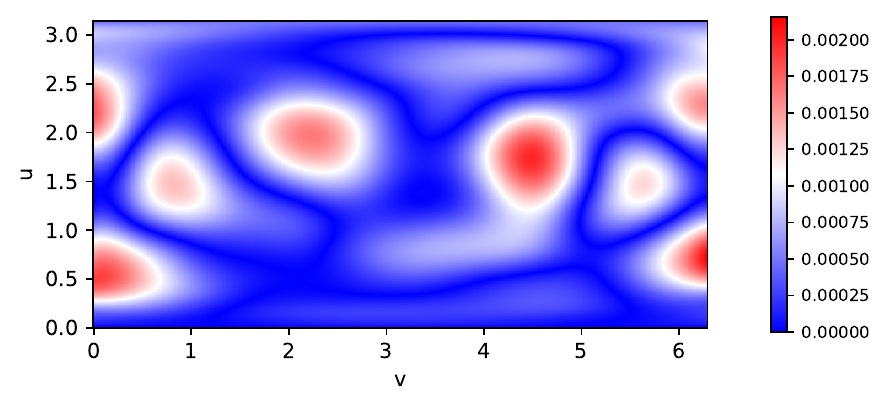} &
		\includegraphics[width=0.28\linewidth]{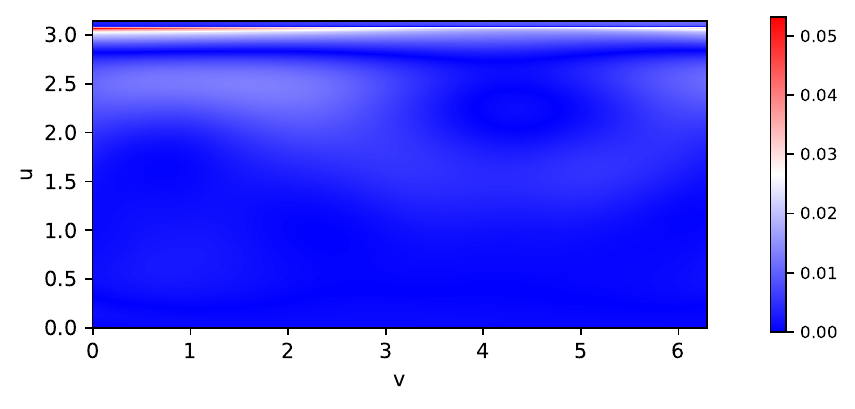} \\[6pt]
		
		\raisebox{1.1cm}[0pt][0pt]{$t = 1.0$} &
		\includegraphics[width=0.28\linewidth]{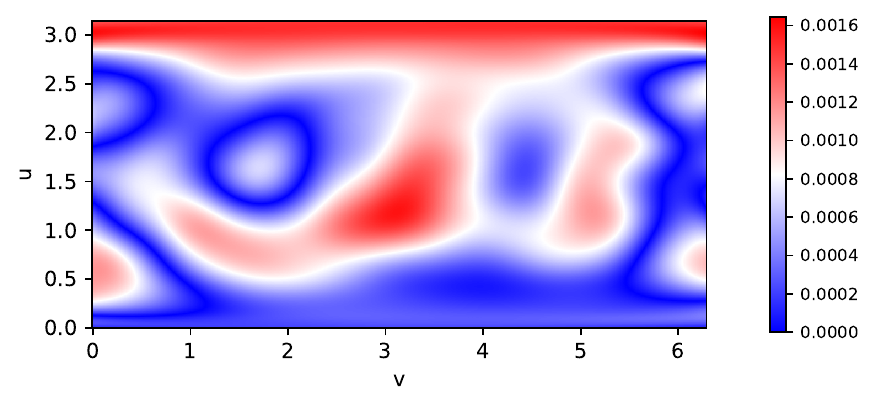} &
		\includegraphics[width=0.28\linewidth]{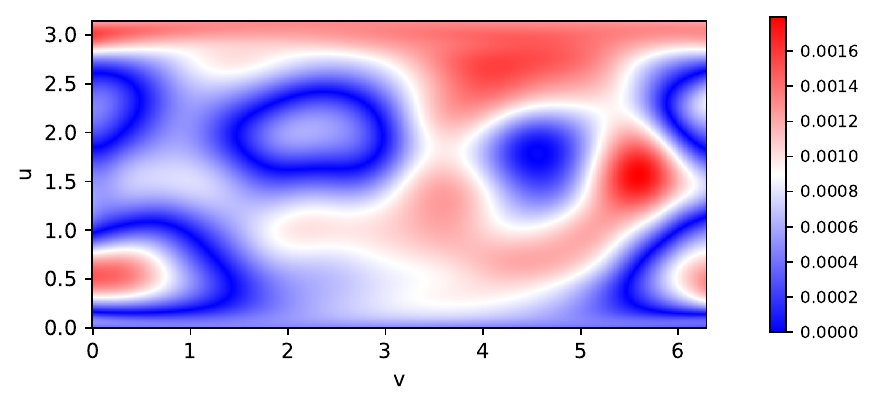} &
		\includegraphics[width=0.28\linewidth]{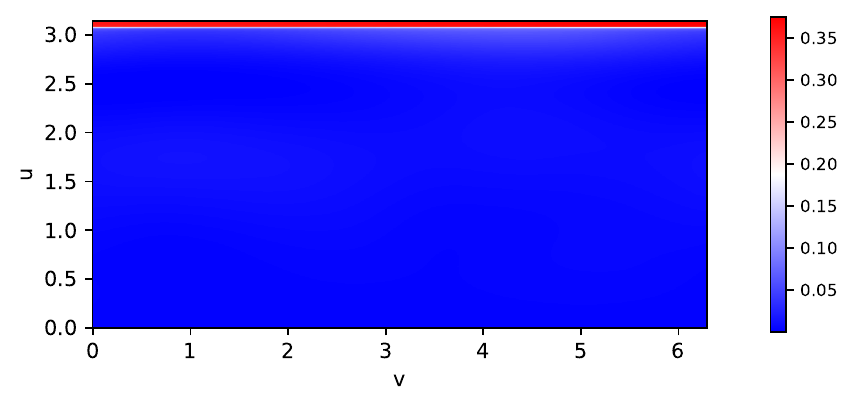} \\[6pt]
		
		\raisebox{1.1cm}[0pt][0pt]{$t = 1.5$} &
		\includegraphics[width=0.28\linewidth]{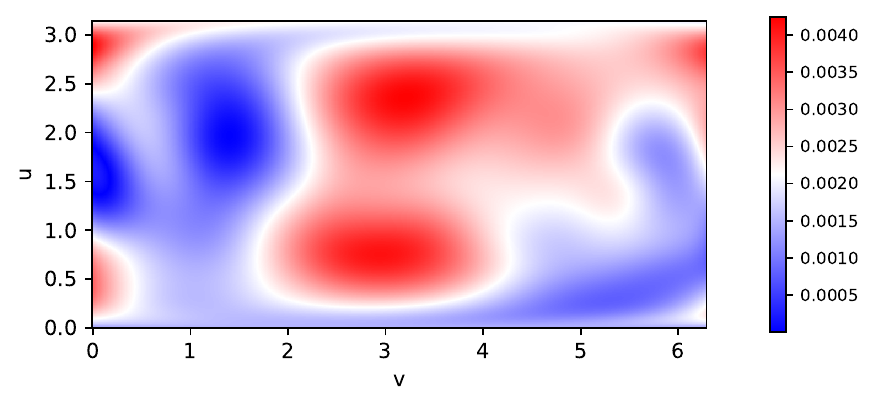} &
		\includegraphics[width=0.28\linewidth]{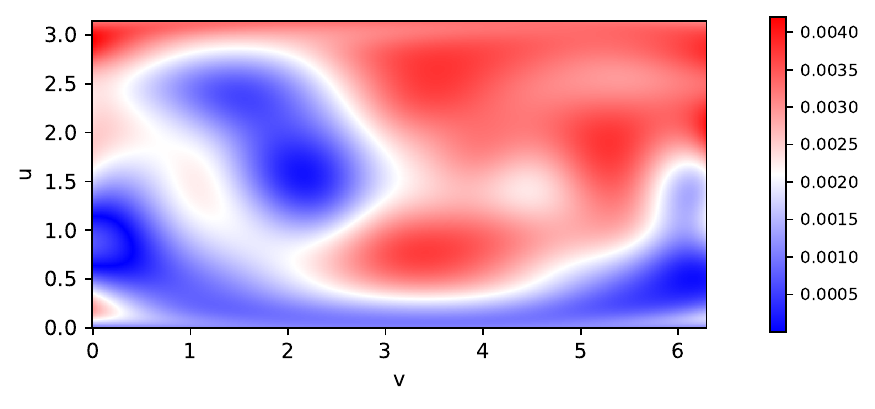} &
		\includegraphics[width=0.28\linewidth]{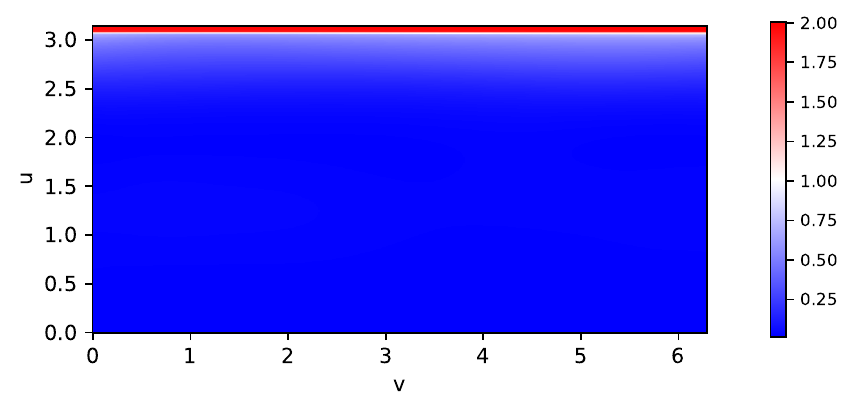} \\
	\end{tabular}
	{\captionsetup{font=small}
		\caption{\textit{The relative point errors with respect to HMCF with initial velocity $r_1=1$ obtained with (a) $\omega_0^s=\omega_b^s=\omega_p^s=1000$, (b) $\omega_0^s=\omega_b^s=\omega_p^s=100$, and (c) $\omega_0^s=\omega_b^s=\omega_p^s=1$}. The relative $\mathbb{L}_2$ errors for $\omega = 1000, 100$, and 1 are 9.74e-4, 1.06e-3 and 3.63e-2, respectively.}
		\label{fig10}}
\end{figure}

We further investigate the effect of the weight parameters $\omega=\omega_0^s=\omega_b^s=\omega_p^s$ in the loss function (\ref{eq315}) on the accuracy of the solution, where $\omega_f^s=1$ is fixed in all simulations. Our training strategy employs a dynamic weight scheme, that is, the initial weights are reduced to $10\%$ during the second 10,000 steps and restored to 1 after 20,000 steps. Here, we consider initial weights of $\omega=1000, 100$, and a baseline case without a dynamic weight scheme, where all weights remain 1 throughout the training. The snapshots of the relative point errors with initial velocity $r_1=1$ at $t=0.0$, $t=1.0$, and $t=1.5$ are shown in Figures \ref{fig10}. For $\omega=1$ (see Figures \ref{fig10}(c)), the largest errors occur near the boundaries $u=\pi$ and $v=2\pi$, where boundary conditions are enforced. These errors decrease as the weights increase to 100 and 1000, respectively. Assigning larger weights to $\mathcal{L}_0^s(\theta_s)$, $\mathcal{L}_b^s(\theta_s)$ and $\mathcal{L}_p^s(\theta_s)$ in the loss function (\ref{eq315}) strongly penalizes deviations from initial and boundary conditions relative to PDE residuals, thereby reducing solution errors. For $r_1=1$, the relative $\mathbb{L}_2$ errors corresponding to $\omega = 1000, 100$, and 1 are 9.74e-4, 1.06e-3 and 3.63e-2, respectively. These results indicate that $\omega = 1000$ yields the most accurate solutions for the underlying  problem. Consequently, we adopt $\omega = 1000$ in all subsequent experiments.

In the following, we aim to apply the hybrid optimization strategy to other nonspherical surfaces. We now consider a surface of ellipsoid given by parametrization
\begin{equation*}
	X_0(u_1, u_2) :=
	\begin{pmatrix}
		a \sin u_1 \cos u_2 \\
		b \sin u_1 \sin u_2 \\
		c \cos u_1
	\end{pmatrix}, \quad u_1 \in [0, \pi], \quad u_2 \in [0, 2\pi],
\end{equation*}
and the initial velocity is
\begin{equation*}
	X_1(u_1, u_2) =-r_1\overrightarrow{N}_0 = \frac{r_1}{|N|}(bc\sin^2 u_1 \cos u_2, ac\sin^2 u_1 \sin u_2,
	ab\sin u_1\cos u_1),
\end{equation*}
where $a=1.5$, $b=1.0$, $c=0.5$, $r_1 \in \mathbb{R}$, $\overrightarrow{N}_0$ the inner normal vector of the initial ellipsoid, and $$N=\sqrt{b^2 c^2 \sin^4 u_1 \cos^2 u_2 + a^2 c^2 \sin^4 u_1 \sin^2 u_2 +
	a^2 b^2 \sin^2 u_1 \cos^2 u_2}\;.$$

We analyze the evolution of ellipsoid under three distinct initial velocities: zero velocity $r_1=0$, a uniform outward velocity $r_1=1$, and a uniform inward velocity $r_1=-1$. The evolution is visualized using front (xz-plane), side (yz-plane), and top (xy-plane) orthographic projections.

\begin{figure}[htpb] 
	\centering 
	\begin{subfigure}[b]{0.95\linewidth} 
		\centering 
		\includegraphics[width=\linewidth]{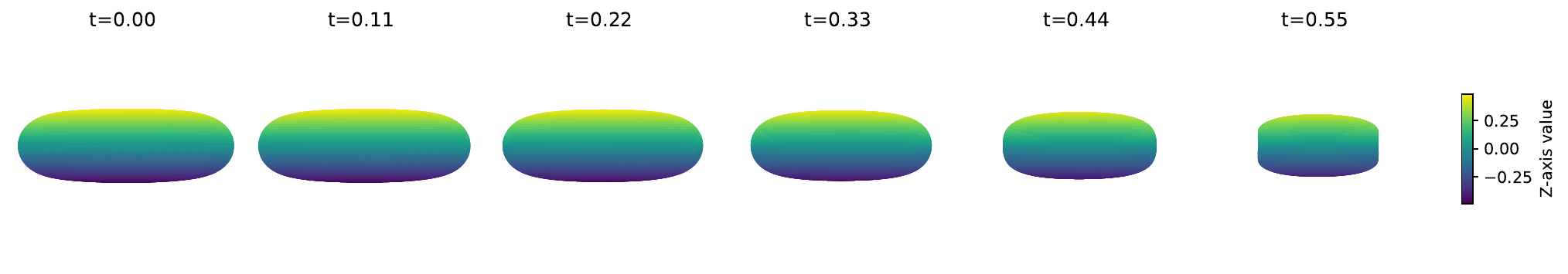}
	\end{subfigure}
	\hfill
	\begin{subfigure}[b]{0.95\linewidth}
		\centering
		\includegraphics[width=\linewidth]{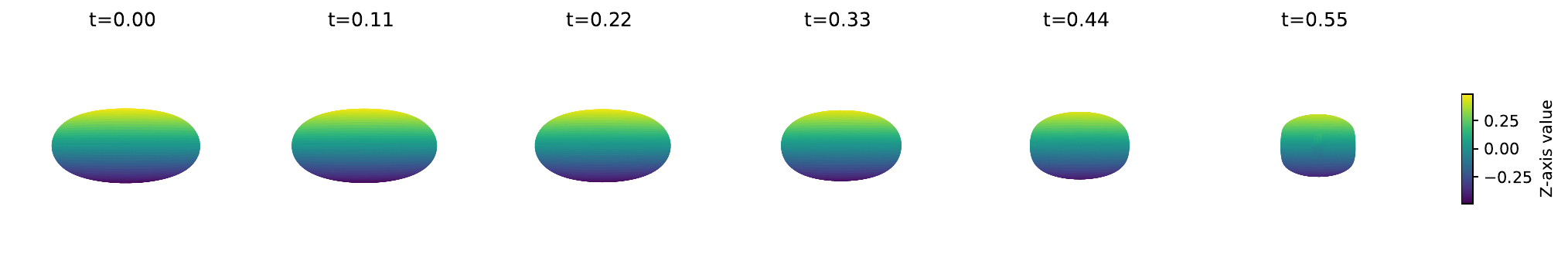} 
	\end{subfigure}
	\hfill
	\begin{subfigure}[b]{0.95\linewidth}
		\centering
		\includegraphics[width=\linewidth]{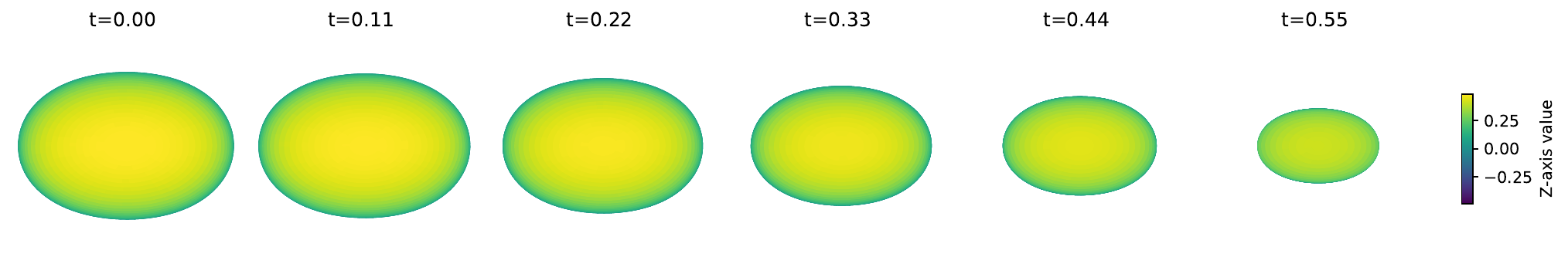} 
	\end{subfigure}

	{\captionsetup{font=small}
		\caption{\textit{HMCF starting from an ellipsoid. We show the evolution for $r_1=0$ in three ways: front view(first row), side view(second row), and top view(third row).}}
		\label{fig11}
	}
\end{figure}

As depicted in Figure \ref{fig11}, when the ellipsoid starts from rest, the initial acceleration is solely governed by the curvature, signifying a purely curvature-driven flow at the onset. The front and side views reveal that the shrinkage is most pronounced along the longest axes (x-axis, a=1.5, and y-axis, b=1.0), where the initial curvature is highest. As the evolution progresses, it seems to approach a shape with four corners. Concurrently, the top view shows the elliptical cross-section shrinking and becoming more circular.

\begin{figure}[htpb] 
	\centering 
	\begin{subfigure}[b]{0.95\linewidth} 
		\centering 
		\includegraphics[width=\linewidth]{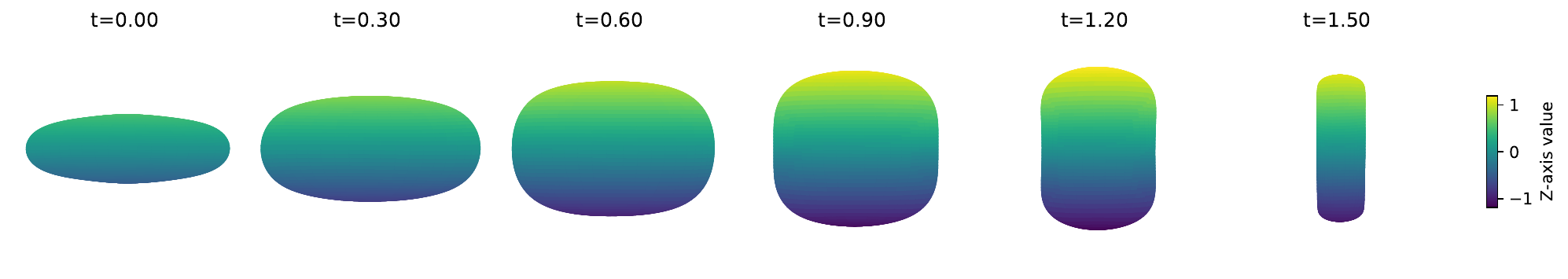}
	\end{subfigure}
	\hfill
	\begin{subfigure}[b]{0.95\linewidth}
		\centering
		\includegraphics[width=\linewidth]{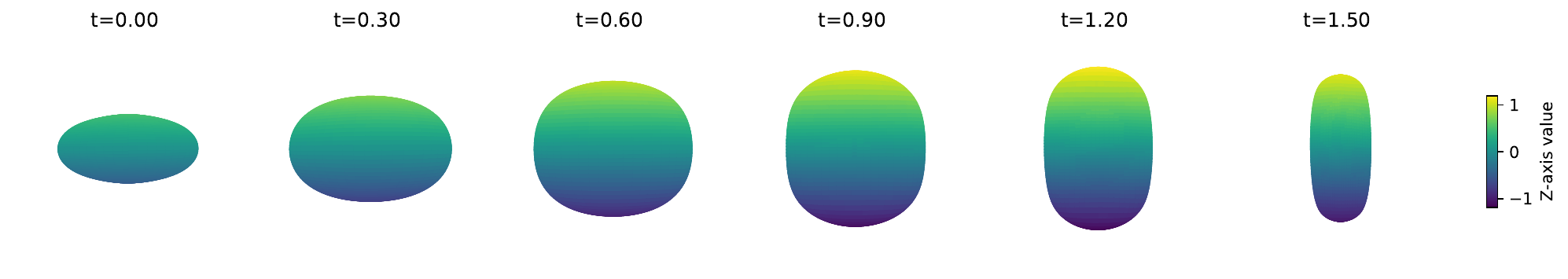} 
	\end{subfigure}
	\hfill
	\begin{subfigure}[b]{0.95\linewidth}
		\centering
		\includegraphics[width=\linewidth]{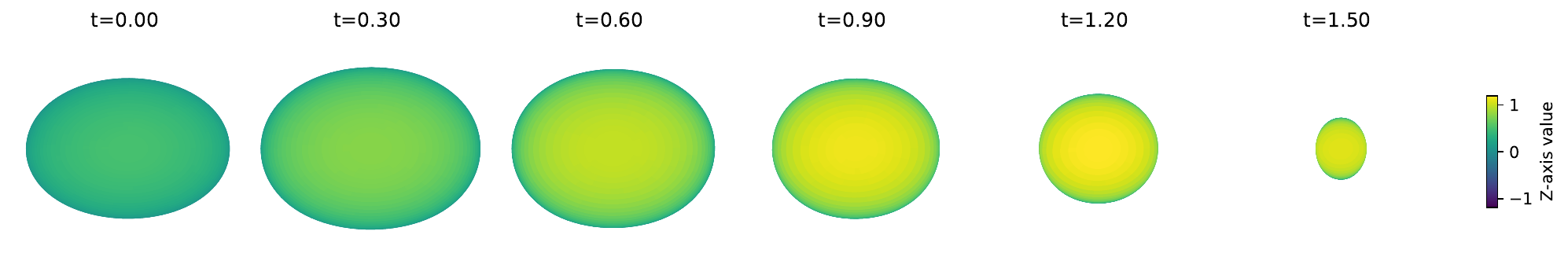} 
	\end{subfigure}

	{\captionsetup{font=small}
		\caption{\textit{HMCF starting from an ellipsoid. We show the evolution for $r_1=1$ in three ways: front view(first row), side view(second row), and top view(third row).}}
		\label{fig12}
	}
\end{figure}

When the initial velocity is outward $r_1=1$, a competition arises between this initial inertia and the intrinsic, inward-pulling curvature force, resulting in a dramatically different evolution, as shown in Figure \ref{fig12}. Initially, the outward velocity is dominant, causing the surface to expand against the intrinsic shrinkage forces. However, the persistent inward acceleration continuously decelerates this outward motion. At a turnaround point, the velocity reverses, and the surface subsequently begins to shrink. Notably, this process is asynchronous. The vertices on the longer axes (x and y), which have the highest curvature, decelerate most rapidly and are the first to reverse course and move inward. In contrast, the vertices on the shortest axis (z), having the lowest curvature, continue to expand for a longer duration before reversing. This asynchrony in motion causes the aspect ratio of the ellipsoid to invert during the evolution, with the ellipsoid thinned and lengthened.

\begin{figure}[htpb] 
	\centering 
	\begin{subfigure}[b]{0.95\linewidth} 
		\centering 
		\includegraphics[width=\linewidth]{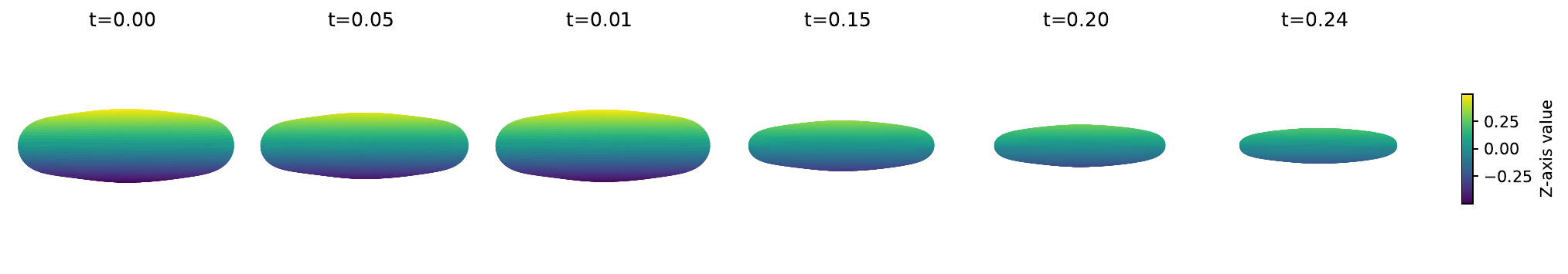}
	\end{subfigure}
	\hfill
	\begin{subfigure}[b]{0.95\linewidth}
		\centering
		\includegraphics[width=\linewidth]{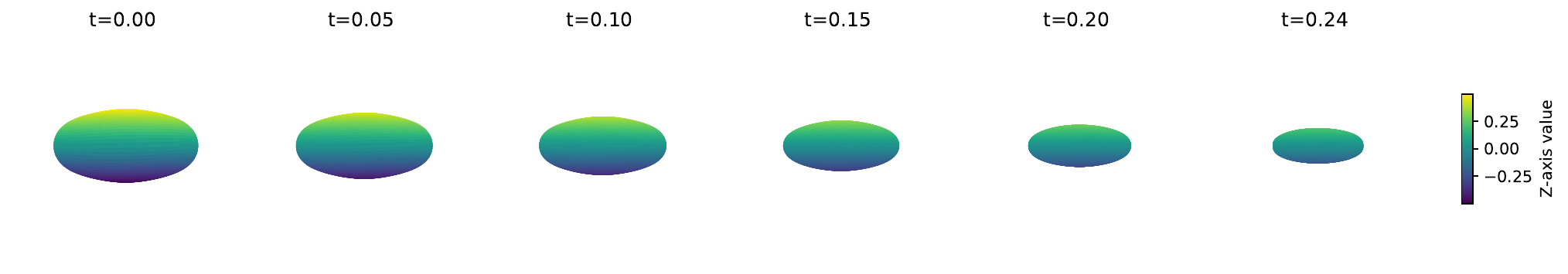} 
	\end{subfigure}
	\hfill
	\begin{subfigure}[b]{0.95\linewidth}
		\centering
		\includegraphics[width=\linewidth]{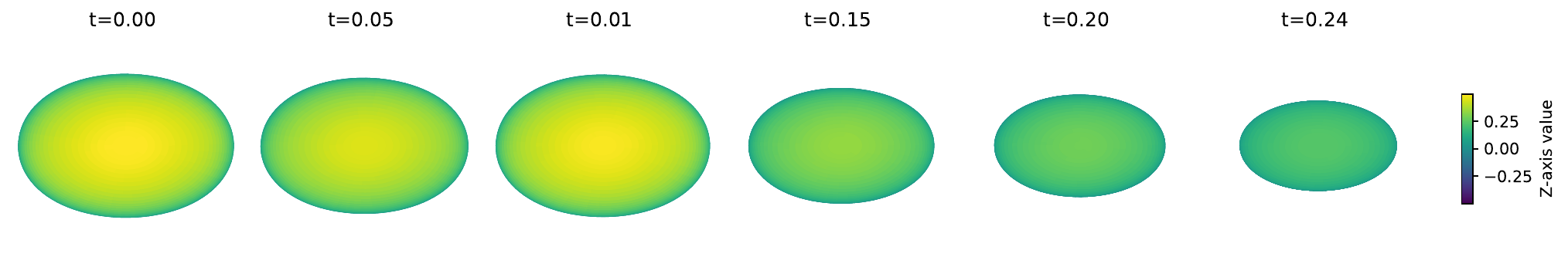} 
	\end{subfigure}

	{\captionsetup{font=small}
		\caption{\textit{HMCF starting from an ellipsoid. We show the evolution for $r_1=-1$ in three ways: front view(first row), side view(second row), and top view(third row).}}
		\label{fig13}
	}
\end{figure}

Repeating the simulation for the initial velocity $r_1=-1$ yields the results in Figure \ref{fig13}. In this case, the initial inward velocity acts in concert with the curvature-induced inward acceleration. This synergistic effect leads to an extremely rapid shrinkage, resulting in the fastest collapse time among the three scenarios. The ellipsoid shrinks quickly without the complex shape transformations observed in the other cases, as both inertia and the intrinsic geometry drive the system unidirectionally towards collapse.

\begin{figure}[htpb] 
	\centering 
	\begin{subfigure}[b]{0.95\linewidth} 
		\centering 
		\includegraphics[width=\linewidth]{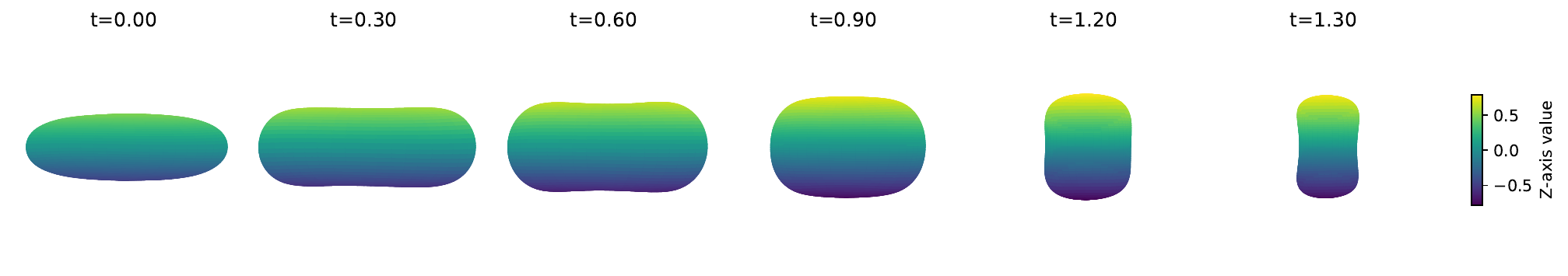}
	\end{subfigure}
	\hfill
	\begin{subfigure}[b]{0.95\linewidth}
		\centering
		\includegraphics[width=\linewidth]{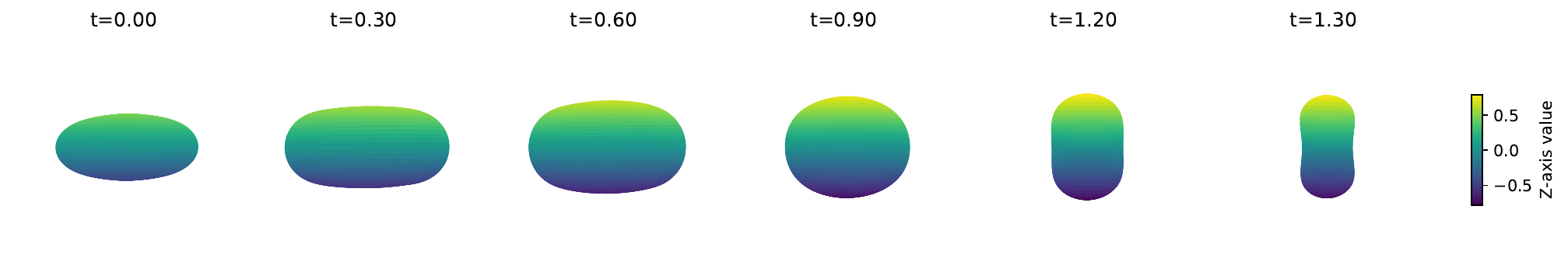} 
	\end{subfigure}
	\hfill
	\begin{subfigure}[b]{0.95\linewidth}
		\centering
		\includegraphics[width=\linewidth]{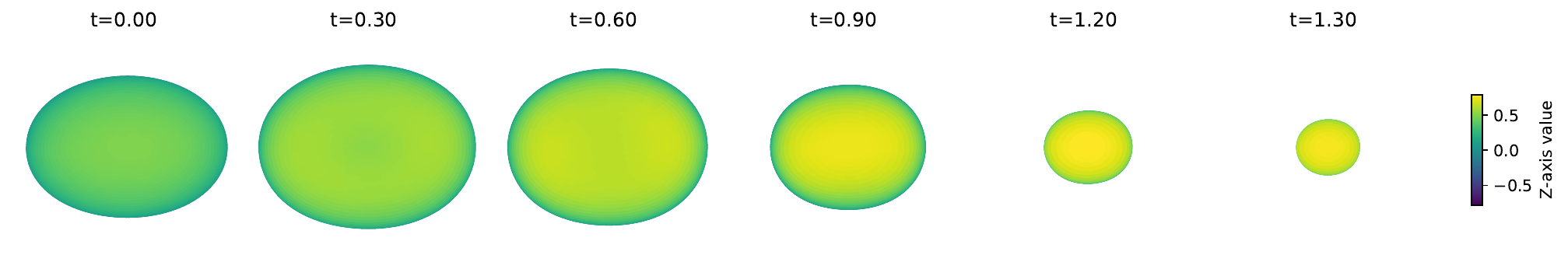} 
	\end{subfigure}

	{\captionsetup{font=small}
		\caption{\textit{HMCF starting from an ellipsoid. We show the evolution for $r_1=sin u$ in three ways: front view(first row), side view(second row), and top view(third row).}}
		\label{fig14}
	}
\end{figure}

As in the discussion for plane curves, non-convexity and asymmetry arise in the evolution of the ellipsoid when the initial velocity is nonconstant. In particular, we choose $r_1=sinu$, where $u\in[0,\pi]$ is the polar angle, the evolution can be seen in Figure \ref{fig14}. This velocity profile imparts the maximum outward velocity at the equator ($u=\frac{\pi}{2}$), which gradually decreases to zero at the poles ($u=0,\pi$). Consequently, the initial phase of the evolution is characterized by a pronounced outward expansion, most prominent along the equatorial region. As the evolution progresses, the non-uniformity of the initial velocity leads to a differential collapse. The equatorial region, having experienced the greatest initial outward displacement, undergoes the most significant inward retraction once the flow reverses. Simultaneously, the polar regions, which have minimal initial motion, are already shrinking. This different motion results in the formation of a distinct concave indentation around the equator, progressively transforming the ellipsoid into a dumbbell-like or biconcave shape as it collapses toward a singularity.

\begin{figure}[htpb] 
	\centering 
	\begin{subfigure}[b]{0.95\linewidth} 
		\centering 
		\includegraphics[width=\linewidth]{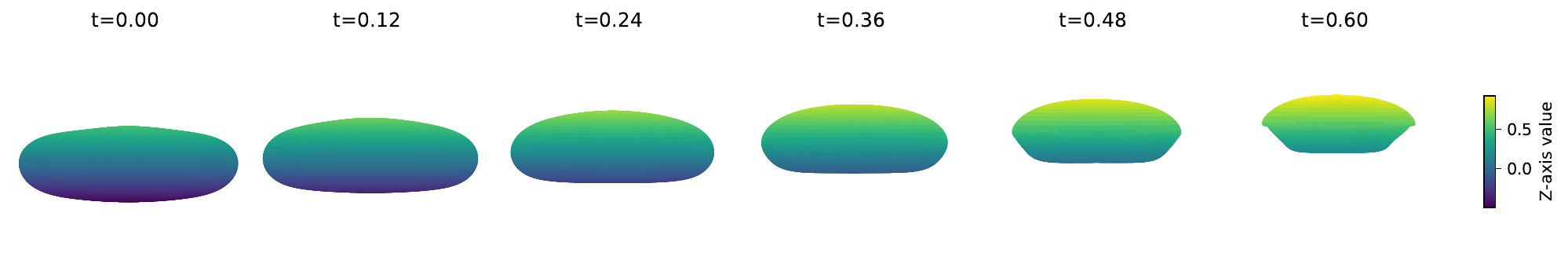}
	\end{subfigure}
	\hfill
	\begin{subfigure}[b]{0.95\linewidth}
		\centering
		\includegraphics[width=\linewidth]{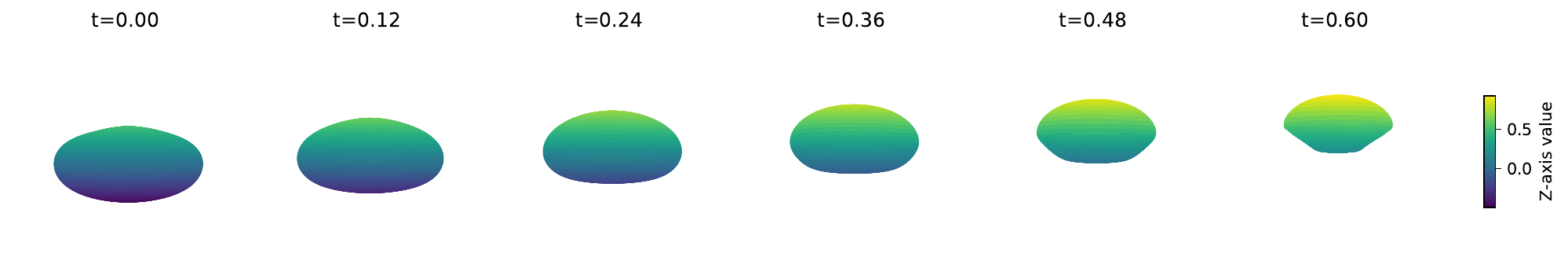} 
	\end{subfigure}
	\hfill
	\begin{subfigure}[b]{0.95\linewidth}
		\centering
		\includegraphics[width=\linewidth]{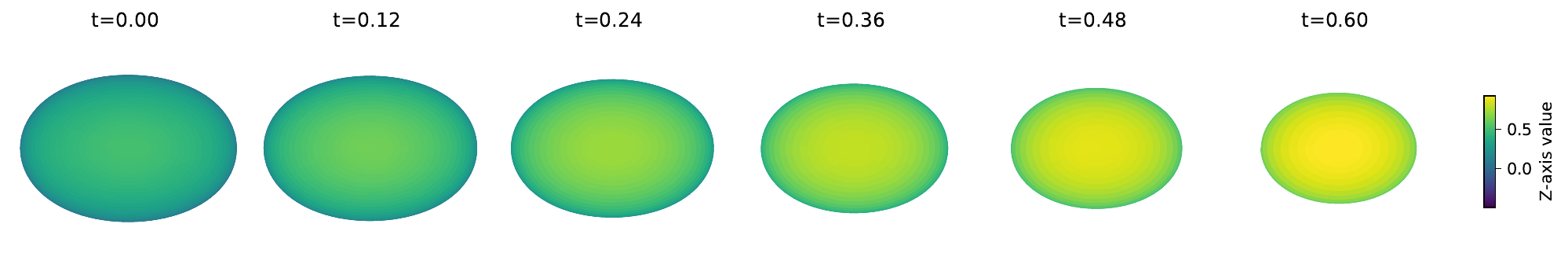} 
	\end{subfigure}

	{\captionsetup{font=small}
		\caption{\textit{HMCF starting from an ellipsoid. We show the evolution for $r_1=cos u$ in three ways: front view(first row), side view(second row), and top view(third row).}}
		\label{fig15}
	}
\end{figure}

The evolution of the ellipsoid under the initial velocity $r_1=cosu$ shows a fascinating asymmetric collapse, driven by a spatially-dependent velocity that changes sign across the equator. The upper half of the ellipsoid, endowed with an initial outward velocity, initially attempts to expand against the inward pull of the mean curvature force. This phase is characterized by a slight inflation and then a slower shrinkage. In contrast, the lower half of the ellipsoid experiences an immediate and highly accelerated shrinkage. The initial inward velocity $r_1<0$ acts synergistically with the intrinsic curvature force, with both contributing to a rapid inward collapse of this section. The combined effect of these opposing initial motions is a pronounced asymmetric evolution, as clearly visible in Figure \ref{fig15}. The southern hemisphere flattens and rapidly migrates upward toward the northern hemisphere. The entire ellipsoid appears to be collapsing upon itself from the bottom up, progressively deforming into a bowl-like shape before eventually collapsing to a singularity.

\begin{figure}[htpb] 
	\centering 
	\begin{subfigure}[b]{0.95\linewidth} 
		\centering 
		\includegraphics[width=\linewidth]{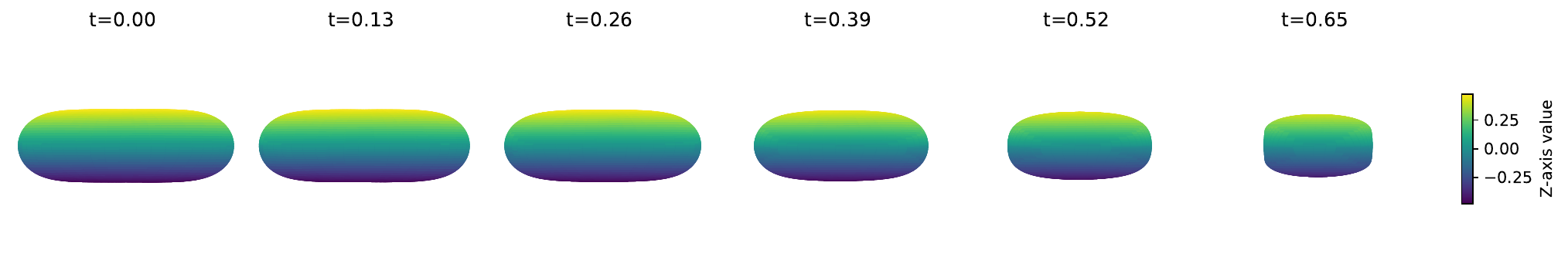}
	\end{subfigure}
	\hfill
	\begin{subfigure}[b]{0.95\linewidth}
		\centering
		\includegraphics[width=\linewidth]{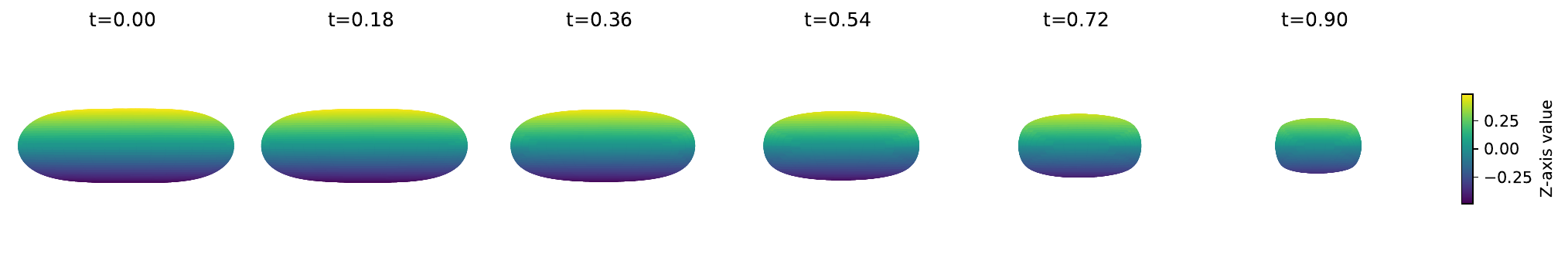} 
	\end{subfigure}
	\hfill
	\begin{subfigure}[b]{0.95\linewidth}
		\centering
		\includegraphics[width=\linewidth]{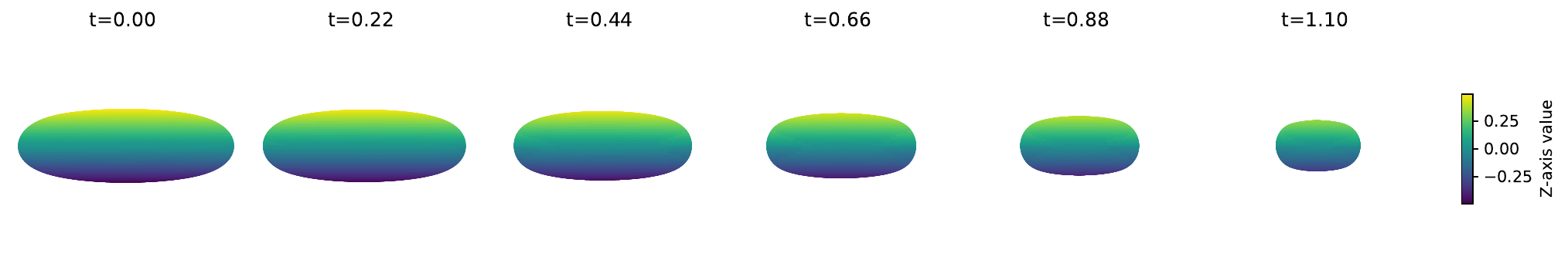} 
	\end{subfigure}

	{\captionsetup{font=small}
		\caption{\textit{HMCF starting from an ellipsoid with initial velocity $r_1=0$. The value of $\beta$ is 1 (first row), 3 (second row), and 5 (third row). Above we visualize from the front view.}}
		\label{fig16}
	}
\end{figure}

\begin{figure}[htpb] 
	\centering 
	\begin{subfigure}[b]{0.95\linewidth} 
		\centering 
		\includegraphics[width=\linewidth]{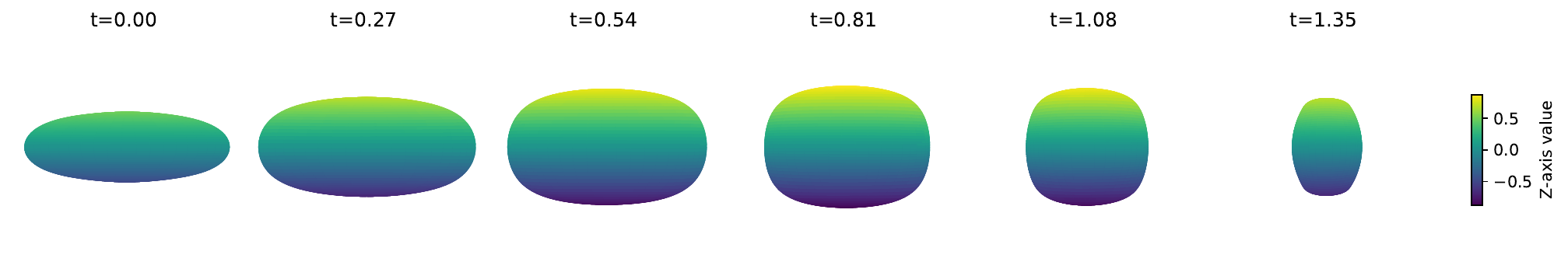}
	\end{subfigure}
	\hfill
	\begin{subfigure}[b]{0.95\linewidth}
		\centering
		\includegraphics[width=\linewidth]{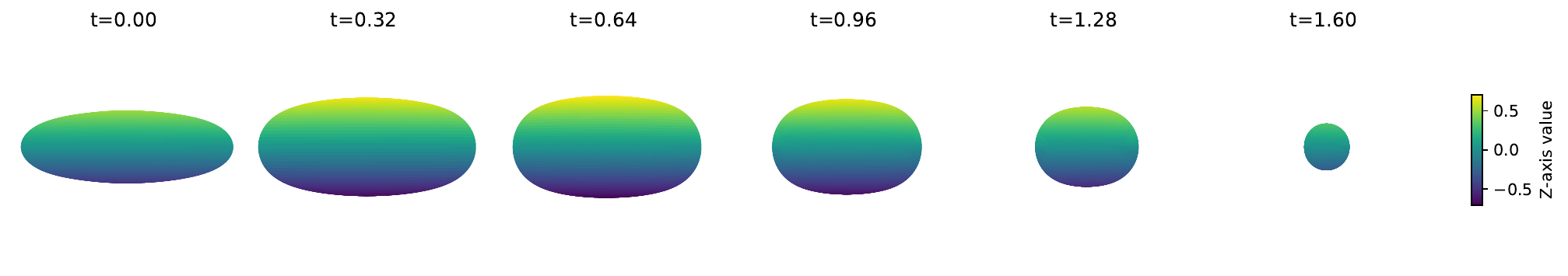} 
	\end{subfigure}
	\hfill
	\begin{subfigure}[b]{0.95\linewidth}
		\centering
		\includegraphics[width=\linewidth]{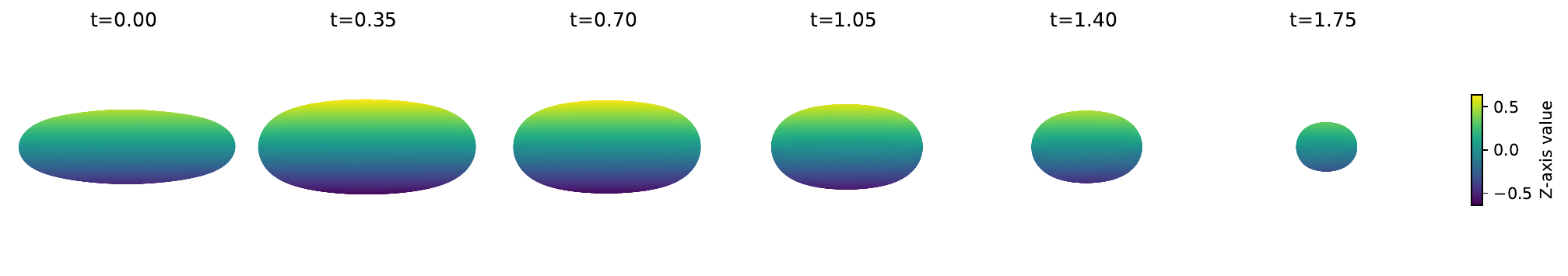} 
	\end{subfigure}

	{\captionsetup{font=small}
		\caption{\textit{HMCF starting from an ellipsoid with initial velocity $r_1=1$. The value of $\beta$ is 1 (first row), 3 (second row), and 5 (third row). Above we visualize from the front view.}}
		\label{fig17}
	}
\end{figure}

\begin{figure}[htpb] 
	\centering 
	\begin{subfigure}[b]{0.95\linewidth} 
		\centering 
		\includegraphics[width=\linewidth]{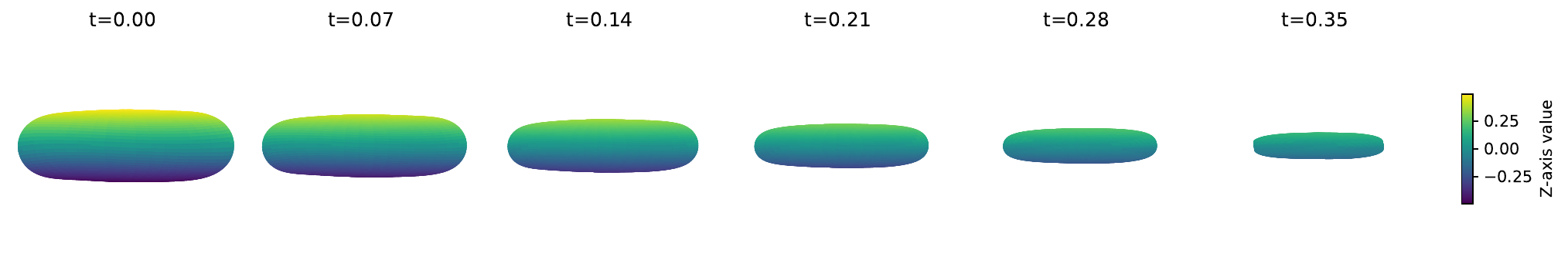}
	\end{subfigure}
	\hfill
	\begin{subfigure}[b]{0.95\linewidth}
		\centering
		\includegraphics[width=\linewidth]{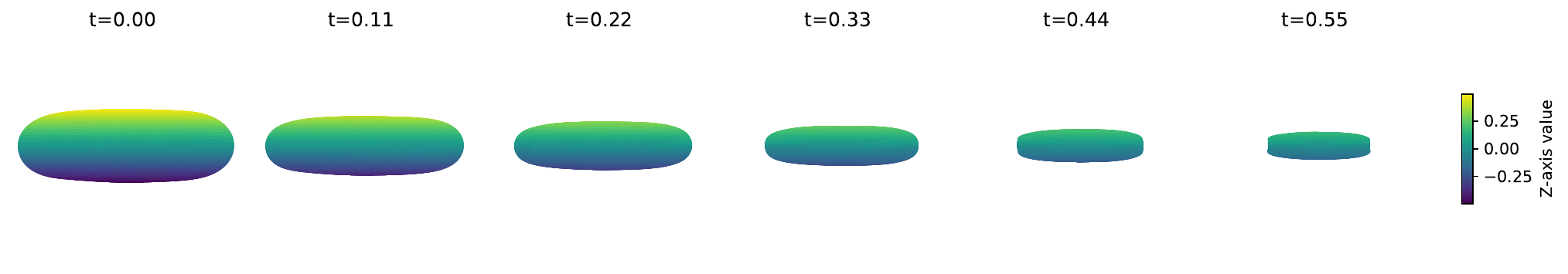} 
	\end{subfigure}
	\hfill
	\begin{subfigure}[b]{0.95\linewidth}
		\centering
		\includegraphics[width=\linewidth]{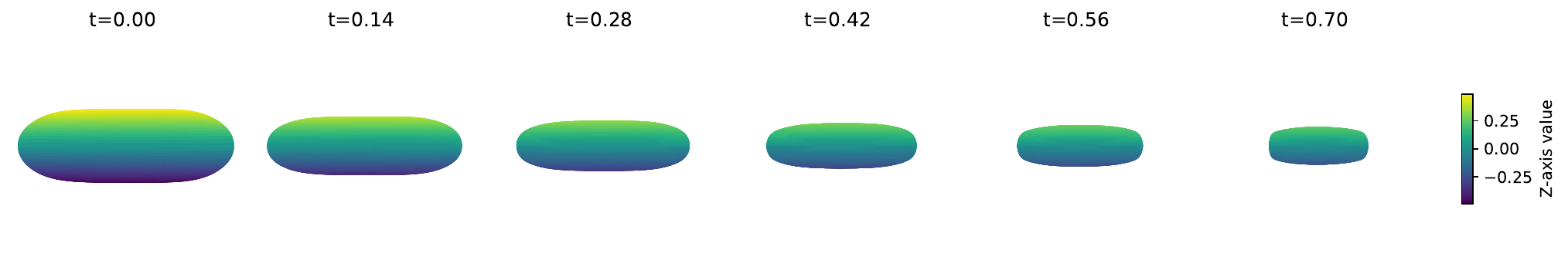} 
	\end{subfigure}

	{\captionsetup{font=small}
		\caption{\textit{HMCF starting from an ellipsoid with initial velocity $r_1=-1$. The value of $\beta$ is 1 (first row), 3 (second row), and 5 (third row). Here we visualize from the front view.}}
		\label{fig18}
	}
\end{figure}

The evolution of the ellipsoid exhibits distinct dynamical behaviors depending on the dissipative coefficient $\beta$. Specifically, we will focus on the values of $\beta$ with 1, 3, and 5. Overall the evolutions shown in Figure \ref{fig16}, Figure \ref{fig17}, and Figure \ref{fig18}. For a surface starting at rest $r_1=0$, the initial acceleration at $t=0$ is identical for all $\beta$ values, being solely determined by the curvature. With a low $\beta$ value, the front view reveals the emergence of a pronounced four-cornered morphology during the evolution, whereas higher $\beta$ values allow the dissipative term to dominate, resulting in a smoother transformation that ultimately makes the surface converge to a circular shape. When an initial outward velocity $r_1=1$ is applied, a low $\beta$ permits significant expansion, accompanied by a reversal of the major and minor axes. Increasing $\beta$ mitigates these distortions, progressively guiding the ellipsoid toward a spherical geometry. Conversely, with an initial inward velocity $r_1=-1$, the inward motion and curvature force act synergistically. In this case, $\beta$ serves as a pure retarding force, where a low $\beta$ leads to an extremely rapid and direct collapse, while a higher $\beta$ slows down the evolution, enabling a more gradual and controlled shrinkage toward isotropy.

Finally, we study the performance of the torus given by parametrization
\begin{equation*}
	X_0(u_1, u_2) :=
	\begin{pmatrix}
		(R + r\cos u_2)\cos u_1\\
		(R + r\cos u_2)\sin u_1 \\
		r\sin u_2
	\end{pmatrix}, \quad u_1 \in [0, 2\pi], \quad u_2 \in [0, 2\pi],
\end{equation*}
and the initial velocity is
\begin{equation*}
	X_1(u_1, u_2) =-r_1\overrightarrow{N}_0 = r_1(\cos u_1 \cos u_2, \sin u_1 \cos u_2, \sin u_2),
\end{equation*}
where $R=2$, $r=1$, $r_1 \in \mathbb{R}$, $\overrightarrow{N}_0$ the inner normal vector of the initial torus.

In contrast to the previous treatment of spheres and ellipsoids, ensuring closure and smoothness for a torus necessitates only
\begin{equation*}
	X(0,u_2,t) = X(2\pi,u_2,t), \; X_{u_1}(0,u_2,t) = X_{u_1}(2\pi,u_2,t),
\end{equation*}
and
\begin{equation*}
	X(u_1,0,t) = X(u_1,2\pi,t), \; X_{u_2}(u_1,0,t) = X_{u_2}(u_1,2\pi,t).
\end{equation*}
As a consequence, the residuals of (\ref{eq311}) and (\ref{eq312}) can be expressed as
\begin{equation*}
	\mathcal{L}_{b1}^s(\theta_s) = \frac{1}{N_b}\sum_{i=1}^{N_b}(|X^i(u_1^i,0,t_b^i)-X^i(u_1^i,2\pi,t_b^i)|^2 + |X_{u_1}^i(u_1^i,0,t_b^i)-X_{u_1}^i(u_1^i,2\pi,t_b^i)|^2),
\end{equation*}
and
\begin{equation*}
	\mathcal{L}_{b2}^s(\theta_s) = \frac{1}{N_b}\sum_{i=1}^{N_b}(|X^i(u_1^i,0,t_b^i)-X^i(u_1^i,2\pi,t_b^i)|^2 + |X_{u_2}^i(u_1^i,0,t_b^i)-X_{u_2}^i(u_1^i,2\pi,t_b^i)|^2).
\end{equation*}

\begin{figure}[htpb] 
	\centering 
	\begin{subfigure}[b]{0.95\linewidth} 
		\centering 
		\includegraphics[width=\linewidth]{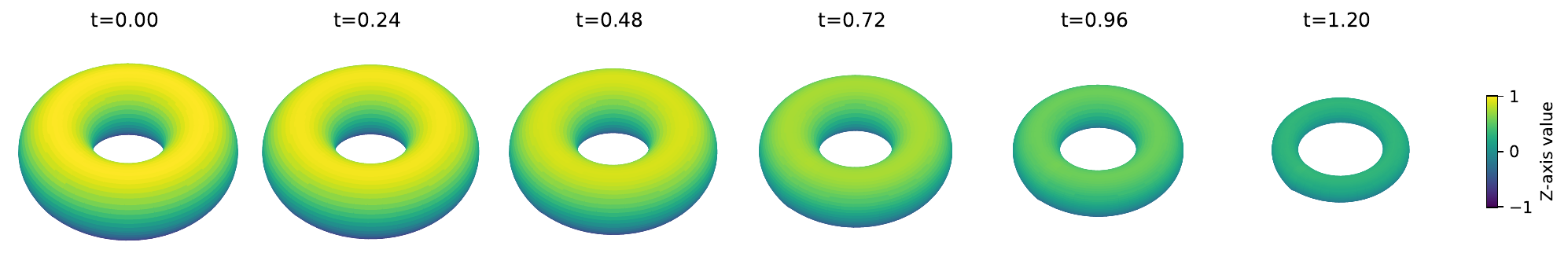}
	\end{subfigure}
	\hfill
	\begin{subfigure}[b]{0.95\linewidth}
		\centering
		\includegraphics[width=\linewidth]{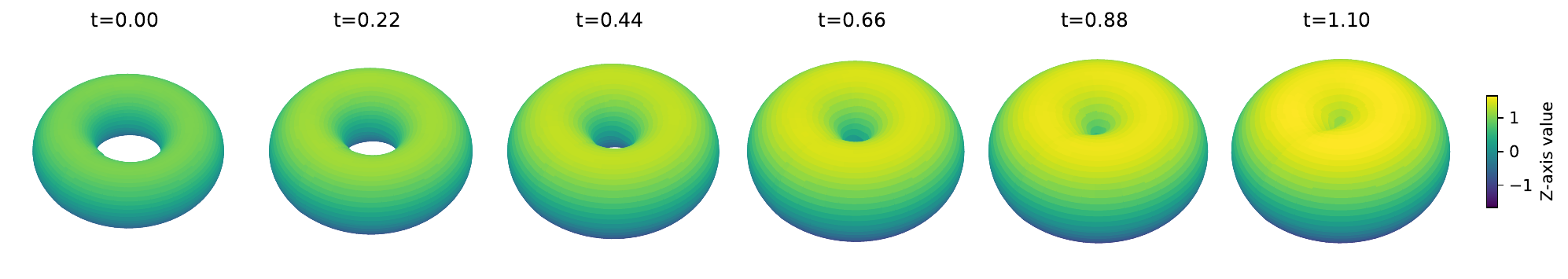} 
	\end{subfigure}
	\hfill
	\begin{subfigure}[b]{0.95\linewidth}
		\centering
		\includegraphics[width=\linewidth]{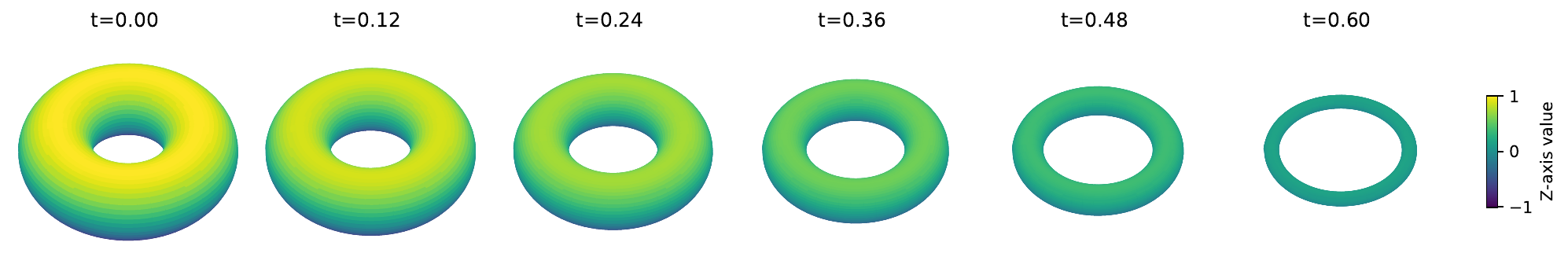} 
	\end{subfigure}
	
	{\captionsetup{font=small}
		\caption{\textit{HMCF starting from a torus. We show the evolution for $r_1=0$, $r_1=1$, and $r_1=-1$ (from top to  bottom).}}
		\label{fig19}
	}
\end{figure}

We follow the same training scheme as for spheres and ellipsoids. The evolution of torus is visualized in Figure \ref{fig19}. When the initial velocity $r_1=0$, the evolution is driven purely by geometry, and the opposing curvatures on the inner and outer equators cause the central hole to narrow as the torus shrinks. With an initial outward velocity $r_1=1$, the torus first undergoes a phase of inflation, during which the central hole enlarges as inertia counteracts the inward pull of curvature. As the flow evolves, the persistent curvature force eventually overcomes this momentum and initiates a contractile phase, where the inner hole closes, leading to a change in topology and a shape that approaches a nearly sphere. For $r_1=-1$, the torus collapses extremely rapidly. The minor radius $r$ shrinks much faster than its major radius $R$, resulting in a pronounced thinning of the tube and an apparent enlargement of the central hole.

\section{Conclusions}
In this paper, we propose a physics-informed neural networks method to solve the evolution of plane curves and surfaces under the HMCF. For plane curve simulations, we employ the Adam optimizer with varying learning rates to enhance accuracy. For surface simulations, a dynamic weighting scheme is adopted to achieve more precise solutions. As a mesh-free method, our approach eliminates the need for discretization and meshing of the computational domain, demonstrating high efficiency in simulations involving high-dimensional problems. To evaluate the performance of our method, we conduct extensive  experiments on various initial curves and surfaces, as well as different initial velocity conditions, including constant and non-constant cases. Furthermore, we investigate the influence of the dissipative coefficient $\beta$ on the evolution of curves and surfaces. The results reveal that, as $\beta$ increases, the system exhibits a transition from hyperbolic to parabolic behavior in the evolution of curves and surfaces under the HMCF. In addition, the dynamics of curves and surfaces under HMCF have been depicted in detail. To the best of our knowledge, this is the first result on solving HMCF by using PINNs. This can also advertise more researchers to investigate the geometric flows with the aid of PINNs.

\section*{Data availability}
No data was used for the research described in the article.

\section*{Acknowledgements} 
This work was supported by Zhejiang Normal University (Grant Nos.YS304222929, YS304222977) and National Natural Science Foundation of China (Grant Nos. 12090020, 12090025).

\end{document}